\documentclass[a4paper,UKenglish,cleveref, autoref, thm-restate]{lipics-v2021}

\hideLIPIcs  

\pdfoutput=1
\usepackage[utf8]{inputenc}
\usepackage{amsmath}
\usepackage{hyperref}
\usepackage[numbers,sort&compress]{natbib}
\usepackage{stmaryrd} 
\usepackage{mathtools}
\usepackage{subcaption}
\usepackage{cleveref}
\usepackage{aligned-overset}
\usepackage{stackengine}

\usepackage{tikzit}
\input{zx.tikzdefs}

\tikzstyle{box}=[shape=rectangle, text height=1.5ex, text depth=0.25ex, yshift=0.5mm, fill=white, draw=black, minimum height=5mm, yshift=-0.5mm, minimum width=5mm, font={\small}]
\tikzstyle{Z dot}=[inner sep=0mm, minimum size=2mm, shape=circle, draw=black, fill={rgb,255: red,221; green,255; blue,221}]
\tikzstyle{Z phase dot}=[minimum size=1.2em, font={\footnotesize\boldmath}, shape=rectangle, rounded corners=0.5em, inner sep=0.2em, outer sep=-0.2em, scale=0.8, tikzit shape=circle, draw=black, fill={rgb,255: red,221; green,255; blue,221}, tikzit draw=blue]
\tikzstyle{X dot}=[Z dot, shape=circle, draw=black, fill={rgb,255: red,255; green,136; blue,136}]
\tikzstyle{X phase dot}=[Z phase dot, tikzit shape=circle, tikzit draw=blue, fill={rgb,255: red,255; green,136; blue,136}, font={\footnotesize\boldmath}]
\tikzstyle{hadamard}=[fill=yellow, draw=black, shape=rectangle, inner sep=0.6mm, minimum height=1.5mm, minimum width=1.5mm]
\tikzstyle{vertex}=[inner sep=0mm, minimum size=1mm, shape=circle, draw=black, fill=black]
\tikzstyle{vertex set}=[inner sep=0mm, minimum size=1mm, shape=circle, draw=black, fill=white, font={\footnotesize\boldmath}]
\tikzstyle{target}=[inner sep=0mm, minimum size=3mm, shape=circle, draw=black]
\tikzstyle{white dot}=[inner sep=0mm, minimum size=2mm, shape=circle, draw=black, fill=white]
\tikzstyle{bn}=[inner sep=0mm, minimum size=2mm, shape=circle, draw=black, fill=black]
\tikzstyle{black dot}=[bn, tikzit fill=black, tikzit draw=black]
\tikzstyle{XY}=[bn, tikzit fill=black, tikzit draw=black]
\tikzstyle{YZ}=[fill=black, draw=black, shape=rectangle, minimum size=2mm, inner sep=0mm]
\tikzstyle{XZ}=[fill=black, draw=black, shape=triangle, minimum size=2mm, inner sep=0mm, tikzit shape=rectangle, tikzit fill={rgb,255: red,75; green,227; blue,9}, tikzit draw={rgb,255: red,240; green,26; blue,18}]
\tikzstyle{Pauli dot}=[inner sep=0mm, minimum size=2mm, shape=rectangle, draw=white, fill=white, tikzit draw={rgb,255: red,5; green,9; blue,255}, font={\tiny}]
\tikzstyle{input}=[shape=rectangle, fill=white, draw=black, minimum size=2mm, label={center:$\bullet$}]
\tikzstyle{output}=[inner sep=0mm, minimum size=2mm, shape=circle, draw=black, fill=white]

\tikzstyle{hadamard edge}=[-, dashed, dash pattern=on 2pt off 1.5pt, thick, draw={rgb,255: red,68; green,136; blue,255}]
\tikzstyle{brace edge}=[-, tikzit draw=blue, decorate, decoration={brace,amplitude=1mm,raise=-1mm}]
\tikzstyle{diredge}=[->]
\tikzstyle{dashed edge}=[-, dashed, dash pattern=on 2pt off 0.5pt, draw=black]

\tikzstyle{diamant}=[diamond, fill=couleurdefond, draw=black,inner sep=0.1em]
\tikzstyle{newe}=[rectangle, fill={gray!15}, draw=black, tikzit shape=rectangle, inner sep=0.2em]
\tikzstyle{cercle}=[circle, fill=couleurdefond, draw=black]
\tikzstyle{scercle}=[circle, fill=couleurdefond, draw=black, tikzit fill=white, inner sep=0.1em]
\tikzstyle{cartouche}=[rounded rectangle, fill=couleurdefond, draw=black,inner sep=0.2em]
\tikzstyle{neg}=[rounded rectangle, fill=couleurdefond, draw=black, execute at end node={$\neg$}]
\tikzstyle{sneg}=[rounded rectangle, fill=couleurdefond, draw=black, execute at end node={$\neg$}, scale=0.8]
\tikzstyle{negserie}=[rounded rectangle, fill=couleurdefond, draw=black, execute at end node={\footnotesize$\star\star$}]
\tikzstyle{diagrammevide}=[rectangle, fill=couleurdefond, draw=black, inner sep=1.25em, borddiagrammevide, tikzit shape=rectangle]
\tikzstyle{mdiagrammevide}=[rectangle, fill=couleurdefond, draw=black, inner sep=0.75em, sborddiagrammevide, tikzit shape=rectangle]
\tikzstyle{msdiagrammevide}=[rectangle, fill=couleurdefond, draw=black, inner sep=0.7em, msborddiagrammevide, tikzit shape=rectangle]
\tikzstyle{sdiagrammevide}=[rectangle, fill=couleurdefond, draw=black, inner sep=0.5em, sborddiagrammevide, tikzit shape=rectangle]
\tikzstyle{xsdiagrammevide}=[rectangle, fill=couleurdefond, draw=black, inner sep=0.4em, xsborddiagrammevide, tikzit shape=rectangle]
\tikzstyle{bs}=[shape=beam, fill=couleurdefond, draw, inner sep=0.25em, thick, tikzit fill=white]
\tikzstyle{sbs}=[shape=beam, fill=couleurdefond, draw, inner sep=0.2em, thick, tikzit fill=white]
\tikzstyle{npbs}=[shape=beam, horizontal fill={{npbsmoitiebasse}{npbsmoitiehaute}}, draw, inner sep=0.25em, thick, tikzit fill={rgb,255: red,128; green,128; blue,128}]
\tikzstyle{npbsalenvers}=[shape=beam, horizontal fill={{npbsmoitiehaute}{npbsmoitiebasse}}, draw, inner sep=0.25em, thick, tikzit fill={rgb,255: red,128; green,128; blue,128}]
\tikzstyle{snpbs}=[shape=beam, horizontal fill={{npbsmoitiebasse}{npbsmoitiehaute}}, draw, inner sep=0.2em, thick, tikzit fill={rgb,255: red,128; green,128; blue,128}]
\tikzstyle{snpbsalenvers}=[shape=beam, horizontal fill={{npbsmoitiehaute}{npbsmoitiebasse}}, draw, inner sep=0.2em, thick, tikzit fill={rgb,255: red,128; green,128; blue,128}]
\tikzstyle{cnot}=[shape=circle, draw, path picture={ 
\draw[black](path picture bounding box.north) -- (path picture bounding box.south) (path picture bounding box.west) -- (path picture bounding box.east);
}, tikzit fill={rgb,255: red,223; green,223; blue,223}]
\tikzstyle{thickcnot}=[shape=circle, draw, thick, path picture={ 
\draw[thick,black](path picture bounding box.north) -- (path picture bounding box.south) (path picture bounding box.west) -- (path picture bounding box.east);
}, tikzit fill={rgb,255: red,223; green,223; blue,223}]
\tikzstyle{boite22}=[fill=white, draw=black, shape=rectangle, minimum height=1cm, minimum width=0.5cm]
\tikzstyle{boite15}=[fill=white, draw=black, shape=rectangle, minimum height=0.7cm, minimum width=0.5cm]
\tikzstyle{boite2}=[fill=white, draw=black, shape=rectangle, minimum height=0cm, minimum width=0cm]
\tikzstyle{snegpotentiel}=[fill=couleurdefond, draw=black, shape=rounded rectangle, inner sep=0.25em, tikzit fill={rgb,255: red,191; green,191; blue,191}, execute at end node={\footnotesize$\star$}]
\tikzstyle{negpotentiel}=[fill=couleurdefond, draw=black, shape=rounded rectangle, tikzit fill={rgb,255: red,191; green,191; blue,191}, execute at end node={$\star$}]
\tikzstyle{token}=[fill=black, draw=black, shape=circle, inner sep=0.1em]
\tikzstyle{whitetoken}=[fill=white, draw=black, shape=circle, inner sep=0.1em]
\tikzstyle{boitePBS}=[fill=white, draw=gray, thick, shape=rectangle, rounded corners=3pt, minimum height=0.6cm, inner sep=0.1em, minimum width=0.5cm]
\tikzstyle{boitePBS2}=[fill=white, draw=gray, thick, shape=rectangle, rounded corners=3pt, minimum height=0.55cm, inner sep=0.1em, minimum width=0.5cm]
\tikzstyle{sgene}=[fill={gray!30}, draw=black, shape=rounded rectangle, rounded rectangle east arc=0pt, minimum height=0.5cm, inner sep=0em, minimum width=0cm, scale=0.8]
\tikzstyle{sdetector}=[fill={gray!30}, draw=black, shape=rounded rectangle, rounded rectangle west arc=0pt, minimum height=0.5cm, inner sep=0em, minimum width=0cm, scale=0.8]
\tikzstyle{xsgene}=[fill={gray!30}, draw=black, shape=rounded rectangle, rounded rectangle east arc=0pt, minimum height=0.5cm, inner sep=0em, minimum width=0cm, scale=0.67]
\tikzstyle{xsdetector}=[fill={gray!30}, draw=black, shape=rounded rectangle, rounded rectangle west arc=0pt, minimum height=0.5cm, inner sep=0em, minimum width=0cm, scale=0.67]
\tikzstyle{PolRot}=[fill={gray!30}, draw=black, shape=rectangle, minimum height=0.5cm, inner sep=0.1em, minimum width=0.1cm]
\tikzstyle{PhS}=[fill=white, draw=black, shape=rectangle, minimum height=0.5cm, inner sep=0.1em, minimum width=0.1cm]
\tikzstyle{gene}=[fill={gray!30}, draw=black, shape=rounded rectangle, rounded rectangle east arc=0pt, minimum height=0.5cm, inner sep=0em, minimum width=0cm]
\tikzstyle{detector}=[fill={gray!30}, draw=black, shape=rounded rectangle, rounded rectangle west arc=0pt, minimum height=0.5cm, inner sep=0em, minimum width=0cm]
\tikzstyle{cartoucherouge}=[rounded rectangle, fill={red!55!white}, draw=black, tikzit fill=red]
\tikzstyle{cartouchebleu}=[rounded rectangle, fill={blue!33!white}, draw=black, tikzit fill=blue]
\tikzstyle{diamantrouge}=[diamond, fill={rgb,255: red,255; green,115; blue,115}, draw=black]
\tikzstyle{diamantbleu}=[diamond, fill={rgb,255: red,171; green,171; blue,255}, draw=black]
\tikzstyle{control}=[fill=black, draw=black, shape=circle, scale=0.35]
\tikzstyle{wcontrol}=[fill=white, draw=black, shape=circle, scale=0.35]
\tikzstyle{boite3qubits}=[fill=white, draw=black, shape=rectangle, minimum height=2.5cm]
\tikzstyle{boite2qubits}=[fill=white, draw=black, shape=rectangle, minimum height=1.75cm]

\tikzstyle{new}=[-, tikzit draw=magenta]
\tikzstyle{tirets}=[-, draw=black, dashed]
\tikzstyle{noire}=[-, draw=black, tikzit draw=magenta]
\tikzstyle{ep}=[-, draw=black, tikzit draw=magenta]
\tikzstyle{longdashed}=[-, dash pattern=on 5pt off 5pt]
\tikzstyle{pointilles}=[-, draw=black, dotted]
\tikzstyle{trait}=[-, draw=black, thick,dashed]
\tikzstyle{boxed}=[-, draw=gray, thick,dashed]
\tikzstyle{grise}=[-, draw={rgb,255: red,191; green,191; blue,191}]
\tikzstyle{rouge}=[-, draw=red]
\tikzstyle{bleue}=[-, draw=bleu, tikzit draw=blue]
\tikzstyle{verte}=[-, draw={rgb,255: red,0; green,230; blue,0}]
\tikzstyle{borddiagrammevide}=[-, dash pattern=on 0.5em off 0.5em on 0.5em off 0.5em on 0.5em off 0em]
\tikzstyle{msborddiagrammevide}=[-, dash pattern=on 0.28em off 0.28em on 0.28em off 0.28em on 0.28em off 0em]
\tikzstyle{sborddiagrammevide}=[-, dash pattern=on 0.2em off 0.2em on 0.2em off 0.2em on 0.2em off 0em]
\tikzstyle{xsborddiagrammevide}=[-, dash pattern=on 0.16em off 0.16em on 0.16em off 0.16em on 0.16em off 0em]
\tikzstyle{mediumdash}=[-, dash pattern=on 2pt off 2pt]
\tikzstyle{rougefonce}=[-, draw={red!50!black}, tikzit draw={rgb,255: red,136; green,0; blue,0}]

\input{figures/styles-pbs.tikzdefs}

\tikzstyle{gate}=[fill=white, draw=black, shape=rectangle, minimum height=0.5cm, minimum width=0.1cm, inner sep=0.1em]
\tikzstyle{control}=[fill=black, draw=black, shape=circle, scale=0.35]
\tikzstyle{not}=[shape=circle, path picture={ 
\draw[black](path picture bounding box.north) -- (path picture bounding box.south) (path picture bounding box.west) -- (path picture bounding box.east);
}, draw=black]
\tikzstyle{wcontrol}=[fill=white, draw=black, shape=circle, scale=0.35]
\tikzstyle{empty}=[fill=white, draw=black, shape=rectangle, inner sep=0.4em, emptyborder]
\tikzstyle{globalphase}=[fill=white, draw=black, inner sep=0.15em, shape=rounded rectangle]
\tikzstyle{ancilla}=[fill=black, draw=black, shape=rectangle, minimum width=0.01cm, minimum height=0.25cm, inner sep=0.01em]
\tikzstyle{ground}=[fill=white, path picture={\draw[black](-1.5mm,0)--(-0.6mm,0);\draw[black,thick](-0.6mm,-1.75mm)--(-0.6mm,1.75mm) (0mm,-0.9mm)--(0mm,0.9mm) (0.6mm,-0.5mm)--(0.6mm,0.5mm);}, minimum width=0.1mm, draw=none, outer sep=0pt]
\tikzstyle{gate22}=[fill=white, draw=black, shape=rectangle, minimum height=0.7cm, minimum width=0.5cm]
\tikzstyle{void}=[shape=rectangle, minimum height=0.5cm]
\tikzstyle{circuit}=[fill=white, draw=black, shape=rectangle]

\tikzstyle{emptyborder}=[-, dash pattern=on 0.16em off 0.16em on 0.16em off 0.16em on 0.16em off 0em]
\tikzstyle{etc}=[-, draw=black, dashed, thick]
\tikzstyle{dots}=[-, dotted, draw=black]
\tikzstyle{shortdashed}=[-, draw=black, dash pattern=on 1pt off 1pt]
\tikzstyle{shortdashed2}=[-, draw=black, dash pattern=on 2pt off 2pt]

\newcommand{\tikzfigM}[1]{\scalebox{0.85}{\input{#1.tikz}}}
\newcommand{\tikzfigS}[1]{\scalebox{0.70}{\input{#1.tikz}}}

\newcommand{\gempty}{\scalebox{.69}{$\begin{tikzpicture}[scale=0.3]
	\begin{pgfonlayer}{nodelayer}
		\node [style=empty] (0) at (0, 0) {};
	\end{pgfonlayer}
\end{tikzpicture}
$}}
\newcommand{\gCNOT}{\scalebox{.69}{$\begin{tikzpicture}[scale=0.5]
	\begin{pgfonlayer}{nodelayer}
		\node [style=none] (1) at (-1.25, 0.75) {};
		\node [style=none] (4) at (1.25, -0.75) {};
		\node [style=none] (5) at (-1.25, -0.75) {};
		\node [style=none] (6) at (1.25, 0.75) {};
		\node [style=control] (7) at (0, 0.75) {};
		\node [style=not] (8) at (0, -0.75) {};
		\node [style=void] (10) at (-1.25, -0.75) {};
	\end{pgfonlayer}
	\begin{pgfonlayer}{edgelayer}
		\draw (6.center) to (1.center);
		\draw (4.center) to (5.center);
		\draw (7) to (8);
	\end{pgfonlayer}
\end{tikzpicture}
$}}
\newcommand{\gSWAP}{\scalebox{.69}{$\begin{tikzpicture}[scale=0.5]
	\begin{pgfonlayer}{nodelayer}
		\node [style=none] (2) at (-1, -0.75) {};
		\node [style=none] (4) at (1, -0.75) {};
		\node [style=none] (5) at (-1, 0.75) {};
		\node [style=none] (6) at (1, 0.75) {};
		\node [style=none] (7) at (-1.25, 0.75) {};
		\node [style=none] (8) at (1.25, 0.75) {};
		\node [style=none] (9) at (1.25, -0.75) {};
		\node [style=none] (10) at (-1.25, -0.75) {};
		\node [style=none] (14) at (-1.25, 0.75) {};
		\node [style=none] (15) at (-1.25, -0.75) {};
		\node [style=void] (17) at (-1.25, -0.75) {};
	\end{pgfonlayer}
	\begin{pgfonlayer}{edgelayer}
		\draw [style=new] (7.center) to (5.center);
		\draw [style=new] (10.center) to (2.center);
		\draw [style=new] (6.center) to (8.center);
		\draw [style=new] (4.center) to (9.center);
		\draw [style=new, in=0, out=180] (6.center) to (2.center);
		\draw [style=new, in=0, out=-180] (4.center) to (5.center);
	\end{pgfonlayer}
\end{tikzpicture}
$}}
\newcommand{\gH}{\scalebox{.69}{$\begin{tikzpicture}
	\begin{pgfonlayer}{nodelayer}
		\node [style=none] (4) at (-0.75, 0) {};
		\node [style=none] (5) at (0.75, 0) {};
		\node [style=gate] (12) at (0, 0) {$H$};
	\end{pgfonlayer}
	\begin{pgfonlayer}{edgelayer}
		\draw (4.center) to (5.center);
	\end{pgfonlayer}
\end{tikzpicture}
$}}

\newcommand{\gI}{\scalebox{.69}{$\begin{tikzpicture}
	\begin{pgfonlayer}{nodelayer}
		\node [style=none] (0) at (-0.55, 0) {};
		\node [style=none] (1) at (0.55, 0) {};
	\end{pgfonlayer}
	\begin{pgfonlayer}{edgelayer}
		\draw (0.center) to (1.center);
	\end{pgfonlayer}
\end{tikzpicture}
$}}

\newcommand{\gP}{\scalebox{.69}{$\begin{tikzpicture}
	\begin{pgfonlayer}{nodelayer}
		\node [style=none] (0) at (-1.25, 0) {};
		\node [style=none] (1) at (1.25, 0) {};
		\node [style=gate] (2) at (0, 0) {$P(\varphi)$};
	\end{pgfonlayer}
	\begin{pgfonlayer}{edgelayer}
		\draw (0.center) to (1.center);
	\end{pgfonlayer}
\end{tikzpicture}
$}}

\newcommand{\ginit}{\scalebox{.69}{$\begin{tikzpicture}
	\begin{pgfonlayer}{nodelayer}
		\node [style=ancilla] (1) at (-0.25, 0) {};
		\node [style=none] (2) at (0.25, 0) {};
	\end{pgfonlayer}
	\begin{pgfonlayer}{edgelayer}
		\draw (1) to (2.center);
	\end{pgfonlayer}
\end{tikzpicture}
$}}

\newcommand{\propQC}{\textbf{\textup{QC}}}
\newcommand{\QC}{\textup{QC}}

\newcommand{\bvdots}{ \tikz[baseline, every node/.style={inner sep=0}]{ \node at (0,0){.}; \node at (0,-6pt){.}; \node at (0,6pt){.}; } }

\newcommand{\defeq}{\coloneqq}      
\newcommand{\R}{\mathbb{R}}         
\newcommand{\C}{\mathbb{C}}         

\newcommand{\interp}[1]{\left\llbracket #1 \right\rrbracket}

\makeatletter
\def\namedlabel#1#2{\begingroup
	#2%
	\def\@currentlabel{#2}%
	\phantomsection\label{#1}\endgroup
}
\makeatother

\usepackage{xspace}
\usepackage{algorithm}
\usepackage{algorithmicx}
\usepackage[noend]{algpseudocode}
\usepackage{verbatim}

\usepackage{nicematrix} 
\NiceMatrixOptions{
code-for-first-row = \color{blue} ,
code-for-last-row = \color{blue} ,
code-for-first-col = \color{blue} ,
code-for-last-col = \color{blue}
}

\newcommand{\symd}{\mathbin{\Delta}\xspace}

\newcommand{\symdi}[1]{{\scalebox{1.5}{$\symd$}}_{#1}\,}

\makeatletter
\newcommand\etc{etc\@ifnextchar.{}{.\@}\xspace}
\newcommand\ie{i.e.\@\xspace}  
\newcommand\eg{e.g.\@\xspace}

\newcommand\LOG{labelled open graph\@\xspace}
\makeatother

\newcommand{\odd}[2]{\textsf{Odd}_{#1}\left(#2\right)}

\newcommand{\codd}[2]{\textsf{Odd}_{#1}\left[#2\right]}

\newcommand{\ld}{\lambda}
\newcommand{\sse}{\subseteq}

\newcommand{\Pred}{\operatorname{Pred}}
\newcommand{\Succ}{\operatorname{Succ}}

\usepackage{bm}

\newcommand{\abs}[1]{\ensuremath{\left| #1 \right|}}
\usepackage{amsmath,amsthm,amssymb}

\newcommand{\past}{\mathrm{past}}
\newcommand{\fut}{\mathrm{future}}
\newcommand{\para}{\mathrm{parallel}}

\title{Completeness for flow-preserving rewrite rules}

\author{Miriam Backens}{Inria Mocqua team and Université de Lorraine, CNRS, LORIA, F-54000 Nancy, France}{}{https://orcid.org/0000-0002-5418-1084}{}
\author{Simon Perdrix}{Inria Mocqua team and Université de Lorraine, CNRS, LORIA, F-54000 Nancy, France}{}{}{}

\authorrunning{M.~Backens and S.~Perdrix}

\Copyright{Miriam Backens and Simon Perdrix}



\category{} 

\relatedversion{} 

\funding{This work was supported by the Plan France 2030 through the PEPR integrated project EPiQ ANR-22-PETQ-0007 and the HQI platform ANR-22-PNCQ-0002; and by the European project MSCA Staff Exchanges Qcomical HORIZON-MSCA2023-SE-01. The project is also supported by the Maison du Quantique MaQuEst.}

\nolinenumbers 

\begin{document}

\maketitle

\begin{abstract}
 \noindent Complete sets of graphical rewrite rules enable fully graphical reasoning about quantum computations and have been an area of active research for more than a decade. Many recent applications of the ZX-calculus have made use of the close correspondence between ZX-diagrams and computations in the one-way model of measurement-based quantum computation. In this model, various kinds of flow properties ensure deterministic implementability; for ZX-diagrams, these same properties allow efficient translation into quantum circuits (a problem that is known to be \#P-hard in general). Therefore, flow-preserving ZX-calculus rewrite rules are of strong interest.

 Here, we extend the set of flow-preserving rules appearing in the literature with a few new rules and extensions of existing rules. We then show that the resulting rule set is complete for all flow-preserving translations between ZX-diagrams of appropriate form.
 The proof employs a manifestly flow-preserving equivalent of circuit extraction, where a diagram with gflow is transformed, using only flow-preserving rewrite rules, into a diagram with causal flow.
\end{abstract}

\section{Introduction}

The ZX-calculus is a powerful diagrammatic language that enables fully graphical reasoning about quantum computations and has been an active area of research for almost two decades. Equipped with complete graphical rewriting rules, the ZX-calculus finds applications in both the foundational study of quantum mechanics and various quantum computing tasks, such as circuit optimization and the development of error-correcting codes. A key property of the ZX-calculus, compared to traditional quantum circuits, is its flexibility, exemplified by the `only connectivity matters' and bialgebra rules, as well as the ability to merge or split spiders at will. This feature greatly facilitates graphical reasoning but makes the extraction of a quantum circuit \#P-hard in the general case \cite{de_beaudrap_circuit_2022}.

In this context, it is desirable to identify classes of ZX-diagrams for which circuit extraction is efficient, while still retaining the expressive flexibility of the ZX-rewriting rules. Natural candidates are the classes of ZX-diagrams that have a \emph{flow}. Introduced in the setting of measurement-based quantum computation \cite{raussendorf_one-way_2001}, the graphical property of flow provides a way of recognising deterministically-implementable computations.
It also yields a rewriting strategy that enables the efficient extraction of quantum circuits from ZX-diagrams possessing this property \cite{duncan_graph-theoretic_2020,backens_there_2021,simmons_relating_2021}.

There exist several kinds of flow, from the  restrictive causal flow \cite{danos_determinism_2006} to the more general Pauli flow  \cite{browne_generalized_2007}, forming a hierarchy of ZX-diagram classes. However, these classes are not closed under standard ZX-rewriting rules. It is therefore essential to develop alternative rewriting rules that preserve the existence of flows.
The recent notion of \emph{ZX-flow} applies directly to a broader class of ZX-diagrams and one variant is preserved under any standard stabiliser-fragment ZX-rewrites \cite{kissingerZXFlowFlexibleCriterion2026}, yet the connection of this variant to determinism is less immediate.
Therefore our paper works in the traditional framework of graph-like ZX-diagrams with causal flow, gflow, or Pauli flow, leaving connections to ZX-flow for future work.

\section*{Contributions}
While a few rewriting rules, such as local complementation and vertex deletion, were already known to be flow-preserving \cite{duncan_graph-theoretic_2020,backens_there_2021,mcelvanney_complete_2023,mcelvanney_flow-preserving_2023,holker_causal_2023,mcelvanney_preservation_2025,backens_inserting_2025} and a complete set of flow-preserving rewrite rules was known for the Clifford fragment \cite{mcelvanney_complete_2023}, we extend this toolkit by identifying a  set of rewriting rules that preserve various kinds of flow,   and we prove the completeness of this rule set.

Consequently, we equip the most general known fragment of ZX-diagrams that admit efficient circuit extractions with a complete set of diagrammatic rewriting rules. 
In the first instance, this includes all diagrams for which it is meaningful to define flow, in particular graph-like ZX-diagrams \cite{duncan_graph-theoretic_2020} and MBQC-form ZX-diagrams \cite{backens_there_2021}.
The former are diagrams that consist only of $Z$-spiders with Hadamard edges between them, without self-loops or multiple edges, and where every input or output is connected to a $Z$-spider.
The latter consist of an open graph state diagram -- basically a graph-like diagram without phases -- where each spider is connected to either an output wire or a measurement effect from the set \begin{tikzpicture}[baseline=-0.05]
	\begin{pgfonlayer}{nodelayer}
		\node [style=none] (0) at (-0.75, 0) {};
		\node [style=Z phase dot] (1) at (0.25, 0) {$\alpha$};
	\end{pgfonlayer}
	\begin{pgfonlayer}{edgelayer}
		\draw (0.center) to (1);
	\end{pgfonlayer}
\end{tikzpicture}
, \begin{tikzpicture}[baseline=-0.05, scale=.8]
	\begin{pgfonlayer}{nodelayer}
		\node [style=none] (0) at (-0.75, 0) {};
		\node [style=Z phase dot] (1) at (0.25, 0) {$\frac{\pi}{2}$};
		\node [style=X phase dot] (2) at (1.5, 0) {$\alpha$};
	\end{pgfonlayer}
	\begin{pgfonlayer}{edgelayer}
		\draw (0.center) to (1);
		\draw (1) to (2);
	\end{pgfonlayer}
\end{tikzpicture}
 and \begin{tikzpicture}[baseline=-0.05]
	\begin{pgfonlayer}{nodelayer}
		\node [style=none] (0) at (-0.75, 0) {};
		\node [style=X phase dot] (1) at (0.25, 0) {$\alpha$};
	\end{pgfonlayer}
	\begin{pgfonlayer}{edgelayer}
		\draw (0.center) to (1);
	\end{pgfonlayer}
\end{tikzpicture}
.

\begin{theorem}[Informal] 
	The flow-preserving rewrite rules given in~\cref{fig:flow-preserving} are complete. 
\end{theorem}

The complete rewriting rules are: 
(\ref{eq:IO}) rules are input/output expansion/contraction, they are counterparts of the standard ZX spider fusion rule, their application is however restricted to input and outputs spiders and the angles cannot be split. Addition of angles can however be performed through the  (\ref{eq:PF}) rule.
The (\ref{eq:LC}) rule corresponds to the local complementation operation that can be implemented by Clifford $\pi/2$ rotations on graph states. 
The (\ref{eq:ZL}) rule can be interpreted as the red copy, each degree-1 red spider being colour-changed to green by the Hadamard edges and then merged with the neighbouring green spiders.  The (\ref{eq:EU-ZX}) rule is a particular instance of Euler angle decomposition  with only two parameters on the LHS, finally the (\ref{eq:SN}) rule is the spider-nest identity where the angles are appropriate dyadic multiples of~$\pi$.
As already shown in previous works about flow-preserving rewriting, some of these rules can only be applied if the flow on the initial diagram satisfies certain properties: in other words, the cost to pay for flow-preserving rewriting is that the applicability of rewrites is no longer determined entirely locally; instead it may depend on the global property of flow.

The proof of completeness proceeds in essentially two main stages. First, we show that any ZX-diagram with a Pauli flow can be transformed into a circuit-like ZX-diagram equipped with a causal flow. Since ZX-diagrams with causal flow can easily be interpreted as quantum circuits, this first stage is called circuit extraction. We use as intermediate steps particular ZX-diagrams with gflow. Note that circuit extraction already appears in circuit optimization procedures \cite{duncan_graph-theoretic_2020,backens_there_2021,simmons_relating_2021} that we refine here by proving that it can be performed using flow preserving rewrite rules. 
The first stage guarantees that any ZX-diagram with Pauli flow can be transformed, in a flow preserving way, into a circuit-like ZX-diagram. To complete the proof of completeness, the second stage relies on a complete set of circuit rewriting rules that we introduce for quantum circuits with initialisation. 
To the best of our knowledge, this represents the first application of quantum circuit completeness results within the ZX-calculus.

\section{Preliminaries}

\subsection{ZX-calculus}

The ZX-calculus \cite{coecke2011} is a diagrammatic language for reasoning about quantum computations. A ZX-diagram is made of green spiders \input{intro/gspider.tikz}, red spiders  \input{intro/rspider.tikz} and Hadamard \begin{tikzpicture}
	\begin{pgfonlayer}{nodelayer}
		\node [style=hadamard] (0) at (1, 0) {};
		\node [style=none] (1) at (0, 0) {};
		\node [style=none] (2) at (2, 0) {};
	\end{pgfonlayer}
	\begin{pgfonlayer}{edgelayer}
		\draw (1.center) to (0);
		\draw (0) to (2.center);
	\end{pgfonlayer}
\end{tikzpicture}
, together with \emph{wires}: \begin{tikzpicture}
	\begin{pgfonlayer}{nodelayer}
		\node [style=none] (13) at (1.5, 0) {};
		\node [style=none] (14) at (0, 0) {};
	\end{pgfonlayer}
	\begin{pgfonlayer}{edgelayer}
		\draw (14.center) to (13.center);
	\end{pgfonlayer}
\end{tikzpicture}
, \input{intro/cap.tikz}, \input{intro/cup.tikz}, \input{intro/swap.tikz}. These generators are combined using sequential and parallel compositions. 

Any ZX-diagram $D$ with $n$ inputs and $m$ outputs can be interpreted as a  linear map $\interp D: \mathbb C^{2^n}\to \mathbb C^{ 2^n}$ defined as follows: 

\[\begin{array}{rclrcl}
		\interp{\input{intro/gspider.tikz}}&\!\!\coloneqq \!\!&\ket{0^m}\!\bra{0^n}+e^{i\alpha}\ket{1^m}\!\bra{1^n}&\interp{\begin{tikzpicture}
	\begin{pgfonlayer}{nodelayer}
		\node [style=none] (13) at (1.5, 0) {};
		\node [style=none] (14) at (0, 0) {};
	\end{pgfonlayer}
	\begin{pgfonlayer}{edgelayer}
		\draw (14.center) to (13.center);
	\end{pgfonlayer}
\end{tikzpicture}
}&\!\!\coloneqq \!\!&\ket{0}\!\bra{0}+\ket{1}\!\bra{1}\\[0.2cm]
		\interp{\input{intro/rspider.tikz}}&\!\!\coloneqq \!\!& \ket{+^m}\!\bra{+^n}+e^{i\alpha}\ket{-^m}\!\bra{-^n}\quad&\quad\interp{\begin{tikzpicture}
	\begin{pgfonlayer}{nodelayer}
		\node [style=hadamard] (0) at (1, 0) {};
		\node [style=none] (1) at (0, 0) {};
		\node [style=none] (2) at (2, 0) {};
	\end{pgfonlayer}
	\begin{pgfonlayer}{edgelayer}
		\draw (1.center) to (0);
		\draw (0) to (2.center);
	\end{pgfonlayer}
\end{tikzpicture}
}&\!\!\coloneqq \!\!&\ket{+}\!\bra{0}+\ket{-}\!\bra{1} \\[0.2cm]
		\interp{\input{intro/cup.tikz}}&\!\!\coloneqq \!\!&\ket{00}+\ket{11}&\interp{\input{intro/cap.tikz}}&\!\!\coloneqq \!\!&\bra{00}+\bra{11}\\[0.2cm]
		\interp{\input{intro/swap.tikz}}&\!\!\coloneqq \!\!&\ket{00}\!\bra{00}+\ket{01}\!\bra{10}+\ket{10}\!\bra{01}+\ket{11}\!\bra{11}&&&
	\end{array}\]

When equal to zero, the angle of the green or red spider is omitted: \vspace{0.2cm}

\centerline{$\input{intro/gspider-0.tikz}\coloneqq\input{intro/gspider-zero.tikz}\qquad$ and $\qquad \input{intro/rspider-0.tikz}\coloneqq\input{intro/rspider-zero.tikz}$\vspace{0.2cm}}

Moreover, a Hadamard gate connecting to green spiders will be depicted with a dotted blue line:

\centerline{\input{intro/dashedHadamard.tikz}}

ZX-diagrams are \emph{universal} for pure qubit quantum mechanics: $\forall n,m\in \mathbb N$, and $\forall M\in  \mathcal{M}_{2^n \times 2^m}(\mathbb{C})$, there exists a ZX-diagram $D$ such that $\interp D\equiv M$, i.e. $\interp D$  is equal to $M$, up to a non-zero renormalisation factor.   

We consider ZX-diagram up to deformations, following the so-called  \textit{Only Topology Matters} paradigm. It means that two diagrams that can be transformed into each other by moving around the wires are equal. This can be derived from the following rules:

\begin{center}
	\begin{tabular}{ccccc}
		$\input{intro/iswap0.tikz}=\input{intro/iswap1.tikz}~~$&$~\input{intro/scup0.tikz}=\input{intro/cup.tikz}~~$&$~\input{intro/scap0.tikz}=\input{intro/cap.tikz}~~$&$~\input{intro/snake0.tikz}=\begin{tikzpicture}
	\begin{pgfonlayer}{nodelayer}
		\node [style=none] (16) at (1, 0) {};
		\node [style=none] (17) at (0, 0) {};
	\end{pgfonlayer}
	\begin{pgfonlayer}{edgelayer}
		\draw (16.center) to (17.center);
	\end{pgfonlayer}
\end{tikzpicture}
=\input{intro/snake2.tikz}~~$&$~\scalebox{0.9}{\input{intro/sdiagp0.tikz}}\!\!=\!\!\scalebox{0.9}{\input{intro/sdiagp1.tikz}}~$\\
	\end{tabular}	
\end{center}

The legs of the spiders of ZX-calculus can be exchanged and bent. This implies that diagrams are essentially graphs with inputs and outputs.

\begin{center}
	\begin{tabular}{cc}
		$\input{intro/sgspider0.tikz}=\input{intro/sgspider1.tikz}\quad$&$\quad\input{intro/bspider0.tikz}=\input{intro/bspider1.tikz}$\\
	\end{tabular}
\end{center}

ZX-diagrams can also be equipped with an equational theory that allows to transform ZX-diagrams while preserving their interpretation. 
In particular the equations given in \cref{fig:ZXaxioms} are complete \cite{vilmart_near-minimal_2019}, i.e.\ whenever two diagrams represent the same quantum evolution, we can transform one into the other using the rules of the language.

\begin{figure}[!h]
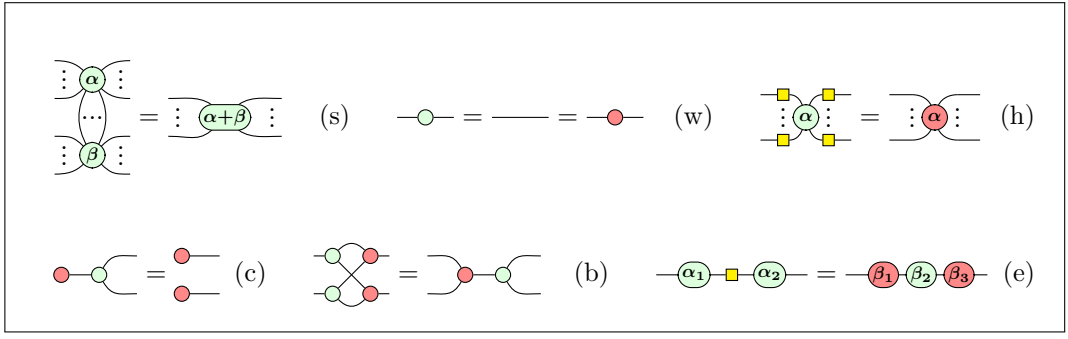

	\fbox{\begin{minipage}{.985\textwidth}
		\begin{subfigure}{.32\textwidth}
			\begin{equation} \tag{s}\label{s}
				\input{intro/spider0.tikz} = \input{intro/spider1.tikz}
			\end{equation}
		\end{subfigure}
		\begin{subfigure}{.34\textwidth}
			\begin{equation} \tag{w}\label{wire}
				\input{intro/wire0.tikz} = \input{intro/wire1.tikz} = \input{intro/wire2.tikz}
			\end{equation}
		\end{subfigure}
		\begin{subfigure}{.3\textwidth}
			\begin{equation} \tag{h}\label{h}
				\input{intro/hadp0.tikz} = \input{intro/hadp1.tikz}
			\end{equation}
		\end{subfigure}
		
		\begin{subfigure}{.24\textwidth}
			\begin{equation} \tag{c}\label{c}
				\input{intro/rcopy0.tikz} = \input{intro/rcopy1.tikz}
			\end{equation}
		\end{subfigure}
		\begin{subfigure}{.32\textwidth}
			\begin{equation} \tag{b}\label{b}
				\input{intro/bigebre0.tikz} = \input{intro/bigebre1.tikz}
			\end{equation}
		\end{subfigure}
		\begin{subfigure}{.4\textwidth}
			\begin{equation} \tag{e}\label{e}
				\input{intro/euler0.tikz} = \input{intro/euler1.tikz}
			\end{equation}
		\end{subfigure}
	\end{minipage}}
	\caption{{\bf A complete equational theory for the ZX-calculus.} \label{fig:ZXaxioms} In Eq.~\eqref{e}, $\alpha_1$, $\alpha_2$ are arbitrary real numbers, and $\beta_{1}:=\arg(z)+\arg(z')$,  $\beta_{2}:=2\arg(i+|\frac{z}{z'}|)$, $\beta_{3}:=\arg(z)-\arg(z')$  where $x^{+}\coloneqq\frac{\alpha_1 + \alpha_2}{2}$, $x^{-}\coloneqq x^{+}- \alpha_2$, $z\coloneqq-\sin(x^{+})+ i\cos(x^{-})$, $z'\coloneqq\cos(x^{+})-i\sin(x^{-})$ and $z'=0 \Rightarrow \beta_2 =0$. 
	In the upper left rule Eq.~\eqref{s}, there must be at least one wire between the spiders annotated by $\alpha$ and $\beta$.
	Equalities hold up to scalar factor.
	}
\end{figure}

Note that one can transform any ZX-diagram into a graph-like diagram, made of green spiders and Hadamard edges only: indeed 
 Eq.~\eqref{h} allows to turn any red spider into a green spider, additionally any sequence of Hadamard gates can be reduced to at most one gate as one can derive \input{intro/hadhad.tikz}, green spiders directly connected can be merged thanks to Eq.~\eqref{s} and finally parallel Hadamard edges can also be simplified thanks to the so-called Hopf law \input{intro/hopf.tikz}.
 
This graph-like representation is strongly related to measurement-based quantum computation as it can be interpreted as a graph state on which one-qubit measurements are performed, depending on the parameters of the spiders \cite{DP2010rewriting}. Before presenting the relationship between ZX-calculus and MBQC  in \cref{sec:Relationship}, we provide in the following section the necessary background on graphs and flows necessary for the rest of the paper.

\subsection{Flow properties}
We will use the following standard graph theory notations. Denote by $\symd$ the symmetric difference operator defined as $A \symd B = (A \cup B) \setminus (A \cap B)$ for any sets $A$ and $B$.
Given a graph $G=(V,E)$, we denote by $N_G(v)$ the neighbourhood of $v\in V$, \ie all vertices that are connected to $v$. We will also write $v\sim w$ if $w\in N_G(v)$. We consider only finite, simple ($v\not\sim v$) and undirected ($u\sim v \Rightarrow v\sim u$) graphs. For any set $S\sse V$, we denote by $\odd{G}{S}$ the \emph{odd neighbourhood} of the set $S$, which contains all vertices that have an odd number of connections to elements of $S$, meaning $\odd{G}{S} = \symdi{v\in S} N_G(v)$.
Additionally, we consider the \emph{closed neighbourhood} $N_G[v] := N_G(v)\cup\{v\}$ and similarly  the \emph{closed odd neighbourhood} of $S$ $\codd{G}{S}:=\symdi{v\in S} N_G[v] = \odd{G}{S}\symd S$. Finally, again for any $S\sse V$, we denote $\comp{S}:=V\setminus S$ the complement of $S$ in $V$.

\begin{definition}
 A \LOG is a tuple $(G,I,O,\ld)$, where $G=(V,E)$ is a finite simple graph, $I\sse V$ is a set of input vertices, $O\sse V$ is a set of output vertices, and $\ld:\comp{O}\to\{X,Y,Z,XY,XZ,YZ\}$ is a function assigning measurement labels to all non-output vertices.
 Measurements with labels in $\{X,Y,Z\}$ are called Pauli measurements and measurements with labels in $\{XY,XZ,YZ\}$ are called planar measurements.
\end{definition}

Traditionally, a flow property is given by a tuple $(c,\prec)$ of a correction function $c$ and a strict partial order $\prec$ which satisfy certain compatibility conditions (and may be further restricted depending on the type of flow property).
Yet these compatibility conditions also have the effect that any valid correction function induces a minimal partial order compatible with it.
This minimal partial order has been called the `dependency order' \cite[p.~2494]{broadbent_parallelizing_2009}, or it has been characterised as a property of `extensivity' of the correction function \cite[p.~4]{perdrix_determinism_2017}.

In this work, we are only concerned with the question of whether a given \LOG satisfies some specific kind of flow property, not with details of that flow.
Therefore, to keep things simpler, we will define flow properties in terms of the correction function and the following relation induced by it.

\begin{definition}\label{def:induced-relation}
 Let $(G,I,O,\ld)$ be a \LOG.
 Any function $c:\comp{O}\to\mathcal{P}(\comp{I})$ induces three relations on $V\times V$ given by
 \begin{align}
  R_X &:= \{(u,v)\mid v\in c(u)\setminus\{u\} \wedge \ld(v)\notin\{X,Y\}\} \label{eq:prec-X} \\
  R_Z &:= \{(u,v)\mid v\in\odd{G}{c(u)}\setminus\{u\} \wedge \ld(v)\notin\{Y,Z\}\} \label{eq:prec-Z}\\
  R_Y &:= \{(u,v)\mid v\in\codd{G}{c(u)}\setminus\{u\} \wedge \ld(v)=Y\} \label{eq:prec-Y}
 \end{align}
 Let $R_c := R_X\cup R_Y\cup R_Z$ and let $R_c^*$ be the transitive closure of $R_c$; we say $R_c^*$ is the \emph{relation induced by $c$}.

 If the induced relation is known to be a strict partial order (which is required for any kind of flow property), we often denote it by $\prec_c$ instead.
\end{definition}

\begin{remark}
 If all measurements are planar, then $R_c$ simplifies to
 \[
  R_c = \{(u,v)\mid v\in p(u)\setminus\{u\}\} \cup \{(u,v)\mid v\in\odd{G}{p(u)}\setminus\{u\}\}.
 \]
\end{remark}

The most general flow property we need is Pauli flow \cite[Definition 5]{browne_generalized_2007}, which allows both planar and Pauli measurements in the underlying \LOG.

\begin{definition}\label{def:Pauli-flow}
A \LOG $(G, I, O, \lambda)$ has Pauli flow if there exists a map $p:\comp{O}\to\mathcal{P}(\comp{I})$ such that $R^*_p$ is a strict partial order and for all $u\in\comp{O}$,
\begin{description}
    \item[\namedlabel{itm:XY}{($\lambda$.XY)}] If $\lambda(u) = XY$, then $u\not \in p(u)$ and $u \in \odd{G}{p(u)}$.
    \item[\namedlabel{itm:XZ}{($\lambda$.XZ)}] If $\lambda(u) = XZ$, then $u \in p(u)$ and $u \in \odd{G}{p(u)}$.
    \item[\namedlabel{itm:YZ}{($\lambda$.YZ)}] If $\lambda(u) = YZ$, then $u\in p(u)$ and $u \not \in \odd{G}{p(u)}$.
    \item[\namedlabel{itm:X}{($\lambda$.X)}] If $\lambda(u) = X$, then $u \in \odd{G}{p(u)}$.
    \item[\namedlabel{itm:Z}{($\lambda$.Z)}] If $\lambda(u) = Z$, then $u \in p(u)$.
    \item[\namedlabel{itm:Y}{($\lambda$.Y)}] If $\lambda(u) = Y$, then $u\in\codd{G}{p(u)}$.
\end{description}
\end{definition}

A Pauli flow on a \LOG where all measurements are planar, \ie where $\ld(V)\sse\{XY,XZ,YZ\}$, is called a \emph{generalised flow} or \emph{gflow}\footnote{This version is sometimes called \emph{extended gflow}, with the plain term `gflow' reserved for \LOG{}s where all measurement labels are $XY$, yet the original definition is the general one and this will be the one we need in this work.} \cite[Definition 5]{browne_generalized_2007}.

\begin{definition}\label{def:gflow}
 A \LOG $(G, I, O, \ld)$ with $\ld(V)\sse\{XY,XZ,YZ\}$ has \emph{gflow} if there exists a map $g:\comp{O}\to\mathcal{P}(\comp{I})$ such that $R^*_g$ is a strict partial order and for all $u\in\comp{O}$,
\begin{description}
    \item[\namedlabel{itm:gXY}{(G.XY)}] If $\lambda(u) = XY$, then $u\not \in g(u)$ and $u \in \mathrm{Odd}_G(g(u))$.
    \item[\namedlabel{itm:gXZ}{(G.XZ)}] If $\lambda(u) = XZ$, then $u \in g(u)$ and $u \in \mathrm{Odd}_G(g(u))$.
    \item[\namedlabel{itm:gYZ}{(G.YZ)}] If $\lambda(u) = YZ$, then $u\in g(u)$ and $u \not \in \mathrm{Odd}_G(g(u))$.
\end{description}
\end{definition}

If a \LOG has Pauli flow, there are generally multiple correction functions which induce a Pauli flow; the same holds for \LOG{}s with gflow.
The following conditions characterise particularly useful types of flow called \emph{focused}.
These focusing conditions were first defined for $XY$-plane only gflow \cite[Definition~5]{mhalla_which_2014}, and then generalised to extended gflow \cite[Section~3.3]{backens_there_2021} and Pauli flow \cite[Definition~4.3]{simmons_relating_2021}; we give the general version here (which encompasses the predecessors as special cases).

\begin{definition}[{\cite[Definition~4.3]{simmons_relating_2021}}]\label{def:focusing}
 Given a \LOG $\Gamma = (G,I,O,\ld)$, a set $\hat{p}\sse\comp{I}$ is \emph{focused over $S\sse\comp{O}$} if:
\begin{description}
 \item[\namedlabel{itm:focus-X}{(FX)}] For all $w\in S\cap\hat{p}$, we have $\ld(w)\in\{XY,X,Y\}$.
 \item[\namedlabel{itm:focus-Z}{(FZ)}] For all $w\in S\cap\odd{G}{\hat{p}}$, we have $\ld(w)\in\{XZ,YZ,Y,Z\}$.
 \item[\namedlabel{itm:focus-Y}{(FY)}] For all $w\in S$ such that $\ld(w)=Y$, we have $w\in\hat{p}$ if and only if $w\in\odd{G}{\hat{p}}$.
\end{description}
 A \emph{focused set $\hat{p}$ for $\Gamma$} is focused over $\comp{O}$.
 A Pauli flow $(p,\prec)$ is \emph{focused} if $p(v)$ is focused over $\comp{O}\setminus\{v\}$ for all $v\in\comp{O}$.
 Similarly, a gflow $(g,\prec)$ is focused if $g(v)$ is focused over $\comp{O}\setminus\{v\}$ for all $v\in\comp{O}$.
\end{definition}

If a \LOG has Pauli flow, then it has focused Pauli flow \cite[Lemma~B.4]{simmons_relating_2021}.
The same holds for (extended) gflow \cite{backens_there_2021,mhalla_which_2014}.
Moreover, if the \LOG satisfies $\abs{I}=\abs{O}$, then the focused flow is unique \cite{mhalla_which_2014,simmons_relating_2021}.
The conditions of a focused gflow or Pauli flow can be expressed in matrix form \cite{mitosekAlgebraicFormulationPauli2026}.

Another flow variant is \emph{causal flow} \cite[Definition 2]{danos_determinism_2006} which is traditionally defined for \LOG{}s in which all vertices are measured in the $XY$ plane (in which case the labelling is usually left implicit).
While historically the first type of flow to be described, causal flow (like gflow) can be considered a special case of Pauli flow.
Causal flow restricts all correction sets to being of size~1; in a slight abuse of notation we will identify functions $\{c:\comp{O}\to\mathcal{P}(\comp{I})\mid\forall v. \abs{c(v)}=1\}$ with functions $c:\comp{O}\to\comp{I}$.

\begin{definition}
 A \LOG $(G, I, O, \ld)$ where $\ld(V)\sse\{XY\}$ has \emph{causal flow} if there exists a map $f: \comp{O}\to\comp{I}$ such that $R^*_f$ is a strict partial order and for all $v\in\comp{O}$, we have $v\sim f(v)$.
\end{definition}

\begin{remark}
 The properties of being causal and being focused are not generally compatible except in fairly trivial cases.
 Taken together, the causal and focusing properties would imply that $f(v)$ cannot have any neighbours other than $v$ and outputs, which severely restricts the structure of the \LOG.
\end{remark}

In older literature, causal flow is simply called `flow', but we will only use the term `flow' to refer to Pauli flow, gflow, or causal flow collectively; and always write `causal flow' when we mean specifically the latter.

When causal flow is considered as a restriction of gflow specifically, this leads to a straightforward extension of the definition, which has been suggested multiple times in the community \cite{perez_measurement-based_2023,holker_causal_2023}.

\begin{definition}\label{def:causal}
 A \LOG $(G, I, O, \ld)$ with $\ld(V)\sse\{XY,YZ\}$ has \emph{extended causal flow} if there exists a map $f: \comp{O}\to\comp{I}$ such that $R^*_f$ is a strict partial order and for all $v\in\comp{O}$:
 \begin{description}
  \item[\namedlabel{itm:cXY}{(C.XY)}] If $\ld(v)=XY$, then $f(v)\sim v$.
  \item[\namedlabel{itm:cYZ}{(C.YZ)}] If $\ld(v)=YZ$, then $f(v)=v$.
 \end{description}
\end{definition}

We will sometimes consider (extended) causal flow to be defined in terms of a map $f:\comp{O}\to\mathcal{P}(\comp{I})$ which satisfies $\abs{f(v)}=1$ for all $v\in\comp{O}$.

This size condition in extended causal flow means the \LOG cannot contain any $XZ$-measurements: the latter need to be both in their correction set and in the odd neighbourhood of their correction set.
Yet since we are considering only simple graphs where there are no self-loops, this is not achievable with a correction set of size~1.

For each of the above flow properties, there is a polynomial-time algorithm which decides whether a given \LOG has this type of flow and, if yes, returns a suitable correction function and its induced partial order \cite{mhalla_finding_2008,backens_there_2021,simmons_relating_2021}.

The following is well-known; for focused Pauli flow or gflow it follows from the algebraic definition \cite{mitosekAlgebraicFormulationPauli2026} and for $XY$-only extended causal flows it follows from the induced path cover~\cite{de_beaudrap_finding_2008}.
We give a quick proof here to also explicitly cover extended causal flow and non-focused gflow or Pauli flow.

\begin{remark}\label{rem:unique-correction-sets}
	For a flow $g$ of any type on some labelled open graph $(G,I,O,\ld)$, we always have that correction sets of distinct vertices are distinct.
	Suppose there exist $v,v'\in\comp{O}$ such that $g(v)=g(v')$.
	Now $\ld(v)\in\{X,XY\}\implies v\in\odd{G}{g(v)}$, $\ld(v)\in\{XZ,YZ,Z\}\implies v\in g(v)$ and $\ld(v)=Y \implies v\in\codd{G}{g(v)}$.
	Thus, since $g(v')=g(v)$, we have either $v=v'$ or $(v',v)\in R_g$.
	But symmetrically we can also argue that $(v,v')\in R_g$ if $v\neq v'$, so in the latter case $R_g^*$ is not a strict partial order.
	Hence two vertices with the same correction set must be identical.
\end{remark}

\subsection{Relationship between MBQC and ZX-calculus}\label{sec:Relationship}

The ZX-calculus provides a powerful and intuitive framework to describe MBQC, where the flow conditions naturally yield a systematic rewriting strategy that translates MBQC computations into quantum circuits~\cite{DP2010rewriting}, enabling efficient circuit optimisation procedures~\cite{duncan_graph-theoretic_2020,backens_there_2021}.

\begin{definition}[\cite{backens_zx-calculus_2014,backens_there_2021}]
 A \emph{graph state diagram} is a ZX-diagram where each vertex is a (phase-free) green spider, each edge connecting spiders has a Hadamard gate on it, and there is a single output wire incident on each vertex.

 A ZX-diagram is in \emph{MBQC-form} if it consists of a graph state diagram with the following modifications: each vertex of the graph may furthermore be connected to an input (in addition to its output), and a measurement effect instead of its output.
 The allowed measurement effects and the corresponding measurement labels are:
 \begin{center}
  \begin{tabular}{c|c|c|c|c|c}
  $XY$ & $XZ$ & $YZ$ & $X$ & $Z$ & $Y$ \\ \hline
  \begin{tikzpicture}[baseline=-0.05]
	\begin{pgfonlayer}{nodelayer}
		\node [style=none] (0) at (-0.75, 0) {};
		\node [style=Z phase dot] (1) at (0.25, 0) {$\alpha$};
	\end{pgfonlayer}
	\begin{pgfonlayer}{edgelayer}
		\draw (0.center) to (1);
	\end{pgfonlayer}
\end{tikzpicture}
 & \begin{tikzpicture}[baseline=-0.05, scale=.8]
	\begin{pgfonlayer}{nodelayer}
		\node [style=none] (0) at (-0.75, 0) {};
		\node [style=Z phase dot] (1) at (0.25, 0) {$\frac{\pi}{2}$};
		\node [style=X phase dot] (2) at (1.5, 0) {$\alpha$};
	\end{pgfonlayer}
	\begin{pgfonlayer}{edgelayer}
		\draw (0.center) to (1);
		\draw (1) to (2);
	\end{pgfonlayer}
\end{tikzpicture}
 & \begin{tikzpicture}[baseline=-0.05]
	\begin{pgfonlayer}{nodelayer}
		\node [style=none] (0) at (-0.75, 0) {};
		\node [style=X phase dot] (1) at (0.25, 0) {$\alpha$};
	\end{pgfonlayer}
	\begin{pgfonlayer}{edgelayer}
		\draw (0.center) to (1);
	\end{pgfonlayer}
\end{tikzpicture}
 & \begin{tikzpicture}[baseline=-0.05]
	\begin{pgfonlayer}{nodelayer}
		\node [style=none] (0) at (-0.75, 0) {};
		\node [style=Z phase dot] (1) at (0.25, 0) {$a\pi$};
	\end{pgfonlayer}
	\begin{pgfonlayer}{edgelayer}
		\draw (0.center) to (1);
	\end{pgfonlayer}
\end{tikzpicture}
 & \begin{tikzpicture}[baseline=-0.05]
	\begin{pgfonlayer}{nodelayer}
		\node [style=none] (0) at (-0.75, 0) {};
		\node [style=X phase dot] (1) at (0.25, 0) {$a\pi$};
	\end{pgfonlayer}
	\begin{pgfonlayer}{edgelayer}
		\draw (0.center) to (1);
	\end{pgfonlayer}
\end{tikzpicture}
 & \begin{tikzpicture}[baseline=-0.05]
	\begin{pgfonlayer}{nodelayer}
		\node [style=none] (0) at (-1.25, 0) {};
		\node [style=Z phase dot] (1) at (0.25, 0) {$(-1)^a\frac{\pi}{2}$};
	\end{pgfonlayer}
	\begin{pgfonlayer}{edgelayer}
		\draw (0.center) to (1);
	\end{pgfonlayer}
\end{tikzpicture}

  \end{tabular}
 \end{center}

\end{definition}

\begin{example}
	On the left-hand side, a graph state diagram and on the right-hand side an MBQC-form diagram with two inputs, two $XY$-measurements and one $YZ$-measurement, and two outputs (on the same underlying graph):
	\[
		\input{graph-state-ex.tikz} \qquad\qquad\qquad\qquad
		\input{MBQC-form-ex.tikz}
	\]
\end{example}

Via a straightforward extension of previous work in the context of extended gflow \cite[Section~2.2]{backens_there_2021}, we have a correspondence between MBQC-form ZX-diagrams and tuples $(\Gamma, \alpha)$ where $\Gamma = (G,I,O,\ld)$ is a \LOG and $\alpha:\comp{O}\to\mathbb{R}$ has the property that $\ld(v)\in\{X,Z\}\implies\alpha(v)\in\{0,\pi\}$ and $\ld(v)=Y\implies\alpha(v)\in\{\pm\frac{\pi}{2}\}$.
If there are vertices whose measurement angle is an integer multiple of $\frac\pi2$, multiple \LOG{}s (which differ only in the labelling) will correspond to the same ZX-diagram: \eg for a measurement effect \begin{tikzpicture}
	\begin{pgfonlayer}{nodelayer}
		\node [style=X phase dot] (0) at (0.75, 0) {$a\pi$};
		\node [style=none] (1) at (0, 0) {};
	\end{pgfonlayer}
	\begin{pgfonlayer}{edgelayer}
		\draw (1.center) to (0);
	\end{pgfonlayer}
\end{tikzpicture}
 with $a\in\{0,1\}$, the \LOG could have a measurement label of $YZ$ or $Z$.\footnote{In principle, it could even be $XZ$, although that would generally be represented as \begin{tikzpicture}
	\begin{pgfonlayer}{nodelayer}
		\node [style=Z phase dot] (0) at (0.75, 0) {$\frac\pi2$};
		\node [style=none] (1) at (0, 0) {};
		\node [style=X phase dot] (2) at (2.75, 0) {$(-1)^{a+1}\frac\pi2$};
	\end{pgfonlayer}
	\begin{pgfonlayer}{edgelayer}
		\draw (1.center) to (0);
		\draw (0) to (2);
	\end{pgfonlayer}
\end{tikzpicture}
.}

In this work we will usually prefer planar labels since we want to work with gflow.
Nevertheless, for any MBQC-form ZX-diagram, the unique \LOG with the maximum number of Pauli measurement labels is also useful: if this \LOG does not have Pauli flow, then no \LOG corresponding to the ZX-diagram can have any kind of flow.
This is because the correction conditions for Pauli measurements are subsets of those of the planes in which the Pauli measurement is contained.

\begin{definition}
 An MBQC-form ZX-diagram is said to have:
 \begin{itemize}
 	\item \emph{Pauli flow} if the corresponding \LOG with the maximum number of Pauli measurements has Pauli flow,
 	\item \emph{gflow} if there exists a corresponding \LOG (where all measurements are planar) which has gflow, and
 	\item \emph{extended causal flow} if there exists a corresponding \LOG (where all measurements are $XY$ or $YZ$) which has extended causal flow.
 \end{itemize}
\end{definition}

\section{Flow-preserving rewrite rules}
\label{s:flow-preserving-rules}

We first recap some known flow-preserving rules in Section~\ref{s:known-rules} and then add some new rules in Section~\ref{s:new-rules}.
Together, these rules allow us to derive many other known rewrite rules in Section~\ref{s:derived-rules}.
As transformations between two ZX-diagrams (or \LOG{}s) may be flow-preserving in only one direction, we denote rewrite rules with either one arrow `$\to$' or a pair of arrows `$\rightleftarrows$' depending on whether they are proved to be flow-preserving in one or both directions.

\subsection{Known flow-preserving rewrite rules}
\label{s:known-rules}

Several of the core flow-preserving rewrite rules employ the graph-theoretic operations of local complementation and pivoting.
Some of these rules modify the measurement effects on certain vertices, which means they cannot be applied in situations where those vertices may be outputs.
Other rules modify certain edges in the graph and thus cannot be applied to input or outputs (as the corresponding `dangling' wires would behave like an edge).
Nevertheless, these restrictions are not a problem as we can `unfuse' inputs and outputs as needed:

\begin{lemma}[{\cite[Propositions~4.1 and~4.2]{backens_there_2021}}]\label{lem:input-output-fusion}
 The following two operations on MBQC-form ZX-diagrams preserve interpretation and the existence of extended causal flow, gflow or Pauli flow in both directions, where in the first diagram the original vertex is an input and in the second diagram, the original vertex is an output and may or may not be an input (if it is, then the leftmost vertex in the right-hand diagram is an input):
 \begin{equation}\tag{IO}
  \input{rules/input-splitting-lhs.tikz} \quad\rightleftarrows\quad \input{rules/input-splitting-rhs.tikz}
   \qquad\qquad\qquad
  \input{rules/output-splitting-lhs.tikz} \quad\rightleftarrows\quad \input{rules/output-splitting-rhs.tikz}
 \end{equation}
 Even if the goal is preservation of Pauli flow, the newly-introduced vertices in the above diagrams can all be labelled $XY$.
\end{lemma}
\begin{proof}
 The original proofs are for extended gflow but they generalise straightforwardly to Pauli flow or restrict to extended causal flow.
 Similarly, the reverse directions of the above rules can straightforwardly be proved.

 For example, preservation of extended causal flow of the case with an input can be seen as follows.
 Label the vertices in the subdiagrams of interest as follows:
 \[
 	\input{rules/input-splitting-lhs-labelled.tikz} \quad\rightleftarrows\quad \input{rules/input-splitting-rhs-labelled.tikz}
 \]
 For the left-to-right direction, suppose the ZX-diagram in which the left-hand side is embedded corresponds to a labelled open graph $\Gamma = (G,I,O,\ld)$ with $G=(V,E)$ and which has extended causal flow $g$ with induced partial order $\prec_g$.
 Then the ZX-diagram in which the right-hand side is embedded corresponds to a labelled open graph $\Gamma' = (G',I',O,\ld')$, where $G' = (V\cup\{a,b\}, E\cup\{\{a,b\},\{b,c\}\})$, $I'=I\setminus\{c\}\cup\{a\}$, and $\ld'(a) = \ld'(b) = XY$ while $\ld'(v) = \ld(v)$ for all $v\in V\setminus O$.
 Define
 \[
 	g'(v) = \begin{cases}
 		b &\text{if } v = a \\
 		c &\text{if } v = b \\
 		g(v) &\text{otherwise,}
 	\end{cases}
 \]
 then conditions \ref{itm:cXY} and \ref{itm:cYZ} hold by inspection for all measured vertices.
 It remains to consider the partial order.
 Let $R_X, R_Z$ defined according to Definition~\ref{def:induced-relation} for $\Gamma$ and $g$, and let $R_X', R_Z'$ be the corresponding sets for $\Gamma'$ and $g'$; we have $R_Y=R_Y'=\emptyset$ as a labelled open graph which has (extended) causal flow cannot contain any $Y$-measurements.
 Then:
 \begin{align*}
 	R_X' &= R_X \cup \{(a,b), (b,c)\} \\
 	R_Z' &= R_Z \cup \{(a,c)\} \cup \{(b,v)\mid v\in N_G(c)\}
 \end{align*}
 In the new relationships, any vertices of $V$ appear only as successors.
 The new vertex $a$ appears only as a predecessor, and the same is nearly true for $b$, with the exception that $a$ precedes $b$.
 Thus $R_g^*$ being a strict partial order implies that $R_{g'}^*$ is also a strict partial order.
 Therefore $\Gamma'$ has extended causal flow.
 
 For the right-to-left direction, use the same definitions for labelled open graphs as before, only this time we start from $\Gamma'$ which is assumed to have an extended causal flow $g'$.
 Consider the restriction $g$ of $g'$ to only those vertices appearing in $\Gamma$.
 Then for each $v\in V\setminus O$, we have $g(v)\in V$.
 Indeed $g(v)\notin I$ as elements of $I'$ do not appear in correction sets and we must have had $g'(a)=b$, $g'(b)=c$ which implies that $g(v)\neq c$ for any $v\in V\setminus O$ as correction sets have to be unique by Remark~\ref{rem:unique-correction-sets}.
 Thus $g$ is valid correction function.
 Moreover, the relationship induced by $g$ on $\Gamma$ is just a sub-relation of that induced by $g'$ on $\Gamma'$, so it must be a strict partial order.
\end{proof}

\begin{definition}
 Let $G = (V, E)$ be a graph and $u\in V$. The \emph{local complementation} of $G$ about $u$ is the operation which maps $G$ to $G*u := (V, E \symd \{(b, c) \mid (b, u), (c, u) \in E \wedge b \neq c\})$.
 The \emph{pivot} of G about the edge $(u, v)\in E$ is the operation mapping $G$ to the graph $G \wedge uv := G * u * v * u$.
\end{definition}

\begin{lemma}[Local complementation {\cite[Lemma~3.1]{backens_there_2021}, \cite[Lemma D.14]{simmons_relating_2021}}]\label{lem:local-complementation-flow}
 Let $u$ be a non-input vertex in an MBQC-form ZX-calculus diagram and suppose none of the elements of $N_G[u]$ are outputs.
 Then a local complementation on the underlying graph in combination with the changes on the measurement effects resulting from merging with \begin{tikzpicture}
	\begin{pgfonlayer}{nodelayer}
		\node [style=X phase dot] (0) at (0, 0) {$-\frac{\pi}{2}$};
		\node [style=none] (1) at (-1, 0) {};
		\node [style=none] (2) at (1, 0) {};
	\end{pgfonlayer}
	\begin{pgfonlayer}{edgelayer}
		\draw (1.center) to (2.center);
	\end{pgfonlayer}
\end{tikzpicture}
 on $u$ and \begin{tikzpicture}
	\begin{pgfonlayer}{nodelayer}
		\node [style=Z phase dot] (0) at (0, 0) {$\frac{\pi}{2}$};
		\node [style=none] (1) at (-0.75, 0) {};
		\node [style=none] (2) at (0.75, 0) {};
	\end{pgfonlayer}
	\begin{pgfonlayer}{edgelayer}
		\draw (1.center) to (2.center);
	\end{pgfonlayer}
\end{tikzpicture}
 on neighbours of $u$ preserves the diagram interpretation and the existence of gflow or Pauli flow:
 \begin{equation}\tag{LC}
 	\input{rules/LC-lhs.tikz} \quad\to\quad \input{rules/LC-rhs.tikz}
 \end{equation}
\end{lemma}

Local complementations on their own preserve extended causal flow only in special cases, but they can be part of sequences of operations that preserve extended causal flow overall.
The local complementation rule is given as a directed rule because applying four successive local complementations about the same $u$ is the identity operation, so local complementations generate their own inverse.
The \emph{pivoting} operation is just a sequence of three local complementations about two neighbouring vertices, so flow-preservation is immediate.
The closed-form expression for the post-pivot gflow or Pauli flow was derived in \cite[Lemma~5.3]{backens_inserting_2025}.

\begin{lemma}[Pivoting {\cite[Lemma 4.5]{backens_there_2021}, \cite[Lemma D.20]{simmons_relating_2021}, \cite[Lemma~5.3]{backens_inserting_2025}}]\label{lem:pivot-Pauli-flow}
 Let $u,v$ be two adjacent non-input vertices in an MBQC-form ZX-calculus diagram and suppose none of the elements of $N_G(u)\cap N_G(v)$ are outputs.
 Then a pivot on the edge $uv$ -- equivalently, a sequence of local complementations about $u$, then $v$, then $u$ again -- preserves interpretation in combination with merging the resulting operators \begin{tikzpicture}
	\begin{pgfonlayer}{nodelayer}
		\node [style=none] (1) at (-0.75, 0) {};
		\node [style=none] (2) at (0.75, 0) {};
		\node [style=hadamard] (3) at (0, 0) {};
	\end{pgfonlayer}
	\begin{pgfonlayer}{edgelayer}
		\draw (1.center) to (2.center);
	\end{pgfonlayer}
\end{tikzpicture}
 on $u$ and $v$ and \begin{tikzpicture}
	\begin{pgfonlayer}{nodelayer}
		\node [style=Z phase dot] (0) at (0, 0) {$\pi$};
		\node [style=none] (1) at (-0.75, 0) {};
		\node [style=none] (2) at (0.75, 0) {};
	\end{pgfonlayer}
	\begin{pgfonlayer}{edgelayer}
		\draw (1.center) to (2.center);
	\end{pgfonlayer}
\end{tikzpicture}
 on joint neighbours of $u$ and $v$ into the measurement effects.
 It also preserves Pauli flow or gflow.
 \[
  \input{rules/pivot-1.tikz} \quad \to \quad \input{rules/pivot-2.tikz}
 \]
 If the original flow is $g$, then the flow after the pivot is given by
 \[
   g'(w) = \begin{cases}
   	g(w) &\text{if } u,v\notin\codd{G}{g(w)} \\
   	g(w)\symd\{u\} &\text{if } u\in\codd{G}{g(w)} \text{ and } v\notin\codd{G}{g(w)} \\
   	g(w)\symd\{v\} &\text{if } u\notin\codd{G}{g(w)} \text{ and } v\in\codd{G}{g(w)} \\
   	g(w)\symd\{u,v\} &\text{if } u,v\in\codd{G}{g(w)}
   \end{cases}
 \]
 for all $w\in\comp{O}$ and satisfies ${\prec_{g'}} = {\prec_g}$.
\end{lemma}

Even when gflow is preserved by an application of pivoting, extended causal flow need not be preserved.
We will use this lemma in a causal flow context only in situations where it is straightforward to manually check that the causal flow is preserved.

Similar to local complementations, pivoting provides its own inverse.
For the following operation, invertibility is more complicated.

\begin{lemma}[$Z$-like deletion {\cite[Lemma~3.4]{backens_there_2021}, \cite[Lemma D.7]{simmons_relating_2021}, \cite[Corollary~4.5]{backens_inserting_2025}}]\label{lem:Z-deletion-flow}
 Let $u$ be a vertex in an MBQC-form diagram whose measurement effect is \begin{tikzpicture}
	\begin{pgfonlayer}{nodelayer}
		\node [style=X phase dot] (0) at (0.75, 0) {$a\pi$};
		\node [style=none] (1) at (0, 0) {};
	\end{pgfonlayer}
	\begin{pgfonlayer}{edgelayer}
		\draw (1.center) to (0);
	\end{pgfonlayer}
\end{tikzpicture}
.
 Then the ZX-diagram that results from deleting $u$, its measurement effect, and all incident edges and applying $Z^a$ to the measurement effects of all neighbours has the same interpretation as the original diagram.
 Moreover, it has Pauli flow if and only if the original diagram has Pauli flow, gflow if and only if the original diagram has gflow, and extended causal flow if and only if the original diagram has extended causal flow.
 \ctikzfig{rules/Z-delete}
\end{lemma}

Flow preservation for deletion of a vertex measured as \begin{tikzpicture}
	\begin{pgfonlayer}{nodelayer}
		\node [style=X phase dot] (0) at (0.75, 0) {$a\pi$};
		\node [style=none] (1) at (0, 0) {};
	\end{pgfonlayer}
	\begin{pgfonlayer}{edgelayer}
		\draw (1.center) to (0);
	\end{pgfonlayer}
\end{tikzpicture}
 does not depend on what measurement label that vertex was assigned -- whether $Z$, $XZ$, or $YZ$.
Yet for the opposite direction, we have to assume that the new vertex is labelled $Z$, otherwise flow need not be preserved.
So while the previous two lemmas hold for multiple types of flow, the next lemma holds in full generality only for the Pauli flow case; further restrictions are needed to ensure the preservation of gflow.

\begin{definition}\label{def:insert-Z}
	Let $(G,I,O,\ld)$ be a \LOG, let $\ell\in\{Z, YZ, XZ\}$, and let $N\sse V$.
	The \emph{\LOG that results from inserting a new $\ell$-measured vertex $z$ with neighbours $N$} is $(G',I,O,\ld')$, where $G' = (V\uplus\{z\}, E\uplus\{\{z,v\}\mid v\in N\}\})$ and
	\[
	\ld'(v) = \begin{cases}
		\ell &\text{if } v=z \\
		\ld(v) &\text{otherwise.}
	\end{cases}
	\]
\end{definition}

\begin{lemma}[$Z$-insertion {\cite[Proposition~4.1]{mcelvanney_complete_2023}}]\label{lem:Z-insertion-flow}
 Consider an MBQC-form ZX-diagram which has Pauli flow, and insert a new $Z$-measured vertex with arbitrary neighbours.
 If the new vertex has phase $\pi$, apply $Z$ to the measurement effects of all neighbours.
 The resulting diagram has the same interpretation as the original, and it again has Pauli flow:
 \ctikzfig{rules/Z-insert}
\end{lemma}

$Z$-deletion and $Z$-insertion correspond to the two directions of the ZX-calculus copy rule.
These two rules together with local complementations are complete for flow-preserving transformations on patterns where all measurements are Pauli \cite{mcelvanney_complete_2023}.

The $Z$-insertion rule generalises to $XZ$- and $YZ$-measurements, though in the latter case the proposed neighbours of the new vertex and the existing Pauli flow, gflow, or extended causal flow must satisfy certain compatibility conditions.
On the other hand, as $XZ$ and $YZ$ are planar measurements, this type of insertion can preserve gflow, whereas the insertion of a Pauli measurement by definition leads to a Pauli flow.

We will formalise the new rule using the language of labelled open graphs and give a definition that is specialised to $YZ$-insertions into gflow as it will only be required in that context.
For the general versions see Theorems~4.3 and~5.2 of Ref.~\cite{backens_inserting_2025}.

\begin{theorem}[adapted from {\cite[Theorem~4.3]{backens_inserting_2025}}]\label{thm:YZ-insertion-preserves-gflow}
 Let $\Gamma = (G,I,O,\ld)$ be a \LOG and let $N\sse V$.
 Suppose $\Gamma' = (G',I,O,\ld')$ is the \LOG that results from inserting a new $YZ$-measured vertex $z$ with neighbours $N$ as in Definition~\ref{def:insert-Z}.
 Then $\Gamma'$ has gflow if and only if there exist a gflow $g$ on $\Gamma$ and a set $C\sse V\setminus I$ such that the following hold:
 \begin{itemize}
  \item The two sets $N$ and $C$ satisfy $\abs{N\cap C}\equiv 0 \bmod 2$.
  \item Define the following two sets in $G$:
   \begin{align}
    \Pred(N) &:= \{u\in V\setminus O : \abs{g(u)\cap N}\equiv 1 \bmod 2\} \label{eq:pred} \\
    \Succ(N,C) &:= C\cup (N\symd\odd{G}{C}), \label{eq:succ}
   \end{align}
   then all $w\in\Pred(N)$ and all $v\in\Succ(N,C)$ satisfy $\neg(v\prec_g w \vee v = w)$.
 \end{itemize}
 The conditions for the insertion of an $XZ$-measurement are almost the same as the above, the only difference is that $\abs{N\cap C}\equiv 1 \bmod 2$.
\end{theorem}

In the above theorem, the first property states that \ref{itm:gYZ} holds for $C\cup\{z\}$ relative to vertex $z$.
The second property ensures that if $g'$ is the extension of $g$ to $\Gamma'$ given by
\begin{equation}\label{eq:yz-insert-correction-function}
	g'(u) = \begin{cases}
		C\cup\{z\} &\text{if } u = z \\
		g(u) &\text{otherwise,}
	\end{cases}
\end{equation}
then the relation induced by $g'$ (cf.\ Definition~\ref{def:induced-relation}) is a partial order.
That then implies that $g'$ is a gflow.
There is also a characterisation of when $YZ$-insertion is extended causal flow-preserving \cite[Theorem~3.7]{backens_inserting_2025}, the details of which will not be needed here.

\subsection{New fundamental flow-preserving rewrite rules}
\label{s:new-rules}

In the complete rule set for the universal ZX-calculus in Figure~\ref{fig:ZXaxioms}, most of the rewrite rules fall within the stabiliser fragment.
Hence their flow-preserving equivalents can already be derived as local complementation, $Z$-insertion and $Z$-deletion are complete for flow-preserving rewriting of stabiliser diagrams \cite[Theorem~4.2]{mcelvanney_complete_2023}.
The exceptions are the spider rule~(s1) and the generalised Euler decomposition rule~(e).
(The colour change rule (h) is not strictly within the stabiliser fragment but can be derived for all values of the phase label $\alpha$ as the label just remains unchanged.)
In the flow-preserving setting, we can use $Z$-insertion together with pivoting to merge or split the neighbourhood of a spider, so we will just need an additional rule for the splitting or addition of phases in the spider rule.
Hence it only remains to consider how the generalised Euler decomposition rule interacts with flow properties.
For technical reasons linked to the circuit completeness result on which we draw, we will also use a `spider nest' rule that allows certain combinations of phase gadgets to be removed as they apply a trivial phase overall.

\begin{lemma}[Euler rule]\label{lem:Euler-rule}
	Suppose $D$ is an MBQC-form ZX-diagram that contains two adjacent $XY$-measured vertices $a,b$.
	Let $D'$ be the diagram that replaces the edge between these two vertices by a path diagram containing three new $XY$-measurements, whose phase angles $\beta_1,\beta_2\beta_3$ are given below as a function of the phases $\alpha_1,\alpha_2$ of $a$ and $b$.
 	Then $D$ and $D'$ have the same interpretation.
    Moreover, $D'$ has extended causal flow if $D$ has an extended causal flow $g$ such that $g(a) = \{b\}$ or $g(b)=\{a\}$:
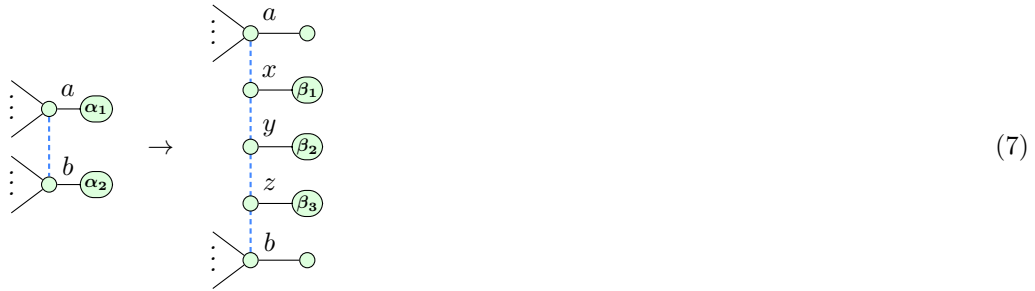
 \begin{equation}
  \input{rules/Euler-prime-lhs.tikz} \quad\to\quad \input{rules/Euler-prime-rhs.tikz}
 \end{equation}
where, for any angles $\alpha_1,\alpha_2\in\R$ the angles $\beta_1,\beta_2,\beta_3$ are restricted to $[0,2\pi)$ and are computed as follows:
\begin{gather*}
  z\defeq -\sin\left(\frac{\alpha_1+\alpha_2}{2}\right)+i\cos\left(\frac{\alpha_1-\alpha_2}{2}\right) \\
  z'\defeq \cos\left(\frac{\alpha_1+\alpha_2}{2}\right)-i\sin\left(\frac{\alpha_1-\alpha_2}{2}\right)
\end{gather*}
\begin{itemize}
  \item If $z'=0$ then  $\beta_1\defeq2\arg(z)$, $\beta_2\defeq0$ and $\beta_3\defeq0$.
  \item If $z=0$ then  $\beta_1\defeq2\arg(z')$, $\beta_2\defeq\pi$ and $\beta_3\defeq0$.
  \item Otherwise  $\beta_1\defeq\arg(z)+\arg(z')$, $\beta_2\defeq2\arg\left(i+\left\lvert\frac{z}{z'}\right\rvert\right)$ and $\beta_3\defeq\arg(z)-\arg(z')$.
\end{itemize}
\end{lemma}
\begin{proof}
 Without loss of generality assume $g(a)=\{b\}$, the other case is symmetric.
 Define
 \[
  g'(v) :=
  \begin{cases}
   \{x\} &\text{if } v = a \\
   \{y\} &\text{if } v = x \\
   \{z\} &\text{if } v = y \\
   \{b\} &\text{if } v = z \\
   g(v) &\text{otherwise.}
  \end{cases}
 \]
 Then \ref{itm:cXY} and \ref{itm:cYZ} are satisfied for all non-output vertices.
 It remains to show that $g'$ induces a strict partial order.

 Denote $R_X,R_Z$ the relations induced by the original gflow $g$ according to \eqref{eq:prec-X} and \eqref{eq:prec-Z}, and denote $R_X',R_Z'$ the corresponding relations induced by $g'$.
 ($R_Y$ is empty in both cases, as there are no $Y$ measurements.)
 The correction set of $a$ is the only original correction set to change, therefore:
 \begin{align*}
  R_X' &= (R_X\setminus\{(a,b)\}) \cup \{(a,x), (x,y), (y,z), (z,b)\} \\
  R_Z' &= (R_Z\setminus\{(a,v)\mid v\in N_G(b)\setminus\{a\}\}) \cup \{(z,v)\mid v\in N_G(b)\setminus\{a\}\} \cup \{(u,x)\mid a\in g(u)\} \\ &\qquad \cup \{(a,y), (x,z), (y,b)\}
 \end{align*}
 Note that in $\{(u,x)\mid a\in g(u)\}$, such a $u$ may or may not exist.

 The transitive closure $R_{g'}^* = (R_X'\cup R_Z')^*$ is equal to $R_{g}^*$ with the new vertices $x,y,z$ inserted between $a$ and $b$.
 To see this, indicate a relationship $v R_{g'} w$ by an arrow from $v$ to $w$; then, ignoring vertices whose relationships remain entirely unchanged, we have the following diagram of $R_{g'}$:
 \ctikzfig{euler-order-causal}
 Here, the dashed arrow denotes a relationship $a R_Z v$ in the original partial order, which is removed and replaced by the relationship $z R_{Z'} v$.
 Thus $R_{g'}^*$ is a strict partial order and $g'$ is an extended causal flow.
\end{proof}

\begin{lemma}[Phase fusion]\label{lem:phase-fusion}
	Suppose $D$ is an MBQC-form diagram that contains a $YZ$-measured vertex with a single $XY$-measured neighbour, and let $D'$ be the same diagram with the phase angles redistributed arbitrarily between the two measurements: i.e.\ if the initial angles are $\alpha$ and $\beta$, the new angles may be $\alpha'$ and $\alpha-\alpha'+\beta$ for any $\alpha'$.
	Then the two diagrams represent the same linear map and $D'$ has Pauli flow if and only if $D$ has Pauli flow.
	The same holds for gflow or extended causal flow.
	\begin{equation}
		\input{rules/phase-fusion-lhs.tikz} \quad\rightleftarrows\quad \input{rules/phase-fusion-rhs.tikz} \tag{PF}
	\end{equation}
\end{lemma}
\begin{proof}
	The two diagrams have the same structure, so it is immediate that any type of flow is preserved in both directions.
	The preservation of the interpretation is straightforward using standard ZX-calculus rewrite rules.
\end{proof}

Finally, we will also require a `spider nest' or `phase gadget' rule, which allows the removal of a certain pattern of phase gadgets that together implement a trivial transformation.
In this case, we will only need the deletion direction of the rule.

\begin{definition}
	A \emph{spider nest} on $n$ qubits consists of a family of phase gadgets such that each phase gadget is connected to a different non-empty subset of the $n$ qubits.
	Identifying the set of qubits with $[n]:=\{1,2,\ldots,n\}$, we may associate each phase gadget with an $n$-bit string where a 1 at position $k$ means the gadget is connected to qubit $k$ and a 0 at position $k$ means the gadget is not connected to qubit $k$.
	The phase angle of the phase gadget associated with bit string $s$ (if it exists) is denoted $\alpha_s$.
\end{definition}

There are $2^n-1$ non-empty subsets of a set of size $n$, hence this is the maximum number of phase gadgets in a spider nest.
	
\begin{definition}\label{def:full-spider-nest}
	The \emph{full spider nest on $n$ qubits with phase $\varphi$} consists of the $2^n-1$ phase gadgets with phases $\alpha_s = (-1)^{\abs{s}-1} \frac{\varphi}{2^{n-1}}$ for all $s\in\{0,1\}^n\setminus\{00\ldots 0\}$, where $\abs{s}$ denotes the Hamming weight, i.e.\ the number of 1s in the bit string $s$.
	Graphically,
	\ctikzfig{phase-gadgets/gadget-arity-n}
\end{definition}

\begin{example}\label{ex:spider-nests}
	The full spider nests on one, two, and three qubits, each with phase $2\pi$, are:
	\[
		\input{phase-gadgets/gadget-arity-1.tikz} \qquad\qquad\qquad
		\input{phase-gadgets/gadget-arity-2.tikz} \qquad\qquad\qquad
		\input{phase-gadgets/gadget-arity-3.tikz}
	\]
\end{example}

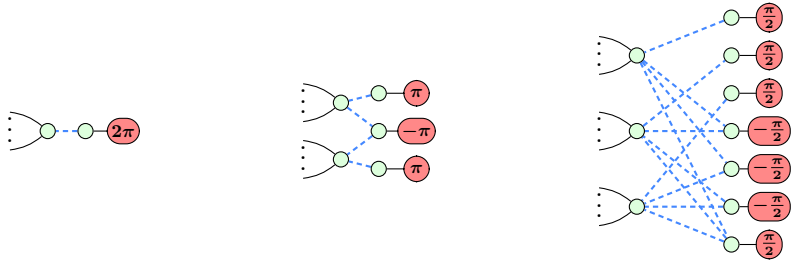
\begin{lemma}[Spider nest rule]\label{lem:spider-nest}
	Suppose $D$ is an MBQC-form diagram that contains a full spider nest on $n$ qubits with phase $2\pi$ and let $D'$ be the diagram without that spider nest.
	Then the two diagrams represent the same linear operation and $D'$ has extended causal flow whenever $D$ has extended causal flow.
	Similarly, $D'$ has gflow whenever $D$ has gflow, and $D'$ has Pauli flow whenever $D$ has Pauli flow.
	Graphically, with $\alpha_s = (-1)^{\abs{s}-1} \frac{2\pi}{2^{n-1}}$ for all $s\in\{0,1\}^n\setminus\{00\ldots 0\}$:
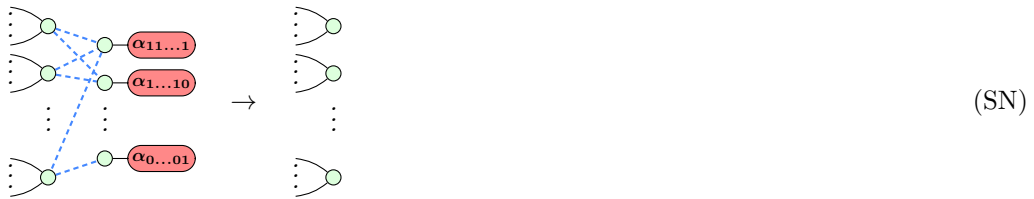
	\begin{equation}
		\input{phase-gadgets/gadget-arity-n.tikz}
		\quad\to\quad \input{phase-gadgets/gadget-arity-n-deleted.tikz} \tag{SN}
	\end{equation}
\end{lemma}
\begin{proof}
	The flow preserving aspect of this rule follows from $Z$-like deletion.
	The fact that the two diagrams represent the same linear operation follows from the Fourier transform relationship between phase gadgets and multi-controlled phase gates: the full spider nest on $n$ qubits corresponds to a single $n$-qubit controlled phase gate of phase $2\pi$, cf.~\cite[Proposition~3.6]{kuijpers_graphical_2019}, and this multi-controlled phase gate is trivial.
\end{proof}

\subsection{Derived rewrite rules}
\label{s:derived-rules}

We now have all the flow-preserving rewrite rules we will need; they are summarised in Figure~\ref{fig:flow-preserving}.
Rules \eqref{eq:PF} and \eqref{eq:LC} are stated as directed rules because they provide their own inverses: in the former case this is by choosing $\alpha'$ suitably, in the latter case it is because repeating the rule four times implements an overall identity transformation.

Note that \eqref{eq:EU-ZX} has been proved specifically in the context of (extended) causal flow, which indeed is the only situation where it will be needed.
Conversely, rule \eqref{eq:LC} preserves extended causal flow only in special cases, namely if the top vertex is $YZ$-measured and the other vertices are $XY$-measured, meaning the measurement labels do not change.
Yet \eqref{eq:LC} can also be used as part of more general sequences of transformations that preserve extended causal flow globally, albeit not at each step.

\begin{lemma}\label{lem:LC-causal}
	Suppose $(G,I,O,\ld)$ has extended causal flow $g$, $z\in\comp{O}$ is $YZ$-measured, and all its neighbours are $XY$-measured (not outputs).
	Then applying \eqref{eq:LC} with $z$ as the top vertex preserves extended causal flow.
\end{lemma}
\begin{proof}
	We necessarily have $z\prec_g v$ for all $v\in N_G(z)$ and we have $u\prec_g z$ for all $u\in\comp{O}$ such that $g(u)\in N_G(z)$.
	Denote $g^{-1}(N_G(z)) := \{u\in\comp{O}\mid g(u)\in N_G(z)\}$ the set of predecessors of neighbours of $z$.
	Then for all $v\in N_G(z)$ and for all $u\in g^{-1}(N_G(z))$, we have $\neg(v\prec_g u \vee v = u)$, cf.\ also \cite[Theorem~3.7]{backens_inserting_2025}.
	In particular, for all $w\in N_G(z)$ we have $g(w)\notin N_G(z)$.
	
	The local complementation about $z$ will not change any measurement labels.
	Moreover, suppose the local complementation removes an edge between $v,v'\in N_G(z)$.
	This cannot cause problems with the partial order but it could potentially break \ref{itm:cXY} if $g(v)=\{v'\}$ (or conversely).
	Yet we noted that if $w\in N_G(z)$, then $g(w)\notin N_G(z)$, so there is no problem.
	Conversely, suppose the local complementation introduces an edge between $v,v'\in N_G(z)$.
	This can cause problems with the partial order if there exists $u$ such that $g(u)=\{v\}$ and $v'\prec_c u$ because in this situation the new edge would imply $u\prec_g v'$.
	Yet this $u$ would be in $g^{-1}(N_G(z))$ and we know $\neg(v'\prec_g u \vee v' = u)$, so again there is no problem.
	
	Thus extended causal flow is preserved.
\end{proof}

We will now derive some more complex transformations that follow from the existing rules, as well as noting some particularly useful cases of those rules that apply only conditionally.

\begin{figure}
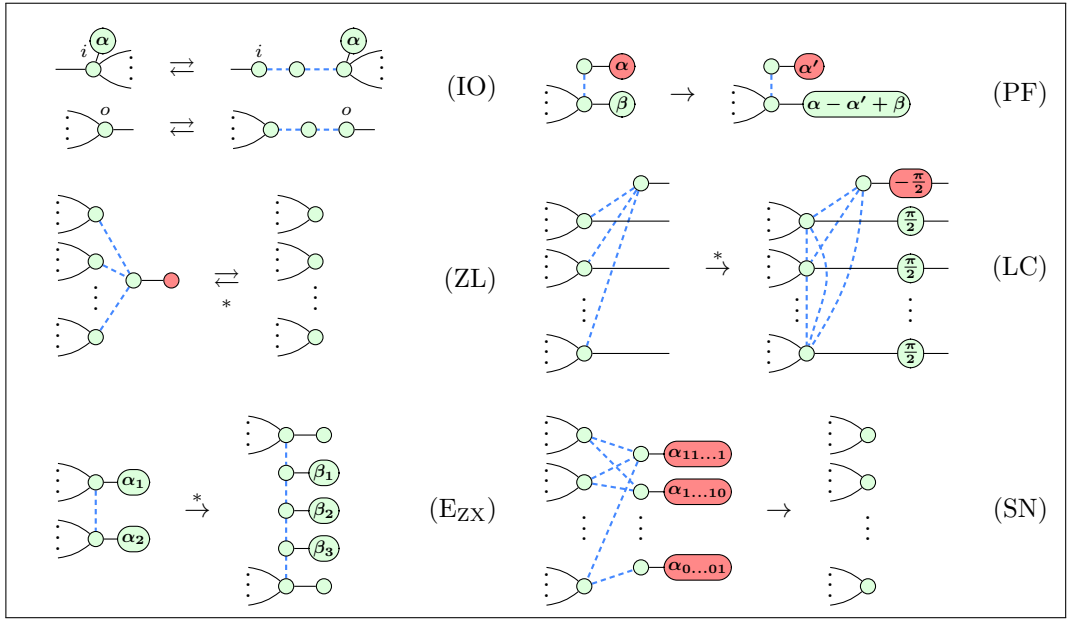

	\fbox{\begin{minipage}{.985\textwidth}
			\vspace{-.5em}
			\begin{subfigure}[c]{.46\textwidth}
				\begin{equation}\label{eq:IO}\tag{IO}
					\begin{aligned}
						\input{rules/input-splitting-lhs.tikz} \quad&\rightleftarrows\quad \input{rules/input-splitting-rhs.tikz} \\
						\input{rules/output-splitting-lhs.tikz} \quad&\rightleftarrows\quad \input{rules/output-splitting-rhs.tikz}
					\end{aligned}
				\end{equation}
			\end{subfigure}
			\begin{subfigure}[c]{0.52\textwidth}
				\begin{align}\label{eq:PF}\tag{PF}
					\input{rules/phase-fusion-lhs.tikz} \quad\to\quad \input{rules/phase-fusion-rhs.tikz}
				\end{align}
			\end{subfigure}
			\vspace{-.5em}
			
			\begin{subfigure}[c]{0.46\textwidth}
				\begin{align}\label{eq:ZL}\tag{ZL}
					\input{rules/Z-like-rhs.tikz} \quad\underset{*}{\rightleftarrows}\quad \input{rules/Z-like-lhs.tikz}
				\end{align}
			\end{subfigure}
			\begin{subfigure}[c]{0.52\textwidth}
				\begin{align}\label{eq:LC}\tag{LC}
					\input{rules/LC-lhs.tikz} \quad\overset{*}{\to}\quad \input{rules/LC-rhs.tikz}
				\end{align}
			\end{subfigure}
			
			\begin{subfigure}[c]{0.46\textwidth}
				\begin{align}\label{eq:EU-ZX}\tag{E\textsubscript{ZX}}
					\input{rules/Euler-display-lhs.tikz} \quad\overset{*}{\to}\quad \input{rules/Euler-display-rhs.tikz}
				\end{align}
			\end{subfigure}
			\begin{subfigure}[c]{0.52\textwidth}
				\begin{align}\label{eq:SN}\tag{SN}
					\input{phase-gadgets/gadget-arity-n.tikz}
					\quad\to\quad \input{phase-gadgets/gadget-arity-n-deleted.tikz}
				\end{align}
			\end{subfigure} \\
	\end{minipage}}
	\caption{The complete flow-preserving rule set.
		The vertices marked with $i$ in \eqref{eq:IO} are inputs and the vertices marked with $o$ are outputs.
		In \eqref{eq:LC}, none of the vertices may be outputs and the top vertex cannot be an input.
		While \eqref{eq:LC} always preserves gflow and Pauli flow, it preserves extended causal flow only in special cases.
		The right-to-left direction of rule \eqref{eq:ZL}, can be used if either the new vertex is assigned a measurement label $Z$ or if it is assigned $YZ$ and there exists a flow on the original diagram satisfying certain conditions, see Lemma~\ref{lem:Z-insertion-flow} and Theorem~\ref{thm:YZ-insertion-preserves-gflow}.
		Similarly, by Lemma~\ref{lem:Euler-rule}, rule \eqref{eq:EU-ZX} requires a specific kind of flow on the original diagram to be flow-preserving.
		The relationship between the phase angles on the left- and right-hand side of \eqref{eq:EU-ZX} is given in Lemma~\ref{lem:Euler-rule}.
		In \eqref{eq:SN}, the phase gadgets on the left-hand side must form a full spider nest with phase $2\pi$, cf.\ Definition~\ref{def:full-spider-nest}.
		The remaining rule directions apply unconditionally for any type of flow.}
	\label{fig:flow-preserving}
\end{figure}

The following lemma will be useful in various situations where we modify a flow in such a way as to add new relationships to the partial order.

\begin{lemma}\label{lem:extend-partial-order}
Let $\prec$ be a strict partial order on $V$ and let $T\sse V$ be upwards closed with respect to $\prec$, \ie for all $t\in T$, if there exists $v\in V$ such that $t\prec v$, then $v\in T$.
\begin{itemize}
	\item Suppose $R\sse \comp{T}\times T$ and $\prec'$ is the transitive closure of ${\prec}\cup R$.
	Then $\prec'$ is a strict partial order on $V$.
	\item Suppose $V'=V\uplus\{z\}$, $P\sse \comp{T}\times\{z\}$, $R\sse\{z\}\times T$.
	Let $\prec'$ be the transitive closure of ${\prec}\cup P\cup R$, then $\prec'$ is a strict partial order on $V'$.
\end{itemize}
\end{lemma}
\begin{proof}
\textbf{Part 1}: Suppose $R\sse \comp{T}\times T$ and let $\prec'$ be the transitive closure of ${\prec}\cup R$.
Assume for a contradiction that $\prec'$ is not a strict partial order, i.e.\ there exists a cycle $v\prec' v$ for some $v\in V$.
Then there exists a sequence $v_0,\ldots,v_k$ with $v_0 = v = v_k$, and such that for all $0\leq j < k$ we have $v_j \prec v_{j+1} \vee v_j R v_{j+1}$.

Let $j^*$ be the smallest $j$ such that $v_j R v_{j+1}$; this exists as otherwise $v\prec v$, contradicting the assumption that $\prec$ is strict.
We have $v_{j^*}\in\comp{T}$ and $v_{j^*+1}\in T$ by the definition of $R$.

Then the relationship between $v_{j^*+1}$ and $v_{j^*+2}$ must be $v_{j^*+1} \prec v_{j^*+2}$ since $v_{j^*+1}\in T$ cannot appear on the left-hand side of a relation from $R$.
Thus $v_{j^*+2}\in T$ by upwards closure.
The argument can be repeated so we have $v_j\in T$ for all $j\geq j^*+1$.

Similarly, the relationship between $v_{j^*-1}$ and $v_{j^*}$ must be $v_{j^*-1} \prec v_{j^*}$ since $v_{j^*}\in\comp{T}$ cannot appear on the right-hand side of a relation from $R$.
Thus, $v_j\in\comp{T}$ for all $j\leq j^*$.
This implies $v=v_0\in\comp{T}$ and $v=v_k\in T$, a contradiction.
Thus the transitive closure of ${\prec}\cup R$ is a strict partial order.

\textbf{Part 2}: Now, suppose $V'=V\uplus\{z\}$, $P\sse \comp{T}\times\{z\}$, $R\sse\{z\}\times T$.

First, let $\prec_1$ be the extension of $\prec$ to domain $V'$ which does not add any new relationships.
Define $T' = T\cup\{z\}$ and notice this is upwards closed with respect to $\prec_1$ since $z$ is not related to anything.
Let $\prec_2$ be the transitive closure of ${\prec_1}\cup P$; this is a strict partial order by part 1 of the lemma.

Secondly, note that the original set $T$ is upwards closed with respect to $\prec_2$ as all predecessors of $z$ are in $\comp{T}$.
Let $\prec_3$ be the transitive closure of ${\prec_2}\cup R$; then this is a strict partial order by part 1 of the lemma.

Finally, denote the transitive closure operation by $\operatorname{TC}$, then it is easy to see that $\operatorname{TC}(\operatorname{TC}({\prec}\cup P) \cup R) = \operatorname{TC}({\prec}\cup P \cup R)$, so ${\prec_3} = {\prec'}$ and we have the desired result.
\end{proof}

By combining local complementation or pivoting with $Z$-insertion and deletion, we can derive the following two rules transforming graph-like ZX-diagrams to graph-like ones.
The `deletion' directions of these rules were first introduced in the context of gflow in \cite[Lemmas~5.2 and~5.3]{duncan_graph-theoretic_2020}.

First, we will derive a simpler change of the gflow under local complementation than the expression in \cite[Lemma~3.1]{backens_there_2021} (corresponding to the version given if the central vertex is an output \cite[Lemma~3.2]{backens_there_2021}). It is also the qubit equivalent of the qudit version in \cite[Proposition~5.8]{booth_measurement-based_2022} as pointed out in \cite{mitosek_working_2026}.

\begin{lemma}\label{lem:LC-gflow-change}
	Let $\Gamma = (G,I,O,\ld)$ be a \LOG with gflow $g$, and let $z\in \comp{O}\cap\comp{I}$ be an internal vertex.
	Let $\Gamma' = (G*z,I,O,\ld')$ be the \LOG after a local complementation about $z$.
	Then $\Gamma'$ has a gflow given by
	\[
		g'(v) = \begin{cases}
			g(v) &\text{if } z\notin\odd{G}{g(v)} \\
			g(v)\symd\{a\} &\text{if } z\in\odd{G}{g(v)}
		\end{cases}
	\]
	and ${\prec_{g'}} = {\prec_g}$.
\end{lemma}
\begin{proof}
	For any $A\sse V$ \cite[Lemma~B.4]{duncan_graph-theoretic_2020}, we have:
	\[
	\odd{G*z}{A} = \begin{cases}
		\odd{G}{A}\symd (N_G(z)\cap A) &\text{if } z\notin\odd{G}{A} \\
		\odd{G}{A}\symd (N_G(z)\setminus A) &\text{if } z\in\odd{G}{A}
	\end{cases}
	\]
	Thus:
	\[
	\odd{G*z}{g'(v)} = \begin{cases}
		\odd{G}{g(v)} \symd (N_G(z)\cap g(v)) &\text{if } z\notin\odd{G}{g(v)} \\
		\odd{G*z}{g(v)}\symd N_{G*z}(z) & \\\quad = \odd{G}{g(v)}\symd (N_G(z)\setminus g(v)) \symd N_G(z)  &\text{if } z\in\odd{G}{g(v)}
	\end{cases}
	\]
	so in fact $\odd{G*z}{g'(v)} = \odd{G}{g(v)} \symd (N_G(z)\cap g(v))$ for all $v$.
	This is compatible with the partial order for the following reasons:
	\begin{itemize}
		\item $z$ is added to the correction set of $v$ only if $z\in\odd{G}{g(v)}$, which implies $v=z$ or $v\prec_g z$.
		It is the only vertex that may be added to correction sets.
		\item A vertex $w$ may be added to the odd neighbourhood of the correction set of $v$ only if $w\in N_G(z)\cap g(v)$.
		Yet this implies in particular $w\in g(v)$ and thus again $v=w\vee v\prec_g w$.
	\end{itemize}
	The argument for compatibility with the flow conditions is analogous to that in \cite[Lemma~3.1]{backens_there_2021}.
\end{proof}

\begin{lemma}[Local complementation \& delete, insert \& local complementation]\label{lem:lc-delete}
 The following operation on graph-like ZX-diagrams preserves Pauli flow in both directions:
 \begin{equation}\label{eq:LCD}
  \input{rules/LC-del-red-lhs.tikz} \quad\rightleftarrows\quad\input{rules/LC-del-red-rhs.tikz} \qquad\qquad\qquad \input{rules/LC-del-green-lhs.tikz} \quad\rightleftarrows\quad\input{rules/LC-del-green-rhs.tikz} \tag{LCD}
 \end{equation}
 It preserves gflow unconditionally in the left-to-right direction; in the reverse direction gflow is preserved if and only if the insertion step preserves it.
 Extended causal flow is preserved left-to-right if all the vertices receiving a \begin{tikzpicture}
	\begin{pgfonlayer}{nodelayer}
		\node [style=Z phase dot] (0) at (0, 0) {$\frac{\pi}{2}$};
		\node [style=none] (1) at (-0.75, 0) {};
		\node [style=none] (2) at (0.75, 0) {};
	\end{pgfonlayer}
	\begin{pgfonlayer}{edgelayer}
		\draw (1.center) to (2.center);
	\end{pgfonlayer}
\end{tikzpicture}
 are measured $XY$.
 
 For the right-to-left direction, extended causal flow is preserved if
 \begin{itemize}
 	\item The conditions for gflow preservation of Theorem~\ref{thm:YZ-insertion-preserves-gflow} are satisfied during the insertion step given the original extended causal flow $g$.
 	 Moreover $\abs{C}=0$ for the $YZ$-insertion case or $\abs{C}=1$ for the $XZ$-insertion case, and for all $w\in\comp{O}$ such that $g(w)\sse N\setminus C$ and for all $v\in N$, we have\footnote{Given the assumption on $C$, this reduces to one of the conditions of Theorem~\ref{thm:YZ-insertion-preserves-gflow} if the inserted vertex is $YZ$-measured but it is stronger than the corresponding condition of the theorem for an $XZ$-insertion.} $\neg(v\prec_g w \vee v=w)$.
 	\item All the vertices receiving a \begin{tikzpicture}
	\begin{pgfonlayer}{nodelayer}
		\node [style=Z phase dot] (0) at (0, 0) {$\frac{\pi}{2}$};
		\node [style=none] (1) at (-0.75, 0) {};
		\node [style=none] (2) at (0.75, 0) {};
	\end{pgfonlayer}
	\begin{pgfonlayer}{edgelayer}
		\draw (1.center) to (2.center);
	\end{pgfonlayer}
\end{tikzpicture}
 are measured $XY$
 \end{itemize}
\end{lemma}
\begin{proof}
	We give the proofs for extended causal flow as the gflow versions follow immediately from the conditions for \eqref{eq:ZL}.
	
	First, consider the version where the vertex that is deleted or inserted is $YZ$-measured.
	Without loss of generality, assume the neighbours are all $XY$-measured (they cannot be $YZ$-measured by the definition of extended causal flow; if any are outputs, apply \eqref{eq:IO}).
	The transformation can be decomposed as follows:
	\[
		\input{rules/LC-del-red-lhs.tikz}
		\quad\overset{\eqref{eq:LC}}{\rightleftarrows}\quad\input{rules/LC-del-red1.tikz}
		\quad\overset{\eqref{eq:ZL}}{\rightleftarrows}\quad\input{rules/LC-del-red-rhs.tikz}
	\]
	By Lemma~\ref{lem:LC-causal}, the local complementation preserves extended causal flow.
	The $YZ$-deletion also preserves flow unconditionally.
	This gives flow-preservation from left-to-right.
	
	As the local complementation is known to preserve extended causal flow under the given conditions, the right-to-left transformation preserves extended causal flow if and only if the $YZ$-insertion preserves it.
	This completes the case of a $YZ$-insertion or deletion.

	Now, consider the version where the vertex $z$ that is deleted or inserted is $XY$-measured, as are all its neighbours.
	The transformation can be decomposed as follows:
	\[
		\input{rules/LC-del-green-lhs.tikz}
		\quad\overset{\eqref{eq:LC}}{\rightleftarrows}\quad\input{rules/LC-del-green1.tikz}
		\quad\overset{\eqref{eq:ZL}}{\rightleftarrows}\quad\input{rules/LC-del-green-rhs.tikz}
	\]
	In this case, the local complementation changes the measurement plane of $z$ to $\ld'(z) = XZ$, so it is not possible to work entirely inside extended causal flow.
		
	For the left-to-right direction, gflow preservation is immediate from \eqref{eq:LC} and $Z$-like deletion.
	By Lemma~\ref{lem:LC-gflow-change}, the local complementation changes the gflow to:
	\[
		g'(v) = \begin{cases}
			g(v) &\text{if } z\notin\odd{G}{g(v)} \\
			g(v)\symd\{z\} &\text{if } z\in\odd{G}{g(v)}
		\end{cases}
	\]
	which leaves the partial order invariant.
	The $XZ$-deletion changes the gflow to \cite[Lemma~3.4]{backens_there_2021}:
	\begin{align*}
		g''(v)
		&= \begin{cases}
			g'(v)\symd g'(z) & \text{if } z\in g'(v) \\
			g'(v) &\text{if } z\notin g'(v)
		\end{cases} \\
		&= \begin{cases}
		g(v)\symd g(z) \symd\{z\} & \text{if } z\in g(v) \wedge z\notin\odd{G}{g(v)} \\
		g(v)\symd g(z) &\text{if } z\notin g(v) \wedge z\in\odd{G}{g(v)} \\
		g(v)\symd \{z\} & \text{if } z\in g(v) \wedge z\in\odd{G}{g(v)} \\
		g(v) &\text{if } z\notin g(v) \wedge z\notin\odd{G}{g(v)}
		\end{cases}
	\end{align*}
	and the only change to the partial order is the removal of the deleted vertex.
	
	To adapt this result to extended causal flow, assume all correction sets in the initial flow have size 1 and that all neighbours of $z$ are $XY$-measured.
	Then, note that in the definition of $g''$ in terms of $g$, the third case cannot happen if the original flow is extended causal because there are no self-loops.
	The first case is unproblematic since $z\in g(v)$ implies $g(v)=\{z\}$ and thus $g(v)\symd g(z)\symd\{z\} = g(z)$ has size 1.
	The last case is similarly fine.
	It remains to consider the second case, where $z\notin g(v)$ but $z\in\odd{G}{g(v)}$, \ie the unique element of $g(v)$ is a neighbour of $z$.
	As all neighbours of $z$ are $XY$-measured, we have $v\notin g(v)$.
	Consider the odd neighbourhood that results if we did not change the correction set of $v$:
	\[
		\odd{G''}{g(v)} = \odd{G'}{g(v)}\setminus\{z\} = (\odd{G}{g(v)}\symd (N_G(z) \cap g(v) ))\setminus\{z\}
	\]
	Since $v\neq z$ and $v\notin g(v)$, we have $v\in\odd{G''}{g(v)}$ and \ref{itm:gXY} holds.
	Moreover, $w\in\odd{G''}{g(v)}\setminus\{v\}$ implies $w\in\odd{G}{g(v)} \vee w\in g(v)$, so $v\prec_g w$.
	Hence there are no problems with the partial order either and extended causal flow is preserved whenever all the neighbours are $XY$-measured.

	Now consider the condition that a flow on the right-hand diagram needs to satisfy to make the $XZ$-insertion gflow-preserving.
	Consider a \LOG $\Gamma = (G,I,O,\ld)$ which has an extended causal flow $g$ that we will treat as a gflow.
	Let $N\sse\{v\in\comp{O}\mid \ld(v)=XY\}$.
	By Theorem~\ref{thm:YZ-insertion-preserves-gflow}, an $XZ$-measurement with neighbours $N$ can be inserted in a flow-preserving way if there exists a set $C\sse\comp{I}$ such that $\abs{N\cap C}\equiv 1\bmod 2$ and for all $v\in\comp{O}$ such that $g(v)\sse N$ and for all $w\in C\cup(N\symd\odd{G}{C})$, we have $\neg(w\prec_g v \vee w=v)$.
	In this case, the correction function after the insertion is
	\[
		g'(u) = \begin{cases}
			C\cup\{z\} &\text{if } u = z \\
			g(u) &\text{otherwise,}
		\end{cases}
	\]
	and the corresponding partial order is the transitive closure of
	\[
	 {\prec_g} \cup \{(w,z)\mid g(w)\sse N\} \cup \{(z,v)\mid v\in C\cup(N\symd\odd{G}{C})\}.
	\]
	By Lemma~\ref{lem:LC-gflow-change}, the correction function after the local complementation becomes :
	\[
		g''(u) = \begin{cases}
			C &\text{if } u = z \\
			g(u)\symd\{z\} &\text{if } g(u)\sse N \\
			g(u) &\text{otherwise}.
		\end{cases}
	\]
	Suppose $\abs{C}=1$, then if $g$ was causal, the only obstacle that remains is the second case of the definition of $g''$.
	If $g(u)\in N$, then we have $u\prec_{g'} z$ as $g'$ is a gflow and $z\in\odd{G'}{g(u)}$.
	\begin{itemize}
		\item Suppose $u\in\comp{O}$ satisfies $g(u)=C$, then $g(u)\sse N$ since \ref{itm:XZ} holds for $g'(z)$.
			 Then we must have $u\sim z$ as otherwise $u\in\odd{G'}{g(z)}$, which would contradict the partial order.
			 Moreover, we have $u\prec_{g'} v$ for all $v\in\odd{G}{C}$ and $z\prec_{g'}v $ for all $v\in N\symd\odd{G}{C}$.
			 Thus, by transitivity, $u\prec_{g'} v$ for all $v\in N$.
			 Setting $g''(u)=\{z\}$ thus satisfies \ref{itm:cXY} as well as the partial order conditions.
		\item Suppose $u\in\comp{O}$ satisfies $g(u)\sse N\setminus C$.
			 Now, $g(u)\in N$ implies $z\in\odd{G'}{g(u)}$, hence
			 \[
			 	\odd{G''}{g(u)} = \odd{G'}{g(u)} \symd (N\setminus g(u)).
			 \]
			 While we only have $z\prec_{g'} v$ for $v\in N\symd\odd{G}{C}$, we also know $\neg (v\prec_g u \vee v = u)$ for all $v\in N\symd\odd{G}{C}$.
			 Thus we may add relationships of the form $u \prec v$ without breaking the partial order.
			 This means $g(u)$ continues to be a valid correction set for $u$.
	\end{itemize}	
	This means the final labelled open graph has extended causal flow.
\end{proof}

\begin{lemma}[Pivot \& delete, insert \& pivot]\label{lem:pivot-delete}
 The following operation on graph-like ZX-diagrams -- a pivot on the edge $uv$ followed by deleting $u$ and $v$ -- preserves Pauli flow in both directions if $u$ and $v$ are $X$-measured:
 \begin{equation}\label{eq:PD}
  \input{rules/pivot-del-1.tikz} \quad \rightleftarrows \quad \input{rules/pivot-del-2.tikz} \tag{PD}
 \end{equation}
 Gflow is preserved unconditionally left-to-right and is preserved in the right-to-left direction if the insertions preserve the flow.
\end{lemma}

The left-to-right preservation of gflow was originally proved in \cite[Lemma B.8]{duncan_graph-theoretic_2020}; the other direction follows from decomposing the operation as two $YZ$-insertions and an (unconditionally gflow-preserving) pivot.
The pivot \& delete or insert \& pivot operation can also preserve causal flow, yet the conditions for this are too complex to work out in generality here.
Even when gflow is preserved by an application of Lemma~\ref{lem:pivot-delete}, extended causal flow need not be preserved.
This is demonstrated for the `pivot \& delete' direction by reading Example~\ref{ex:circuit-extraction} in reverse.
We will use this lemma in a causal flow context only in situations where it is straightforward to manually check that the causal flow is preserved.
A special case in which one of $u,v$ has exactly two neighbours will be analysed below in Proposition~\ref{prop:vertex-splitting}.

One potentially non-obvious consequence of the pivot-and-delete lemma is the ability to remove or insert pairs of edges similar to the Hopf law in the ZX-calculus.
The edges cannot share both endpoints as graph-like diagrams must be simple, but by `unfusing' one of the endpoints, we nevertheless get a similar effect.\footnote{Note that, while this example is very similar to fusing and then splitting a vertex, the `Hopf' effect cannot be achieved with the previously-given `vertex splitting' rule as it is \cite[Proposition~5.1]{mcelvanney_flow-preserving_2023}, though the proof of that proposition could straightforwardly be extended to allow this.}

\begin{example}
 The following Hopf-like rewriting process is flow-preserving in both directions by \eqref{eq:PD} applied to the two internal vertices and \eqref{eq:IO}:
 \ctikzfig{rules/hopf-ex}
\end{example}

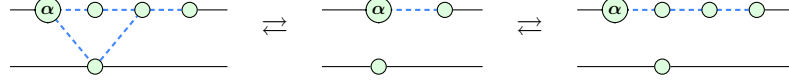
\begin{lemma}[Graph-like copy]\label{lem:graph-like-copy}
 The following rewrite rule on graph-like diagrams preserves interpretation (up to scalar factor) and Pauli flow in both directions if the vertex with phase $a\pi$ is $Z$-measured:
 \begin{equation}
  \input{rules/copy.tikz}
 \end{equation}
 It removes or inserts a degree-1 vertex with angle 0 or $\pi$ and its unique neighbour.
 The neighbour cannot be an input or output.
 Gflow preservation in the left-to-right direction is unconditional.
 If the vertex with phase $a\pi$ is $YZ$-measured, then gflow is preserved in the right-to-left direction if the insertions in the proof are flow-preserving.
\end{lemma}
\begin{proof}
 As before the proof goes via the MBQC-form.
 \ctikzfig{rules/copy-proof}
 It employs pivoting, $Z$-deletion or insertion and the property that the flow on separate tensor components is independent, so a disconnected scalar diagram can be dropped or added at will.
\end{proof}

Graph-like copy is, in a sense, a variant of the pivot-and-delete/insert-and-pivot rule: it allows one of the two vertices to have an arbitrary phase label, but to make up for this, the other vertex must have degree~1.

\begin{lemma}
	It suffices to use the spider nest rule for $n\geq 4$, the smaller cases can be derived from existing rules.
\end{lemma}
\begin{proof}
	From Example~\ref{ex:spider-nests}, the spider nests for $n\leq 3$ involve only phase angles that are integer multiples of $\frac\pi2$.
	The result thus follows from the completeness of flow-preserving rewriting for computations in the Clifford fragment \cite{mcelvanney_complete_2023}, which uses only rules \eqref{eq:IO}, \eqref{eq:LC}, and \eqref{eq:ZL}.
	
	Indeed the case $n=1$ uses $Z$-like deletion and the property that a phase of $2\pi$ is equivalent to $0$:
	 \begin{align*}
	 	\input{phase-gadgets/gadget-arity-1.tikz} \quad\overset{\eqref{eq:ZL}}&{\rightleftarrows}\quad 
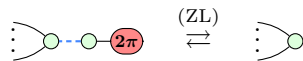
\begin{tikzpicture}
	\begin{pgfonlayer}{nodelayer}
		\node [style=Z dot] (0) at (0, 0) {};
		\node [style=none] (5) at (-1, 0.5) {};
		\node [style=none] (6) at (-1, 0.25) {$\vdots$};
		\node [style=none] (7) at (-1, -0.5) {};
	\end{pgfonlayer}
	\begin{pgfonlayer}{edgelayer}
		\draw [bend left=15] (5.center) to (0);
		\draw [bend left=15] (0) to (7.center);
	\end{pgfonlayer}
\end{tikzpicture}
 \\
	 	\intertext{The same holds for $n=2$:}
	 	\input{phase-gadgets/gadget-arity-2.tikz}
	 	\quad\overset{\eqref{eq:ZL}}&{\rightleftarrows}\quad \input{phase-gadgets/gadget-arity-2-del1.tikz}
	 	\quad\overset{\eqref{eq:ZL}}{\rightleftarrows}\quad \input{phase-gadgets/gadget-arity-2-del2.tikz}
	 	 \quad =\quad \input{phase-gadgets/gadget-arity-2-del3.tikz}
	 \end{align*}
	 The case $n=3$ uses the `local complementation \& delete' rule derived from \eqref{eq:LC} and \eqref{eq:ZL}:
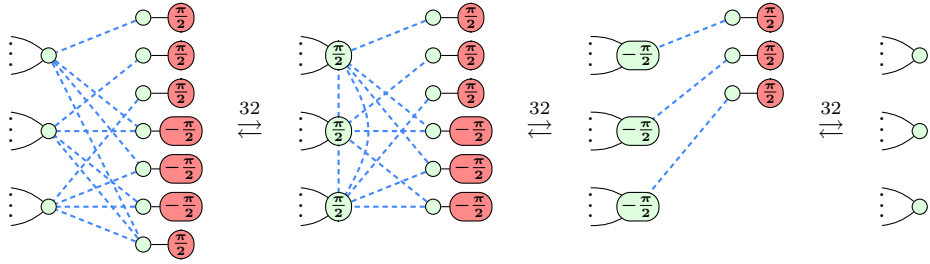
	 \[
	 	\input{phase-gadgets/gadget-arity-3.tikz}
	 	\quad\overset{\ref{lem:lc-delete}}{\rightleftarrows}\quad \input{phase-gadgets/gadget-arity-3-del1.tikz}
	 	\quad\overset{\ref{lem:lc-delete}}{\rightleftarrows}\quad \input{phase-gadgets/gadget-arity-3-del2.tikz}
	 	\quad\overset{\ref{lem:lc-delete}}{\rightleftarrows}\quad \input{phase-gadgets/gadget-arity-3-del3.tikz}
	 \]
	 Note for the second step that each vertex on the left is connected to two of the gadgets of phase $-\frac\pi2$, hence the phase changes by $2(-\frac\pi2) = -\pi$.
\end{proof}

We have the following corollary to Theorem~\ref{thm:YZ-insertion-preserves-gflow}, which sets out simplifications for certain cases of $YZ$-insertion.

\begin{corollary}\label{cor:YZ-insertion}
	Let $\Gamma = (G,I,O,\ld)$ be a \LOG with gflow $g$.
	Suppose $(D,U)$ is a partition of $V$ such that $O\sse U$ and $U$ is upwards closed according to $\prec_g$.
	Furthermore, suppose the two sets $N\sse V$ and $C\sse\comp{I}$ satisfy $\abs{N\cap C}\equiv 0 \bmod 2$, as well as $\Pred(N) \sse D$ and $\Succ(N,C) \sse U$, where $\Pred(N)$ and $\Succ(N)$ are as defined in \eqref{eq:pred} and \eqref{eq:succ}.
	Define $\Gamma'=(G',I,O,\ld')$ as in Definition~\ref{def:insert-Z}, then the function $g'$ defined in \eqref{eq:yz-insert-correction-function} is a gflow on $\Gamma'$.

	We will use several particular cases of this corollary.
	\begin{enumerate}
		\item\label{it:only-outputs} Inserting a $YZ$-measurement with only output neighbours, i.e.\ $N\sse O$ and $C=\emptyset$, is always flow-preserving.
		 
		\item\label{it:single-neighbour} Inserting a new $YZ$-measurement with a single, $XY$-measured neighbour $N=\{a\}$ preserves flow with $C=\emptyset$.
		 
		\item\label{it:single-neighbour-focused} Inserting a new $YZ$-measurement with a single, $XY$-measured neighbour $N=\{a\}$ preserves flow with $C=g(a)$.
		 
		\item\label{it:parallel-gadgets} Inserting a new $YZ$-measurement with the same neighbours as a pre-existing $YZ$-measurement, i.e.\ $N=\{a\}$ and $C=g(a)\setminus\{a\}$, is flow-preserving.
	\end{enumerate}
	In all cases where $C=\emptyset$, extended causal flow is preserved if the original gflow is causal.
\end{corollary}
\begin{proof}
	We consider each case in turn.
	\begin{enumerate}
		\item Suppose $N\sse O$ and $C=\emptyset$, then $\abs{N\cap C}=0$ trivially.
		We may take $U := O$; this is trivially upwards closed as the outputs do not have successors in the partial order.
		By the same reasoning, $\Pred(N)\sse D = V\setminus U$ .
		Finally, $\Succ(N,\emptyset) = N \sse O = U$, so the corollary applies.
		
		\item Suppose there exists an $XY$-measured vertex $a\in V$ such that $N=\{a\}$, and $C=\emptyset$, then again $\abs{N\cap C}=0$ trivially.
		Take $U := O \cup \{a\} \cup \{v\in V\mid a\prec_g v\}$, then $U$ is upwards closed.
		We have $\Pred(N) = \{v\in V\setminus O \mid a\in g(v)\}\sse D = V\setminus U$ because $v\prec_g a$ for all $v\in\Pred(N)$.
		Finally, $\Succ(N,\emptyset) = \{a\} \sse U$, so the corollary applies.
		
		\item Suppose there exists an $XY$-measured vertex $a\in V$ such that $N=\{a\}$, and $C=g(a)$, then $\abs{N\cap C}=0$ by \ref{itm:XY}.
		Take $U := O \cup \{v\in V\mid a\prec_g v\}$, then $U$ is upwards closed.
		We have $\Pred(N) = \{v\in V\setminus O \mid a\in g(v)\}\sse D = V\setminus U$ because $v\prec_g a$ for all $v\in\Pred(N)$.
		Finally, $\Succ(\{a\},g(a)) = g(a) \sse U$, so the corollary applies.
		
		\item Suppose there exists a $YZ$-measured vertex $a\in V$.
		Set $N=N_G(a)$ and $C=g(a)\setminus\{a\}$, then $\abs{N\cap C}=0$ by \ref{itm:YZ}.
		The partial order conditions all work out by inserting the new vertex $z$ parallel to $a$, meaning with all the same predecessors and successors. \qedhere
	\end{enumerate}
\end{proof}

Variants of the following `vertex fusion and splitting' rule (sometimes also called `neighbour unfusion') have been considered in various contexts and for different types of flow \cite{staudacher_reducing_2023,holker_causal_2023,mcelvanney_flow-preserving_2023,perez_measurement-based_2023}; we show here how it follows from the rewrite rules of Figure~\ref{fig:flow-preserving} and give conditions for the preservation of extended causal flow specifically.
Without the angle splitting aspect, i.e.\ if $\beta=0$ in the below diagram, this result is a special case of Lemma~\ref{lem:pivot-delete}.

\begin{proposition}[Vertex splitting and fusion]\label{prop:vertex-splitting}
	Let $\Gamma = (G,I,O,\ld)$ be a \LOG with extended causal flow $g$.
	Suppose $a\in\comp{O}$ satisfies $\ld(a)=XY$ and let $N\sse \{v\in N_G(a)\mid \neg(v\prec_g a)\}$ be such that $g(a)\in N$ and furthermore for all $v\in N$ and for all $w\in N_G(a)$, we have $v\prec_g w\implies w\in N$: \ie $N$ is upwards closed within $N_G(a)$.
	Then the \LOG in which vertex $a$ has been split as below also has extended causal flow:
	\[
		\input{rules/neighbour-unfusion-lhs.tikz} \quad\rightleftarrows\quad \input{rules/neighbour-unfusion-rhs.tikz}
	\]
	Flow is also preserved in the opposite direction.
	
	Moreover, if $a$ has a predecessor in the extended causal flow -- \ie there exists $b\in\comp{O}$ such that $g(b)=\{a\}$, then an analogous result holds with $N$ being a downwards-closed subset of $\{v\in N_G(a)\mid \neg(a\prec_g v)\}$ that contains $b$.
\end{proposition}
\begin{proof}
	While this lemma is about extended causal flow, the proof actually goes via gflow.
	The transformations proceed as follows:
	\[
	 \input{rules/neighbour-unfusion0.tikz}
	 \;\overset{\eqref{eq:ZL}}{\rightleftarrows}\; \input{rules/neighbour-unfusion1.tikz}
	 \;\overset{\eqref{eq:PF}}{\rightleftarrows}\; \input{rules/neighbour-unfusion2.tikz}
	 \;\overset{\eqref{eq:ZL}}{\rightleftarrows}\; \input{rules/neighbour-unfusion3.tikz}
	 \;\overset{\ref{lem:pivot-Pauli-flow}}{\rightleftarrows}\; \input{rules/neighbour-unfusion4.tikz}
	\]
	where the fourth diagram cannot have extended causal flow as it contains two adjacent $YZ$-measurements.
	
	In the left-to-right direction, the first step is to insert a new $YZ$-measured vertex $z$ whose only neighbour is $a$.
	Denote the resulting \LOG by $\Gamma'$.
	The insertion preserves extended causal flow by Case~\ref{it:single-neighbour} of Corollary~\ref{cor:YZ-insertion} (as $C=\emptyset$ implies the new vertex has a correction set of size~1 and correction sets of other vertices do not change).
	Yet it will be more useful to work with Case~\ref{it:single-neighbour-focused} instead, where the new correction function $g'$ on $\Gamma'$ is given by $g'(z) = \{z\}\cup g(a)$ and $g'(v)=g(v)$ for all original vertices $v\in\comp{O}$.
	Importantly, $z$ is parallel to $a$ in $\prec_g$: they have the same predecessors and the same successors, and are mutually incomparable.
	
	The second step splits the phase, which preserves gflow and indeed changes neither the \LOG nor the correction function.
	The third step inserts another new $YZ$-measured vertex $z'$ whose neighbours are $N\cup\{z\}$ (this would not preserve extended causal flow even if the first step had done so).
	Denote the \LOG after the insertion of $z'$ by $\Gamma''$ and the underlying graph by $G''$.
	Let $N':= N\cup\{z\}$ and $C=\emptyset$, and take $U:= O\cup \{v\in V\mid \exists w\in N'.w\prec_{g'} v\}$.
	Apply Corollary~\ref{cor:YZ-insertion} to the pair $N',C$ with partition $(V\setminus U, U)$.
	Since $g'(z)$ has size 2 and is contained in $N'$, and all other correction sets of $g'$ have size 1, we find:
	\begin{align*}
		\Pred(N') &= \{u\in V\setminus O \mid g(u)\sse N\} \\
		\Succ(N',\emptyset) &= N'
	\end{align*}
	Now, if the unique element of $g(u)$ is in $N$ for some $u$, then $u\prec_g a$ and thus $u\notin N$ and $u\notin U$.
	On the other hand, $N'\sse U$ by definition.
	Thus by Corollary~\ref{cor:YZ-insertion} there is a gflow on $\Gamma''$ satisfying $g''(z) = g(a)\cup\{z\}$, $g''(z') = \{z'\}$, and $g''(v)=g(v)$ otherwise.
	In the partial order $\prec_{g''}$, the new vertex $z'$ thus appears before any elements of $N'$.
	For any $u\in \Pred(N')$, we have $u\prec_{g} a$ since the unique element of $g(u)$ was originally a neighbour of $a$; hence $u\prec_{g''} z'$ only if $u\prec_g a$.
	
	Finally, pivot on the edge $z z'$, denote the resulting \LOG by $\Gamma'''$ and its underlying graph by $G'''$.
	The pivot preserves gflow by Lemma~\ref{lem:pivot-Pauli-flow}.
	Moreover, noting that $z,z'$ do not appear in correction sets of $g''$ other than their own, and $g''=g$ on the original vertices, the flow after the pivot according to  Lemma~\ref{lem:pivot-Pauli-flow} is given by:
	\[
		g'''(w) = \begin{cases}
			g(a) &\text{if } w = z \\
			\{z\} &\text{if } w = z' \\
			\{a, z\} &\text{if } g(w) = \{a\} \\
			g(w)\symd\{z'\} &\text{if } g(w)\in N \\
			g(w) &\text{otherwise.}
		\end{cases}
	\]
	Now if $g(w)\in N$, then we have $\odd{G'''}{g(w)\symd\{z'\}} = \odd{G}{g(w)}$ and $\odd{G'''}{g(w)} =\odd{G}{g(w)}\symd\{a,z\}$.
	Yet $g(w)\in N$ means $a\in\odd{G}{g(w)}$, so $w\prec_g a$ and thus $w\prec_{g''} z'\prec_{g''}z$.
	We can therefore replace $g'''(w)$ by $g(w)$ without breaking the partial order, except when $w=a$ for which we need to set the correction set to $\{z'\}$ as it is no longer adjacent to its original corrector.
	
	Similarly, if $g(w)=\{a\}$, then $\odd{G'''}{\{a,z\}} = N_G(a)$ and $\odd{G'''}{\{a\}} = (N_G(a)\setminus N)\cup\{z'\}$.
	Removing elements from the odd neighbourhood of a correction set cannot break the partial order and we already have $w\prec_g a \prec_{g'''} z'$, hence we can also drop $z$ from the correction set.
	Thus the final \LOG has extended causal flow.
	
	Flow-preservation in the reverse direction is straightforward as it involves deletion rather than insertion of $YZ$-measurements.
\end{proof}

\begin{corollary}
	The case $N = g(a)$ and $\beta=0$ of the above proposition (which is also a special case of Lemma~\ref{lem:pivot-delete}) allows a transformation equivalent to \eqref{eq:IO} to be used along flow edges in a extended causal flow $g$ -- \ie edges connecting a vertex $a$ to its corrector, the unique element of $g(a)$ -- as well as along the input and output wires.
\end{corollary}

We can now derive the well-known phase gadget addition property from the rules of Figure~\ref{fig:flow-preserving}.

\begin{lemma}[Phase gadget addition]\label{lem:phase-gadget-addition}
	Let $\Gamma = (G,I,O,\ld)$ be a labelled open graph with extended causal flow $g$.
	Suppose there exist two $YZ$-measured vertices $a,b\in\comp{O}$ such that $N_G(a)=N_G(b)$.
	Then these two phase gadgets can be replaced by a single phase gadget whose phase is the sum of the original two phases, or conversely:
	\[
		\input{rules/phase-gadget-addition-lhs.tikz} \quad\rightleftarrows\quad \input{rules/phase-gadget-addition-rhs.tikz}
	\]
\end{lemma}
\begin{proof}
	Recall that in an extended causal flow, all the neighbours of a $YZ$-measurement must be outputs or $XY$-measured.
	Without loss of generality, assume the two phase gadgets have one $XY$-measured neighbour $x$, which is not an input (if all neighbours are outputs and/or inputs, use \eqref{eq:IO}).
	For the diagram sequence below, we will only show two further neighbours of $x$, this is mainly to avoid visual clutter and is not necessary to the proof.
	(Yet, if desired, it could always be achieved by unfusing any other connections using Proposition~\ref{prop:vertex-splitting}.)
	
	The proof then proceeds as follows, where the first and third steps are pivots about the edge between the bottom phase gadget (measured with angle $\alpha$) and the neighbour $x$ (measured with angle $\gamma$):
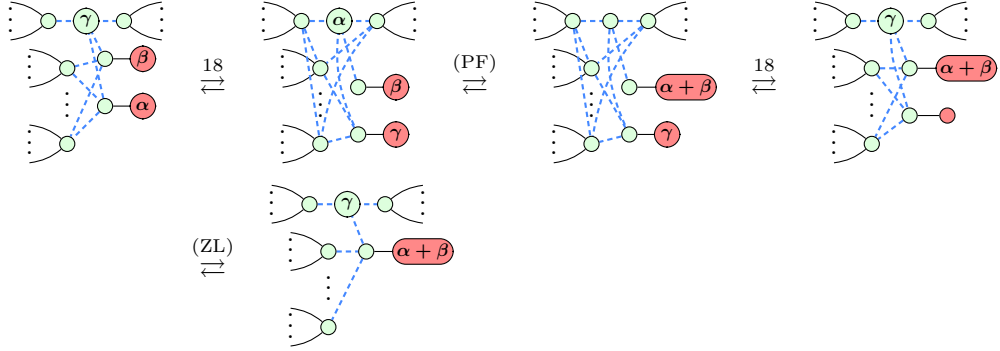
	\begin{align*}
		\input{rules/phase-gadget-addition1.tikz}
		\quad\overset{\ref{lem:pivot-Pauli-flow}}&{\rightleftarrows}\quad \input{rules/phase-gadget-addition2.tikz}
		\quad\overset{\eqref{eq:PF}}{\rightleftarrows}\quad \input{rules/phase-gadget-addition3.tikz}
		\quad\overset{\ref{lem:pivot-Pauli-flow}}{\rightleftarrows}\quad \input{rules/phase-gadget-addition4.tikz} \\
		\quad\overset{\eqref{eq:ZL}}&{\rightleftarrows}\quad \input{rules/phase-gadget-addition5.tikz}
	\end{align*}
	The left-to-right direction involves pivoting, phase fusion, and $Z$-like deletion, which all preserve the existence of gflow unconditionally; the preservation of extended causal flow follows by inspection.
	The $Z$-like insertion in the right-to-left direction is flow-preserving by Case~\ref{it:parallel-gadgets} of Corollary~\ref{cor:YZ-insertion}; again the overall preservation of extended causal flow is straightforward by inspection.
\end{proof}

\section{Flow-preserving circuit extraction}

We now show how to transform an arbitrary MBQC-form ZX-diagram which has Pauli flow into an equivalent diagram with extended causal flow.
Since diagrams with (extended) causal flow are very close to circuits, this is sufficient for our purposes.
For an explanation of how to get a circuit from a ZX-diagram with causal flow (i.e.\ from a diagram without phase gadgets), see \cite[Proposition~B.6]{duncan_graph-theoretic_2020}.
This procedure can straightforwardly be extended to handle phase gadgets via `nests' of multi-controlled phase gates or via a decomposition into controlled-NOT gates and single-qubit $Z$-rotations.

The process of getting from Pauli flow to extended causal flow has two steps. 
First, we transform to the following phase-gadget form:

\begin{definition}[Phase-gadget form {\cite[Definition~4.14]{backens_there_2021}}]
 An MBQC-form ZX-diagram is in \emph{phase-gadget form} if:
 \begin{itemize}
  \item all measurements are planar,
  \item there are no $XZ$-type measurement effects, and
  \item whenever two vertices both have $YZ$-type measurement effects, then they are not neighbours.
 \end{itemize}
\end{definition}

During this transformation, we will remove all Pauli-measured vertices or relabel them to a corresponding planar label; this is to ensure the resulting diagram has gflow.
In a second step, we transform the phase-gadget--form diagram into something that has extended causal flow.
This process is very similar to the circuit extraction algorithm for extended gflow \cite[Section~5]{backens_there_2021}, except that every step is performed in a manifestly flow-preserving way.

\begin{lemma}\label{lem:to-phase-gadget}
 Any MBQC-form ZX-diagram which has Pauli flow can be rewritten to an MBQC-form diagram which is in phase-gadget form and which has gflow.
\end{lemma}
\begin{proof}
 To get from an arbitrary MBQC-form diagram to one that is in phase-gadget form, we need to remove all Pauli measurements and all $XZ$-plane measurements, and furthermore guarantee that there are no pairs of adjacent $YZ$-measured vertices.
 Let $n_p$ be the number of Pauli measurements in the diagram and let $n_r$ be the number of $XZ$ and $YZ$ measurements.

 Proceed as follows, unfusing inputs and outputs via Lemma~\ref{lem:input-output-fusion} where necessary for applying local complementations or pivots and noting that the new vertices resulting from this lemma can always be labelled $XY$.
 \begin{enumerate}
  \item\label{it:Z-del} While there is a $Z$-measured vertex, delete it using Lemma~\ref{lem:Z-deletion-flow}.
  This step decreases $n_p$.
  When there are no more matches, proceed to step~\ref{it:X-YZ-pivot}.

  \item\label{it:X-YZ-pivot} While there exists a pair of adjacent vertices $a,b$ such that $\ld(a)\in\{X, YZ\}$ and $\ld(b)\in\{X,Y,XZ,YZ\}$, pivot on the edge $ab$ using Lemma~\ref{lem:pivot-Pauli-flow} and delete the resulting $Z$-measurements (if any).
  Pivoting changes the measurement labels of the edge endpoints as follows: $X\mapsto Z$, $Y\mapsto Y$, $XZ\mapsto XZ$ and $YZ\mapsto XY$.
  Other measurement labels are left invariant.
  Hence this step decreases at least one of $n_p$ and $n_r$ and does not increase either number.
  When there are no more matches, proceed to step~\ref{it:LC}.

  \item\label{it:LC} At this point, if there exists $a$ such that $\ld(a)\in\{Y, XZ\}$ then for all $b\in N_G(a)$ we have $\ld(b)\notin\{X,Z,YZ\}$.
  So apply a local complementation to $a$ using Lemma~\ref{lem:local-complementation-flow} and delete the resulting $Z$-measurement (if any).
  Local complementations change the measurement label of the central vertex according to $Y\mapsto Z$ and $XZ \mapsto XY$, and the labels of neighbours according to $Y\mapsto X$, $XY\mapsto XY$, $XZ\mapsto YZ$, hence this step decreases at least one of $n_p$ and $n_r$ and does not increase either number.
  If a local complementation was applied, go back to step~\ref{it:X-YZ-pivot}, otherwise proceed to step~\ref{it:X-XY-pivot}.

  \item\label{it:X-XY-pivot} At this point, we know that neither step~\ref{it:X-YZ-pivot} nor step~\ref{it:LC} found any matches.
  This means in particular that there are no $Y$ or $XZ$-measurements left.
  Also, there are no pairs of adjacent $YZ$-measurements.
  Thus the only way the diagram can fail to be in phase gadget form is if there are any $X$-measurements left.
  A remaining $X$ cannot be adjacent to an $X$ or $YZ$-measurement, but it must have a neighbour.
  Hence the neighbour is either $XY$-measured or an output.
  Using Lemma~\ref{lem:input-output-fusion} to unfuse if necessary, we can ensure this neighbour is neither an input nor an output.

  Then pivot on the edge between the $X$ and the $XY$, this sends $X\mapsto Z$ and $XY\mapsto YZ$ and does not affect the measurement labels of any neighbours.
  Delete the resulting $Z$-measurement and go back to step~\ref{it:X-YZ-pivot}.
  This step decreases $n_p$ and increases $n_r$, though $n_p+n_r$ does not increase.
  If there were no matches for this step, stop the procedure.
 \end{enumerate}

 Overall, each step decreases at least one of $n_p$ and $n_r$.
 The number $n_p$ never increases and when $n_r$ increases, then $n_p$ decreases by the same amount.
 Thus the procedure terminates.
 When it does, there are no $Z$-measurements (which are removed in step~\ref{it:Z-del} and then as soon as they appear), no $Y$- or $XZ$-measurements (which are removed in step~\ref{it:LC}), no $X$-measurements (which are removed in steps~\ref{it:X-YZ-pivot} and~\ref{it:X-XY-pivot}) and no pairs of adjacent $YZ$-measurements (which are removed in step~\ref{it:X-YZ-pivot}).
 Thus the diagram is in phase-gadget form.

 All transformations preserve interpretation and the existence of Pauli flow.
 As all measurement labels at the end are planar, the final diagram has gflow.
\end{proof}

\subsection{Extracted vertices and frontier}
\label{s:extracted-frontier}

For flow-preserving circuit extraction, we need a way of describing a flow which is causal towards the top of the partial order but a gflow further towards the bottom.
To this end, we define what it means for a vertex in a \LOG to be `extracted'.
Throughout this section, we will assume the \LOG contains only $XY$ and $YZ$ measurements.

\begin{definition}\label{def:extracted-frontier}
 Suppose $\Gamma = (G,I,O,\ld)$ is a \LOG with $\ld(V)\sse\{XY,YZ\}$ which has gflow $g$.
 We say $S$ is a \emph{set of extracted vertices for $(\Gamma,g)$} if $O\sse S$ and, for all $v\in S\setminus O$,
 \begin{equation}\label{eq:definition_S}
  (\abs{g(v)}=1) \wedge (\forall w \in V : v\prec_g w \implies w\in S).
 \end{equation}
 If $S$ is a set of extracted vertices, then its \emph{frontier} $F$ is defined as
 \begin{equation}\label{eq:definition-frontier}
    F := \{v\in S \mid \nexists u\in S\setminus O \text{ s.t. } g(u)=\{v\}\}.
 \end{equation}
 We sometimes refer to $\overline{S}:=V\setminus S$ as the \emph{set of unextracted vertices}.
\end{definition}

While the definition of $S$ is self-referential, we can straightforwardly check it by working down the partial order starting from the outputs.

In words, a set of extracted vertices consists of vertices at the top of the partial order which are either outputs or have causal correction sets.
The \emph{frontier} consists of those extracted vertices which are outputs or can be corrected causally, but are not the unique element of the correction set for any extracted vertex.

\begin{remark}\label{rem:trivial_extracted_sets}
 The choice $S=O$ yields a valid set of extracted vertices for any gflow.
 The choice $S=V$ yields a valid set of extracted vertices for any extended causal flow.
 
 Moreover, suppose $S\neq O$ is a set of extracted vertices, then there exists some $v\in S\setminus O$ such that $S\setminus\{v\}$ is also a set of extracted vertices.
 In particular, any non-output $v$ that is minimal in the partial order (among elements of $S$) can be removed in this way.
\end{remark}

\begin{lemma}\label{lem:frontier-properties}
 Suppose $\Gamma = (G,I,O,\ld)$ is a \LOG with $\ld(V)\sse\{XY,YZ\}$ which has gflow $g$.
 Let $S$ be a set of extracted vertices for $(\Gamma,g)$, let $F$ be its frontier, and let $a\in F$.
 The following hold:
 \begin{enumerate}
  \item If $a\in g(b)$ for some $b\in\comp{O}$, then $b\notin S$.
  \item Either $a\in O$ or $\ld(a)=XY$.
  \item $\abs{F} = \abs{O}$.
 \end{enumerate}
\end{lemma}
\begin{proof}
 For the first part, suppose for a contradiction there exists $b\in S$ such that $a\in g(b)$.
 But then $\abs{g(b)}=1$ by \eqref{eq:definition_S}, and thus $g(b)=\{a\}$, which contradicts \eqref{eq:definition-frontier}.
 So if $a$ appears in a correction set, it must be the correction set of an unextracted vertex.

 For the second part, note that if $\ld(v)=YZ$ for some $v$ and $\abs{g(v)}=1$, then $g(v)=\{v\}$ by \ref{itm:gYZ}.
 Thus, by the previous argument, any $YZ$-measured vertex is either unextracted or in $S\setminus F$.
 
 For the third part, we will prove this by induction on the membership of $S\setminus O$ for some given pair $(\Gamma,g)$.
 The base case is $S\setminus O = \emptyset$, \ie $S=O=F$, in which case $S$ is always a valid set of extracted vertices by Remark~\ref{rem:trivial_extracted_sets} and moreover the desired property holds trivially.
 For the inductive step, suppose the desired property holds  for some set of extracted vertices $S$ with frontier $F$, and that there exists $v\in\comp{S}$ such that $S'=S\cup\{v\}$ is also a set of extracted vertices for $(\Gamma,g)$.
 Denote by $F'$ the frontier of $S'$ and distinguish cases according to the measurement label of $v$.
 \begin{itemize}
 	\item If $\ld(v)=YZ$, then $g(v)=\{v\}$ since $v$ is an extracted vertex.
 	Hence $v$ can neither appear in the frontier nor does it have a predecessor it prevents from appearing in the frontier.
 	So by \eqref{eq:definition-frontier}:
 	\[
 	  F' = \{v\in S' \mid \nexists u\in S'\setminus O \text{ s.t. } g(u) = v\} = \{v\in S \mid \nexists u\in S\setminus O \text{ s.t. } g(u) = v\} = F
 	\]
 	and we have $\abs{F'}=\abs{O}$ by the inductive hypothesis.
 	
 	\item If $\ld(v)=XY$, then there exists $v'\in S$ such that $g(v)=\{v'\}$ since $S'=S\cup\{v\}$ satisfies Definition~\ref{def:extracted-frontier}.
 	By the uniqueness of correction sets, cf.~Remark~\ref{rem:unique-correction-sets}, there cannot be another element of $S$ that is corrected by $v'$, hence $v'\in F$ but $v'\notin F'$.
 	For all $w\in F\setminus\{v'\}$, we have $w\in F'$ as $v$ is the only vertex that was added to the set of extracted vertices and $g$ has not changed.
 	We will now show that $v\in F'$: suppose otherwise for a contradiction, then there exists $v''\in S$ such that $g(v'')=\{v\}$.
 	But then $v''\prec_g v$ and $S$ is upwards closed, so we must have $v\in S$.
 	This contradicts the assumption $v\in\comp{S}$, so indeed $v\in F'$.
 	There can be no other new elements of $F'$ as the set of extracted vertices has not changed much.
 	Thus $\abs{F'} = \abs{F\setminus\{v'\}\cup\{v\}} = \abs{F} = \abs{O}$ by the inductive hypothesis.
 \end{itemize}
 This completes the argument regarding the size of $F$.
\end{proof}

Having looked at properties of vertices within the frontier in the previous lemma, we now consider properties of vertices internal to $S$, \ie vertices in $S\setminus F$.

\begin{lemma}\label{lem:neighbours-extracted-internal}
 Suppose $\Gamma = (G,I,O,\ld)$ is a \LOG with $\ld(V)\sse\{XY,YZ\}$ which has gflow $g$, and suppose $S$ is a set of extracted vertices for $(\Gamma,g)$.
 Let $F$ be the frontier of $S$.
 Then any $v\in S\setminus F$ satisfies $N_G(v) \sse S$.
\end{lemma}
\begin{proof}
 By Definition~\ref{def:extracted-frontier}, $v\in S\setminus F$ implies that there exists $u\in S\setminus O$ such that $g(u)=\{v\}$.
 Now $w\in\odd{G}{g(u)} = N_G(v)$ implies $w = u \vee u\prec_g w$ by \eqref{eq:prec-Z} of Definition~\ref{def:induced-relation}.
 But by \eqref{eq:definition_S}, ${u\prec_g w} \implies w\in S$ for all $w$.
 Thus, $w\in S$ in either case, which implies $N_G(v)\sse S$.
\end{proof}

\begin{lemma}\label{lem:correction-remove-extracted}
 Suppose $\Gamma = (G,I,O,\ld)$ is a \LOG with $\ld(V)\sse\{XY,YZ\}$ which has gflow $g$.
 Let $S$ be a set of extracted vertices for $(\Gamma,g)$.
 Suppose $b\in S$ has the property that $N_G(b)\sse S$.
 Then $g'$ is a gflow for $\Gamma$, where
 \[
  g'(u) :=
  \begin{cases}
   g(u) &\text{if } u \in S \\
   g(u)\setminus\{b\} &\text{otherwise.}
  \end{cases}
 \]
 Moreover, $S$ is a set of extracted vertices for $(\Gamma,g')$.
\end{lemma}
\begin{proof}
 As $b\in S$ and only correction sets of vertices in $\comp{S}$ are changed, property \ref{itm:gYZ} remains satisfied for all vertices.
 Moreover, for $u\in\comp{S}$, we can rewrite $g'(u) = g(u)\setminus\{b\} = g(u)\symd(g(u)\cap\{b\})$, hence:
 \begin{equation}\label{eq:odd-nhood-remove-extracted}
 \odd{G}{g'(u)} = \odd{G}{g(u)\symd(g(u)\cap\{b\})} = \odd{G}{g(u)} \symd \odd{G}{g(u)\cap\{b\}}
 \end{equation}
 Now $g(u)\cap\{b\}\sse\{b\}$, so we have $\odd{G}{g(u)\cap\{b\}}\sse S$ because of the assumption $N_G(b)\sse S$.
 Thus property \ref{itm:gXY} also remains satisfied for all vertices (correction sets and thus odd neighbourhoods do not change for elements of $S$).
 It remains to show that the relation $R^*_{g'}$ induced by $g'$ is a strict partial order.

 Denote $R_X,R_Z$ the relations induced by the original gflow $g$ according to \eqref{eq:prec-X} and \eqref{eq:prec-Z} and denote $R_X',R_Z'$ the relations induced by $g'$.
 ($R_Y$ is empty for both $g$ and $g'$ as there are no $Y$-measurements.)
 Now, $R_X' \sse R_X$ as no elements are added to correction sets.
 Furthermore, \eqref{eq:odd-nhood-remove-extracted} implies $R_Z' \sse R_Z \cup (\comp{S}\times S)$ since any newly introduced relationship $(u,v)$ satisfies $u\in\comp{S}$ and $v\in S$.
 This means $R_{g'}^*$ is contained in the transitive closure of $R_g^*\cup (\comp{S}\times S)$, denote the latter by $R_g'$.
 By Lemma~\ref{lem:extend-partial-order}, $R_g'$ is a strict partial order and thus $R_{g'}^*$ is also a strict partial order.
 Therefore $g'$ is a gflow.
\end{proof}

Together, these two lemmas imply the following particular case.

\begin{corollary}\label{cor:correction-remove-extracted-internal}
 Suppose $\Gamma = (G,I,O,\ld)$ is a \LOG with $\ld(V)\sse\{XY,YZ\}$ which has gflow $g$.
 Let $S$ be a set of extracted vertices for $(\Gamma,g)$ with frontier $F$.
 Then $g'$ is a gflow for $\Gamma$, where
 \[
  g'(v) :=
  \begin{cases}
   g(v) &\text{if } v \in S \\
   g(v)\setminus (S\setminus F) &\text{otherwise.}
  \end{cases}
 \]
\end{corollary}

\begin{definition}\label{def:pseudo-focused}
 Suppose $\Gamma = (G,I,O,\ld)$ is a \LOG with $\ld(V)\sse\{XY,YZ\}$ which has gflow $g$.
 Let $S\sse V$.
 We say $g$ is \emph{pseudo-focused with respect to $S$} if the following properties hold:
   \begin{enumerate}
    \item\label{it:extracted} $S$ is a set of extracted vertices for $(\Gamma,g)$; call its frontier $F$.
    \item\label{it:correction-set} For all $v\in\comp{S}$, if $w\in g(v)\setminus\{v\}$, then $w\notin S\setminus F$ and $w\in O \vee \ld(w)=XY$.
    \item\label{it:odd-neighbourhood} For all $v\in\comp{O}$, if $u\in\odd{G}{g(v)}\setminus\{v\}$, then $u\in S$ or $\ld(u)=YZ$.
   \end{enumerate}
\end{definition}

The following lemma shows that focused gflow and extended causal flow are two extremal examples of pseudo-focused flows.
This is indeed the motivation for the definition, which -- as mentioned at the beginning of Section~\ref{s:extracted-frontier} -- is to define flows that are causal near the top of the partial order and focused near the bottom.

\begin{lemma}\label{lem:trivial-pseudo-focused}
 Suppose $\Gamma = (G,I,O,\ld)$ is a \LOG with $\ld(V)\sse\{XY,YZ\}$ which has gflow $g$.
 Then:
 \begin{itemize}
  \item $g$ is a focused gflow if and only if $g$ is pseudo-focused with respect to $O$.
  \item $g$ is an extended causal flow if and only if $g$ is pseudo-focused with respect to $V$.
 \end{itemize}
\end{lemma}
\begin{proof}
 First, suppose $g$ is a focused gflow.
 \begin{itemize}
  \item Then by Remark~\ref{rem:trivial_extracted_sets}, $O$ is a set of extracted vertices so Property~\ref{it:extracted} holds.
  Moreover, this set is its own frontier $F=O$.
  \item For any $v\in\comp{S}=\comp{O}$, if $w\in g(v)\setminus\{v\}$, then trivially $w\notin S\setminus F = \emptyset$.
  The second part of Property~\ref{it:correction-set} holds by \ref{itm:focus-X}.
  \item For any $w\in\comp{O}$, if $u\in\odd{G}{g(v)}\setminus\{v\}$, then by \ref{itm:focus-Z} $u\in O$ or $\ld(u)=YZ$ so Property~\ref{it:odd-neighbourhood} holds.
 \end{itemize}
 Thus $g$ is pseudo-focused with respect to $O$.

 Conversely, suppose $g$ is a gflow which is pseudo-focused with respect to $O$.
 Then by Property~\ref{it:correction-set}, for all $v\in\comp{O}$, $w\in g(v)\cap (\comp{O}\setminus\{v\})$ implies $\ld(w) = XY$ and thus \ref{itm:focus-X} is satisfied.
 Moreover, by Property~\ref{it:odd-neighbourhood}, for all $v\in\comp{O}$, $u\in\odd{G}{g(v)}\cap (\comp{O}\setminus\{v\})$ implies $\ld(u) = YZ$, so \ref{itm:focus-Z} is satisfied.
 Thus, noting that \ref{itm:focus-Y} is trivial for gflows, $g$ is a focused gflow.

 Now suppose $g$ is an extended causal flow.
 Then by Remark~\ref{rem:trivial_extracted_sets}, $V$ is a set of extracted vertices so Property~\ref{it:extracted} holds.
 We have $\comp{S} = \emptyset$, so Property~\ref{it:correction-set} is trivial.
 Similarly, $S=V$ makes Property~\ref{it:correction-set} trivial.
 Thus $g$ is pseudo-focused with respect to $V$.

 Conversely, suppose $g$ is a gflow which is pseudo-focused with respect to $V$.
 Then by Definition~\ref{def:extracted-frontier}, $\abs{g(v)} = 1$ for all $v\in V\setminus O$.
 Hence $g$ is an extended causal flow.
\end{proof}

Generalising the idea that any gflow can be focused, we now show that any pair of a gflow and a compatible set of extracted vertices can be pseudo-focused.

\begin{lemma}\label{lem:pseudo-focusing}
 Suppose $\Gamma = (G,I,O,\ld)$ is a \LOG with $\ld(V)\sse\{XY,YZ\}$ which has gflow $g$.
 Let $S$ be a set of extracted vertices for $(\Gamma,g)$.
 Then there exists a gflow $g'$ on $\Gamma$ which is pseudo-focused with respect to $S$.
\end{lemma}
\begin{proof}
 Let $F$ be the frontier of $S$ and apply Corollary~\ref{cor:correction-remove-extracted-internal}.
 The remaining step is to focus the flow on $\comp{S}$, using a standard focusing procedure as in the proof of \cite[Proposition~3.14]{backens_there_2021} but treating all vertices in $S$ like outputs.
\end{proof}

\subsection{Circuit extraction lemmas}

We now develop the manifestly flow-preserving circuit extraction algorithm, which closely follows the extended gflow circuit extraction algorithm \cite{backens_there_2021}.
In addition to straightforward unfusion, that algorithm uses two transformations to push operations from the MBQC-form part of the diagram into the circuit-like part of the diagram:
\begin{itemize}
 \item pushing a CNOT gate into the circuit to apply a row operation to the adjacency matrix between the frontier and its neighbours in the MBQC-form part, and
 \item pushing a Hadamard gate into the circuit to pivot on the edge between a frontier vertex and an adjacent $YZ$-measured vertex.
\end{itemize}
We now show that these two operations can be performed in a gflow-preserving way via $YZ$-insertions and pivoting.
The first lemma is the equivalent of inserting a CNOT gate with control $a$ and target $b$, where $a,b$ are two distinct vertices in the frontier.
It is a special case of the `insert \& pivot' rule of Lemma~\ref{lem:pivot-delete}.

\begin{lemma}\label{lem:insert-cnot}
 Let $\Gamma=(G,I,O,\ld)$ be a \LOG where $\ld(V)\sse\{XY,YZ\}$ and which has gflow $g$; let $S$ be a set of extracted vertices for $(\Gamma,g)$ with frontier $F$.
 Suppose $a,b\in F$ are two distinct vertices.
 \begin{enumerate}
  \item Let $\Gamma'$ be the \LOG after inserting into $\Gamma$ a new $YZ$-measured vertex~$z_1$ with neighbours $\{a,b\}$.
  \item Let $\Gamma''$ be the \LOG after inserting into $\Gamma'$ a new $YZ$-measured vertex~$z_2$ with neighbours $K\cup\{z_1\}$, where $K := N_G(b)\cap\comp{S}$.
  \item Let $\Gamma'''$ be the \LOG after pivoting on the edge $z_1 z_2$ in $\Gamma''$.
 \end{enumerate}
 This sequence of operations is illustrated in Figure~\ref{fig:insert-cnot}.
 Then $\Gamma'''$ has a gflow for which $S\cup\{z_1,z_2\}$ is a set of extracted vertices, \ie the number of unextracted vertices has not increased.
 \begin{figure}
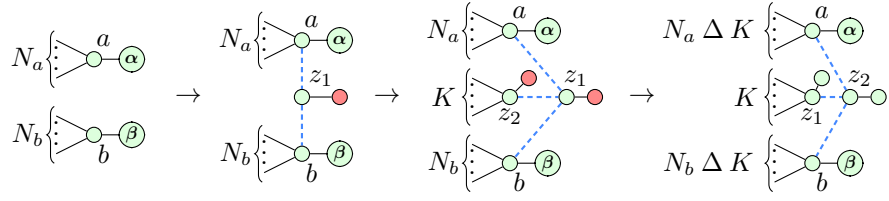

 	\ctikzfig{cnot-insertion}
 	\caption{The sequence of rewriting steps in Lemma~\ref{lem:insert-cnot}, with $N_a:=N_G(a)$ and $N_b:=N_G(b)$. Since $a,b\in F$, they must be either $XY$-measured (as illustrated) or outputs. The vertices $a$ and $b$ could also be neighbours, this would not change anything about the procedure and they would remain neighbours. Note that $N_b\symd K = N_b \symd (N_b\cap\comp{S}) = N_b\cap S$, \ie $b$ is left with only extracted neighbours.}
 	\label{fig:insert-cnot}
 \end{figure}
\end{lemma}
\begin{proof}
 We first show how the gflow changes after each step.
 \begin{enumerate}
  \item Take $N=\{a,b\}$ and $C=\emptyset$, then \eqref{eq:pred} and \eqref{eq:succ} become
   \begin{align*}
    \Pred(\{a,b\}) &= \{v\in\comp{O}\mid a\in g(v) \oplus b\in g(v)\} \\
    \Succ(\{a,b\},\emptyset) &= \{a,b\}
   \end{align*}
   Now the first set satisfies $\Pred(\{a,b\})\sse\comp{S}$ by Lemma~\ref{lem:frontier-properties} since $a,b\in F$. The second set satisfies $\Succ(\{a,b\},\emptyset)\sse S$ because $a,b$ are extracted.
   Since $S$ is upwards closed, Corollary~\ref{cor:YZ-insertion} with $U=S$, $D=\comp{S}$, $N=\{a,b\}$ and $C=\emptyset$ shows that the first insertion preserves the existence of gflow.
   Call the resulting correction function $g'$.

  \item First, note that the set $S$ is still upwards closed in $\Gamma'$ because extracted vertices appear only as successors in the newly-added relations.
  
   Now, applying \eqref{eq:pred} and \eqref{eq:succ} to $\Gamma'$ and $g'$ with $N=K\cup\{z_1\}$ and $C=\{b\}$, we have
   \begin{align*}
    \Pred(K\cup\{z_1\}) &= \{z_1\}\cup\{v\in\comp{O}\mid \abs{K\cap g'(v)}\equiv 1\bmod 2\} \\
    \Succ(K\cup\{z_1\},\{b\}) &= \{b\} \cup N_G(b)\setminus K = \{b\} \cup (N_G(b)\cap S)
   \end{align*}
   Recall that $K\in\comp{S}$ by definition and, being the complement of an upwards-closed set, $\comp{S}$ is downwards closed.
   Any $v\in\comp{O}$ such that $\abs{K\cap g'(v)}\equiv 1\bmod 2$ satisfies $v\prec_{g'} k$ for some $k\in K\cap g'(v)$.
   Thus we have $\Pred(K\cup\{z_1\}) \sse \comp{S}\cup\{z_1\}$.
   
   Furthermore, $\Succ(K\cup\{z_1\},\{b\})\sse S$ by inspection.
   
   Thus Corollary~\ref{cor:YZ-insertion} with $U=S$, $D=\comp{S}\cup\{z_1\}$, $N=K\cup\{z_1\}$ and $C=\{b\}$ shows that the second insertion step preserves the existence of gflow.
   Call the new correction function $g''$ and note that $g''(v)=g(v)$ for all original vertices $v$.

  \item The pivoting step preserves the existence of gflow by Lemma~\ref{lem:pivot-Pauli-flow}.
  That lemma implies a new correction function where the correction set of $v$ changes only if $z_1$ or $z_2$ is contained in $\codd{G''}{g''(v)}$.
  Now, $z_1, z_2$ do not appear in any correction sets other than their own, hence for original vertices $v$ the condition reduces to appearing in $\odd{G''}{g''(v)}$, which is equivalent to $v\in\Pred(N_{G''}(z_1))$ or $v\in\Pred(N_{G''}(z_2))$, respectively.
  Thus:
  \begin{equation}\label{eq:flow-after-pivot}
   g'''(v) =
   \begin{cases}
    \{z_2\} &\text{if } v = z_1 \\
    \{b\} &\text{if } v = z_2 \\
    g(v)\cup\{z_1,z_2\} &\text{if } (a\in g(v) \oplus b\in g(v)) \wedge (\abs{K\cap g'(v)}\equiv 1\bmod 2) \\
    g(v)\cup\{z_1\} &\text{if } (a\in g(v) \oplus b\in g(v)) \wedge (\abs{K\cap g'(v)}\equiv 0\bmod 2) \\
    g(v)\cup\{z_2\} &\text{if }  \neg(a\in g(v) \oplus b\in g(v)) \wedge (\abs{K\cap g'(v)}\equiv 1\bmod 2) \\
    g(v) &\text{otherwise.}
   \end{cases}
  \end{equation}
  The correction sets of original vertices change only if those vertices are direct predecessors of $z_1$ or $z_2$, so in particular for all extracted vertices the correction set remains invariant: $g'''(v)=g(v)$ for all $v\in S\setminus O$.
  Note that pivoting does not affect the induced partial order.
 \end{enumerate}
 This completes the proof that gflow is preserved and leads to the final part of the lemma.
 Let $S':=S\cup\{z_1,z_2\}$.
 Since $S$ is a set of extracted vertices and the correction function has not changed for elements of $S$, it only remains to show that $z_1,z_2$ satisfy the conditions for being extracted.
 Now, $\abs{g(z_1)}=\abs{g(z_2)}=1$ by inspection.
 Moreover, in the gflow preservation proofs, we already showed that all successors (among original vertices) of $z_1$ and $z_2$ are in $S'$.
 The only new direct successor of $z_1$ in step~2 is $z_2$, which is also in $S'$.
 Thus $S'$ is a set of extracted vertices by Definition~\ref{def:extracted-frontier}.
\end{proof}

\begin{corollary}\label{cor:insert-cnot-simplified}
 The set $S' = S\cup\{z_1,z_2\}$ is a set of extracted vertices for the final \LOG $\Gamma'''$ in Lemma~\ref{lem:insert-cnot}; it has frontier $F'=(F\setminus\{b\})\cup\{z_1\}$.
 Thus, according to Corollary~\ref{cor:correction-remove-extracted-internal}, we can simplify the correction function on $\Gamma'''$ by removing the vertices $b$ and $z_2$, from the correction sets of all unextracted vertices.
 This means \eqref{eq:flow-after-pivot} becomes
 \[
  g''''(v) =
   \begin{cases}
    \{z_2\} &\text{if } v = z_1 \\
    \{b\} &\text{if } v = z_2 \\
    g(v)\cup\{z_1\} &\text{if } a\in g(v) \wedge b\notin g(v) \\
    (g(v)\cup\{z_1\})\setminus\{b\} &\text{if }  a\notin g(v) \wedge b\in g(v) \\
    g(v)\setminus\{b\} &\text{if }  a, b\in g(v) \\
    g(v) &\text{otherwise.}
   \end{cases}
 \]
 It is straightforward to see that $g''''$ is pseudo-focused with respect to $S'$ if and only if $g$ is pseudo-focused with respect to $S$.
\end{corollary}

The extended gflow circuit extraction algorithm would fail if at any point all the next vertices to be extracted -- i.e.\ all the maximal non-extracted vertices -- are $YZ$-measured and none of them have a neighbour in the frontier.
Fortunately, this situation cannot happen: we now adapt that result to our terminology and notation and give a proof for completeness.

\begin{lemma}[adapted from {\cite[Lemma~5.4]{backens_there_2021}}]\label{lem:phase-gadget-extracted-neighbour}
 Let $\Gamma=(G,I,O,\ld)$ be a \LOG in phase-gadget form which has gflow $g$; let $S$ be a set of extracted vertices for $(\Gamma,g)$ with frontier $F$.
 Without loss of generality, assume $g$ is pseudo-focused with respect to $S$, otherwise apply Lemma~\ref{lem:pseudo-focusing}.

 Suppose $S\neq V$ and define the set of maximal unextracted vertices $M:= \{v\in\comp{S}\mid \nexists w\in\comp{S}. v\prec_g w\}$.
 If $\ld(v)=YZ$ and $N_G(v)\cap\comp{S}\neq\emptyset$ for all $v\in M$, then there exists $v\in M$ such that $N_G(v)\cap F\neq\emptyset$.
\end{lemma}
\begin{proof}
 By the assumption of phase gadget form, all neighbours of $YZ$-measured vertices in $\comp{S}$ must be $XY$-measured.
 Let $u$ be maximal according to $\prec_g$ among the $XY$-measured non-extracted vertices.
 Then $g(u)\sse S$ because the only unextracted vertices after $u$ in the partial order are $YZ$-measured and thus do not appear in correction sets of $g$.
 Then by pseudo-focusing (Definition~\ref{def:pseudo-focused}), $g(u)\sse F$.

 Yet since $u\notin M$ (because all elements of $M$ are $YZ$-measured), there must be a vertex in $M$ that comes after $u$ and by maximality of $u$, it must be a direct successor.
 Thus there exists $v\in M\cap\odd{G}{g(u)}$, which means this $v$ has a neighbour in~$F$.
\end{proof}

We now show how to handle such a maximal $YZ$-measured vertex with a frontier neighbour in the circuit extraction procedure.
This is the equivalent of pushing a Hadamard gate into the circuit on the wire corresponding to the frontier neighbour $b$ of the maximal non-extracted $YZ$-measured vertex $a$.

\begin{lemma}\label{lem:pivot-YZ}
 Let $\Gamma=(G,I,O,\ld)$ be a \LOG in phase-gadget form which has gflow $g$; let $S$ be a set of extracted vertices for $(\Gamma,g,\prec_g)$ with frontier $F$.
 Suppose the vertex $a$ is $YZ$-measured and maximal in $\comp{S}$, and that there exists $b\in N_G(a)\cap F$.
 \begin{enumerate}
  \item Let $\Gamma'$ be the \LOG after inserting a new $YZ$-measured vertex $z$ with neighbours $K:=N_G(b)\cap\comp{S}$ into $\Gamma$.
  \item Let $\Gamma''$ be the \LOG after pivoting on the edge $a z$ in $\Gamma'$.
 \end{enumerate}
 Then $\Gamma''$ has a gflow for which $S\cup\{z\}$ is a set of extracted vertices, \ie the number of unextracted vertices has not increased.
 \begin{figure}
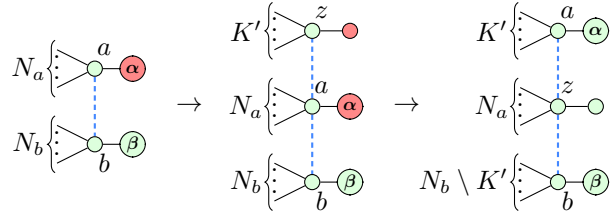

  \ctikzfig{H-pivot}
  \caption{The operations performed in Lemma~\ref{lem:pivot-YZ}.
  Here, $N_a := N_G(a)\setminus\{b\}$, $N_b := N_G(b)\setminus\{a\}$, and $K' := K\setminus\{a\} = (N_G(b)\cap\comp{S})\setminus\{a\}$.
  Vertex $b$ could also be an output.}
  \label{fig:pivot-YZ}
 \end{figure}
\end{lemma}
\begin{proof}
 We first show how the gflow changes after each step, see Figure~\ref{fig:pivot-YZ} for the relevant diagrams.
 Note $a\in K$.
 \begin{enumerate}
  \item Take $N=K=N_G(b)\cap\comp{S}$ and $C=\{b\}$, then \eqref{eq:pred} and \eqref{eq:succ} become:
   \begin{align*}
    \Pred(K) &= \{v\in\comp{O}\mid \abs{K\cap g(v)}\equiv 1 \bmod 2\} \\
    \Succ(K,\{b\}) &= \{b\} \cup (N_G(b)\symd K) = \{b\} \cup (N_G(b)\cap S)
   \end{align*}
   Now the first set satisfies $\Pred(K)\sse\comp{S}$ since $K\sse\comp{S}$ by definition, $K\cap g(v)\neq\emptyset$ implies $\exists w\in K.v\prec w$, and $\comp{S}$ is downwards closed.
   The second set satisfies $\Succ(K,\{b\})\sse S$.

   Hence Corollary~\ref{cor:YZ-insertion} with $U=S$, $D=\comp{S}$, $N=N_G(b)$ and $C=\{b\}$ shows that the insertion preserves the existence of gflow.
   Call the resulting correction function $g'$ and note that $g'(v)=g(v)$ for all original vertices $v\in\comp{O}$.

  \item The pivoting step preserves the existence of gflow by Lemma~\ref{lem:pivot-Pauli-flow}.
  That lemma implies a new correction function where the correction set of $v$ changes only if $a$ or $z$ is contained in $\codd{G'}{g'(v)}$.
  Now, $a,z$ do not appear in any correction sets other than their own, hence for other vertices $v$ the condition reduces to appearing in $\odd{G'}{g'(v)}$.
  The condition $z\in\odd{G'}{g'(v)}$ is equivalent to $\abs{K\cap g'(v)} \equiv 1 \bmod 2$.
  As $a$ is maximal in $\comp{S}$, we have $g'(a)\setminus\{a\} \sse S$, which implies $K\cap g'(a) = \{a\}$ because all elements of $K$ are in $\comp{S}$ by definition.
  Therefore $z\in\odd{G'}{g'(a)}$.
  Thus:
  \[
   g''(v) =
   \begin{cases}
    \{b\} &\text{if } v = z \\
    g(a)\setminus\{a\}\cup\{z\} &\text{if } v = a \\
    g(v)\cup\{a,z\} &\text{if } a\in\odd{G}{g(v)} \wedge \abs{K\cap g(v)}\equiv 1 \bmod 2 \\
    g(v)\cup\{a\} &\text{if } a\in\odd{G}{g(v)} \wedge \abs{K\cap g(v)}\equiv 0 \bmod 2 \\
    g(v)\cup\{z\} &\text{if } a\notin\odd{G}{g(v)} \wedge \abs{K\cap g(v)}\equiv 1 \bmod 2 \wedge v\neq a \\
    g(v) &\text{otherwise.}
   \end{cases}
  \]
  The correction sets of original vertices change only if they precede $a$ or $z$, so in particular $g''(v)=g(v)$ for all $v\in S$.
  Note that pivoting does not affect the induced partial order.
 \end{enumerate}
 This completes the proof that gflow is preserved and leads to the final part of the lemma.
 Let $S':=S\cup\{z\}$.
 Since $S$ is a set of extracted vertices and the correction function has not changed for elements of $S$, it only remains to show that $z$ satisfy the conditions for being extracted.
 Now, $\abs{g(z)}=1$ by inspection.
 Moreover, in the gflow preservation proofs, we already shows that all successors (among original vertices) of $z$ are in $S'$.
 Thus $S'$ is a set of extracted vertices by Definition~\ref{def:extracted-frontier}.
\end{proof}

\subsection{Circuit extraction algorithm}
\label{s:circuit-extraction}

We are now ready for the manifestly flow-preserving circuit extraction algorithm -- \ie the algorithm that turns any ZX-diagram with flow into an equivalent ZX-diagram with extended causal flow, following the algorithms of \cite[Section~7]{duncan_graph-theoretic_2020} and \cite[Section~5.1]{backens_there_2021}.

\begin{algorithm}
	\caption{From Pauli flow to extended causal flow}
	\begin{algorithmic}[1]
		\Require $D$ is an MBQC-form ZX-diagram which has Pauli flow
		\Procedure{PauliToCausal}{$D$}
			\State $\Gamma \gets$ labelled open graph corresponding to $D$ (with max.\ number of Pauli meas.)
			\State $g \gets$ focused Pauli flow on $\Gamma $
			\State $D, \Gamma, g \gets \Call{ToPhaseGadgetForm}{D, \Gamma, g}$ \Comment{Lemma~\ref{lem:to-phase-gadget}}
			\State $T \gets$ set of non-outputs of $\Gamma$ \Comment{non-extracted vertices, denoted $\comp{S}$ in text}
			\While{$T\neq\emptyset$}
				\State $M \gets \{v\in T \mid \nexists w\in T. v\prec_g w\}$ \Comment{maximal non-extracted vertices}
				\If{$\exists v\in M$ such that $\abs{g(v)}=1$}
					\State $v\gets$ an element of $M$ whose correction set has size~1
					\State Remove $v$ from $T$
				\ElsIf{$\exists v\in M$ such that $v$ is $XY$-measured}
					\State $v\gets$ an $XY$-measured element of $M$
					\State $a \gets$ an element of $g(v)$
					\While{$g(v)\setminus\{a\}\neq\emptyset$}
						\State $b \gets$ an element of $g(v)\setminus\{a\}$
						\State $D, \Gamma, g \gets \Call{InsertCNOT}{D, \Gamma, g, a, b}$ \Comment{Lemma~\ref{lem:insert-cnot}}
					\EndWhile
				\Else \Comment{all elements of $M$ are $YZ$-measured}
					\State $a \gets$ an element of $M$ which has a neighbour in the current frontier 
					\Statex \Comment{this exists by Lemma~\ref{lem:phase-gadget-extracted-neighbour}}
					\State $b \gets$ a frontier neighbour of $a$
					\State $D, \Gamma, g \gets \Call{InsertAndPivot}{D, \Gamma, g, a, b}$ \Comment{Lemma~\ref{lem:pivot-YZ}}
				\EndIf
			\EndWhile
			\State \Return $D$
		\EndProcedure
	\end{algorithmic}
\end{algorithm}

\begin{enumerate}
	\item Given a ZX-diagram, compute the corresponding labelled open graph with the maximum number of Pauli measurements, and find a focused Pauli flow on it.
	This tuple of a labelled open graph and a Pauli flow on it will be updated throughout the algorithm.
	\item Transform the ZX-diagram to phase-gadget form using Lemma~\ref{lem:to-phase-gadget}, so the resulting diagram has focused gflow.
	 Then by Lemma~\ref{lem:trivial-pseudo-focused}, the gflow is pseudo-focused with respect to $O$.
	
	 From now on, all measurement labels are $XY$ and $YZ$ and all intermediate ZX-diagrams have pseudo-focused gflow.
	 
	\item\label{step:extraction} Denote by $g$ the pseudo-focused gflow (that we are keeping track of) on the current ZX-diagram and by $S$ its set of extracted vertices (which is initially $O$ and will be updated through subsequent steps).
	 Let $M = \{v\in\comp{S}\mid \nexists w\in\comp{S}$ s.t.\ $v\prec_g w\}$, \ie $M$ is the set of maximal non-extracted vertices.
	 Note that if any element of $M$ has a correction set of size~1, it can be added to the set of extracted vertices since all its successors are already in $S$ by definition.
	 We distinguish cases according to the properties of $M$.
	 \begin{enumerate}
	 	\item If $M$ is empty, then $S=V$.
	 	 Thus by Lemma~\ref{lem:trivial-pseudo-focused}, $g$ is an extended causal flow and we are done.
	 	 
	 	\item\label{step:extractable} Suppose there exists $v\in M$ such that $\abs{g(v)}=1$.
	 	 Then $S\cup\{v\}$ is a valid set of extracted vertices as all successors of $v$ are in $S$ by the definition of $M$.
	 	 The number of non-extracted vertices has decreased by one.
	 	 Go back to the beginning of Step~\ref{step:extraction}.
	 	 
	 	\item\label{step:extract-XY} Suppose $\abs{g(v)}>1$ for all elements of $M$ but there exists $v\in M$ such that $\ld(v) = XY$.
	 	 By the definition of $M$ (which excludes correction set members in $\comp{S}$) and the definition of pseudo-focusing (Definition~\ref{def:pseudo-focused}, which excludes correction set members in $S\setminus F$ for non-extracted vertices), we have $g(v)\sse F$.
	 	 Fix a designated element $a\in g(v)$.
	 	 Then for each $b\in g(v)\setminus\{a\}$ (if any), apply Lemma~\ref{lem:insert-cnot} to the pair $a,b$.
	 	 By Corollary~\ref{cor:insert-cnot-simplified}, each of these operations removes the chosen $b$ from the correction set of $v$ (and does not add any new elements), so at the end $g(v) = \{a\}$.
	 	 The partial order among non-extracted vertices does not change.
	 	 Go to Step~\ref{step:extractable}, which can now be applied.
	 	 
	 	\item\label{step:extraction-YZ-helper} Suppose $\abs{g(v)}>1$ and $\ld(v)=YZ$ for all $v\in M$.
	 	 By Lemma~\ref{lem:phase-gadget-extracted-neighbour}, there exists $v\in M$ such that $N_G(v)\cap F\neq\emptyset$.
	 	 Let $a=v$ and pick $b\in N_G(v)\cap F$, then apply Lemma~\ref{lem:pivot-YZ}.
	 	 This lemma does not increase the number of non-extracted vertices and the partial order relationships among non-extracted vertices remain the same.
	 	 Hence $v$ is still maximal among non-extracted vertices, but it is now $XY$-measured.
	 	 We may therefore go to Step~\ref{step:extract-XY} as the size of its correction set has not changed and those of other non-extracted vertices are either the same or have increased.
	 \end{enumerate}
\end{enumerate} 

After the diagram is in phase-gadget form, none of the further operations increase the number of non-extracted vertices and Step~\ref{step:extract-XY} as well as the first subcase of Step~\ref{step:extraction-YZ-helper} decrease this number.
Step~\ref{step:extract-XY} involves at most $\abs{F}$ applications of Lemma~\ref{lem:insert-cnot}, so adds at most $2\abs{F}$ new vertices.
The second subcase of Step~\ref{step:extraction-YZ-helper} is performed at most once for each Step~\ref{step:extract-XY} and adds one new vertex.
Therefore, the algorithm terminates and the size of the ZX-diagram does not increase too much.
Indeed, by Lemma~\ref{lem:frontier-properties} the size of the frontier is always $\abs{F}=\abs{O}$.
Thus, setting $n:=\abs{V}$ and $m:=\abs{O}$, the final size of the diagram is at most $2n(m+1) = \mathcal{O}(n^2)$.
If $\abs{O}$ is small compared to $\abs{V}$ -- which will be the case in many circuit optimisation applications -- then the increase in diagram size is even smaller.
Other optimisations, particularly regarding Step~\ref{step:extract-XY}, should be possible in analogy with the extended gflow circuit extraction algorithm, see e.g.~\cite[Section~5.2]{backens_there_2021}.

\begin{example}\label{ex:circuit-extraction}
	Consider the following ZX-diagram, which is well-known to have gflow but not causal flow \cite[Figure~2]{browne_generalized_2007}:
	\ctikzfig{extraction/extract-ex0}
	The diagram is in phase-gadget form and the focused gflow is given by $g(i_1) = \{o_2,o_3\}$, $g(i_2) = \{o_1,o_2,o_3\}$ and $g(i_3) = \{o_1,o_3\}$.
	The set of extracted vertices is $S = \{o_1,o_2,o_3\}$ and the set of maximal non-extracted vertices is $M = \{i_1,i_2,i_3\}$.
	
	Suppose we pick $v=i_1$ for the purposes of an extraction step using Lemma~\ref{lem:insert-cnot}, then $a=o_3$ as this is the only element of $g(i_1)\cap N_G(i_1)$.
	Then $b=o_2$ and the `insert \& pivot' procedure becomes:
	\[
	 \to\quad
	 \input{extraction/extract-ex1a.tikz} \quad\to\quad
	 \input{extraction/extract-ex1b.tikz} \quad\to\quad
	 \input{extraction/extract-ex1c.tikz} \quad\to\quad
	 \input{extraction/extract-ex2.tikz}
	\]
	where the final step is just rearranging the vertices.
	The resulting diagram has causal flow given by $g(i_1) = o_3$, $g(i_2) = o_1$, $g(i_3) = a$, $g(a) = b$, $g(b) = o_2$, which induces the following partial order (where an arrow $v\to w$ indicates $v R_g w$):
	\ctikzfig{extraction/extract-ex2-order}
	Other choices of $v$ at the beginning of the extraction will lead to different circuits.
\end{example}

\section{Quantum circuits}

We have shown in the previous section that any ZX-diagram with a Pauli flow can be turned into a ZX-diagram with extended causal flow using the flow-preserving rewrite rules.
These rewrite steps are all invertible.
Therefore completeness for ZX-diagrams with Pauli flow reduces to completeness for ZX-diagrams with extended causal flow. 
The close correspondence between ZX-diagrams with extended causal flow and quantum circuits will thus allow us to use a completeness result for quantum circuits to prove completeness for flow-preserving rewriting.
Nevertheless, we will not be able to use existing circuit completeness results `out of the box' as they consider either strictly unitary circuits~\cite{clement_complete_2023} or circuits where ancillas can be both initialised and terminated~\cite{CDP24}.
The first step is therefore to set out a completeness proof for circuits representing isometries, where ancillas can be initialised but not terminated.

We then show how to derive each of the circuit rewrite rules from the flow-preserving ZX-calculus rewrite rules in two steps.
First, we show that mimicking a circuit cut-and-paste rewriting step in the ZX-calculus always preserves extended causal flow.
This means that the proofs of individual rewriting rules in the second part do not need to consider the context: i.e.\ it suffices to show the rule is derivable in principle, without regard to the conditions under which it is flow-preserving.

\subsection{Quantum circuits with initialisation}
\label{s:circuits}

We consider the $\propQC$-circuits generated by the standard gate set: Hadamard $\gH$, Z-rotations $\gP$, CNOT $\gCNOT$ and qubit initialisation $\ginit$, together with the identity \scalebox{0.7}{\begin{tikzpicture}
	\begin{pgfonlayer}{nodelayer}
		\node [style=none] (13) at (1.5, 0) {};
		\node [style=none] (14) at (0, 0) {};
	\end{pgfonlayer}
	\begin{pgfonlayer}{edgelayer}
		\draw (14.center) to (13.center);
	\end{pgfonlayer}
\end{tikzpicture}
}, and swap  \scalebox{0.7}{\input{intro/swap.tikz}}. These generators are combined using sequential and parallel compositions. 

We consider $\propQC$-circuits up to deformation, i.e.:
\begin{center}
	\begin{tabular}{cc}
		$\input{intro/iswap0.tikz}=\input{intro/iswap1.tikz}~\qquad~$&$~\scalebox{0.9}{\input{intro/sdiagp0-circ.tikz}}=\scalebox{0.9}{\input{intro/sdiagp1-circ.tikz}}~$\\
	\end{tabular}	
\end{center}
Any $\propQC$-circuit $C$ can be interpreted as a linear map  $\interp C$ defined as follows.

\begin{definition}[Semantics]\label{def:QCsem}
  For any $\propQC$-circuit $C:n\to m$ with $n$ inputs and $m$ outputs, let $\interp{C}: \mathbb C^{\{0,1\}^n} \to \mathbb C^{\{0,1\}^m}$ be the \emph{semantics} of $C$ inductively defined as the linear map satisfying $\interp{C_2\circ C_1} = \interp{C_2}\circ\interp{C_1}$; $\interp{C_1\otimes C_2} = \interp{C_1}\otimes\interp{C_2}$; and
\[
    \interp{\gempty} = 1\mapsto 1\quad
     \interp \ginit = \ket 0 \quad
     \interp{\gI} = \ket{x}\mapsto \ket{x}\quad
    \interp{\gP} = \ket x\mapsto e^{ix\varphi}\ket{x}\quad
\]\[
    \interp{\gCNOT}  =\ket{x,y}\mapsto  \ket{x,x\oplus y}\qquad
    \interp{\gH} = \ket x \mapsto \frac{\ket{0}+(-1)^x\ket{1}}{\sqrt{2}}\quad
    \interp{\gSWAP} = \ket{x,y}\mapsto \ket{y,x}
\]
\end{definition}

Note that for any $\propQC$-circuit $C$, $\interp C$ is an isometry, i.e. $\interp C^\dagger  \circ \interp C = Id$. Conversely, it is well known that any isometry map acting on a finite number of qubits can be represented by a $\propQC$-circuit, up to a global phase:
\begin{proposition}[Universality]\label{prop:QC-universal}
  $\propQC$ is universal, i.e.~for any isometry $V:\C^{\{0,1\}^n} \to \C^{\{0,1\}^m}$ there exists a $\propQC$-circuit $C:n\to m$ such that $\interp{C}=e^{i\theta}V$ for some $\theta\in [0,2\pi)$.
\end{proposition}

Quantum circuits, as defined above, only have four different kinds of generators, however, it is often convenient to use other gates that can be defined by combining them.
In particular, we use the  multi-controlled phase gate defined as follows: 
\begin{align}\label{mctrlPinducdef}\tikzfigM{./shortcut/mctrlPphi-}=\tikzfigM{./shortcut/mctrlPphidef-}\end{align}

In the above definition, with a slight abuse of notation, dashed lines denote an arbitrary number of control qubits, e.g.\ $\tikzfigS{./shortcut/0+Pphi}:n+1\to n+1 $ or simply $\tikzfigS{./shortcut/0+1Pphi}:n+1\to n+1$ have $n\ge 0$ control qubits (possibly zero), whereas $\tikzfigS{./shortcut/1+Pphi}:n+2\to n+2$ and $\tikzfigS{./shortcut/+1Pphi}:1+n+1\to 1+n+1$ have at least one control qubit.

\begin{figure*}
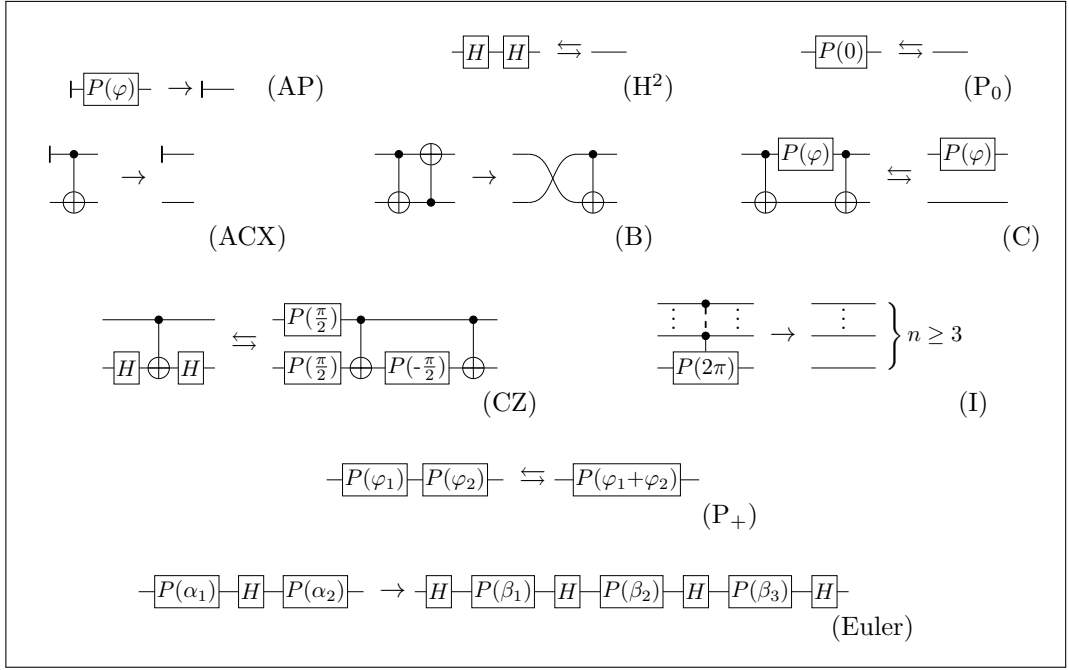

  \fbox{\begin{minipage}{.985\textwidth}\begin{center}
    \vspace{-.5em}
    \hspace{-1.5em}
     \begin{subfigure}{0.28\textwidth}
      \begin{align}\label{initP-ancilla}\tag{AP}\tikzfigM{./qcancilla-axioms/initPphi}\rightarrow\tikzfigM{./qcancilla-axioms/initId}\end{align}
    \end{subfigure}\hspace{3em}
    \begin{subfigure}{0.25\textwidth}
      \begin{align}\label{HH-ugp}\tag{H$^2$}\tikzfigM{./qc-axioms/HH}\leftrightarrows\tikzfigM{./qc-axioms/Id}\end{align}
    \end{subfigure}\hspace{3em}
    \begin{subfigure}{0.24\textwidth}
      \begin{align}\label{P0-ugp}\tag{P$_0$}\tikzfigM{./qc-axioms/P0}\leftrightarrows\tikzfigM{./qc-axioms/Id}\end{align}
    \end{subfigure}

    \hspace{-1.5em}
    \begin{subfigure}{0.27\textwidth}
      \begin{align}\label{initCNOT-ancilla}\tag{ACX}\tikzfigM{./qcancilla-axioms/initCNOT}\rightarrow\tikzfigM{./qcancilla-axioms/initIdId}\end{align}
    \end{subfigure}\hspace{1.3em}
        \begin{subfigure}{0.31\textwidth}
      \begin{align}\label{bigebre-ugp}\tag{B}\tikzfigM{./qc-axioms/CNOTNOTC}\rightarrow\tikzfigM{./qc-axioms/SWAPCNOT}\end{align}
    \end{subfigure}\hspace{1.3em}
    \begin{subfigure}{0.33\textwidth}
      \begin{align}\label{CNOTPCNOT-ugp}\tag{C}\tikzfigM{./qc-axioms/CNOTPphiCNOT}\leftrightarrows\tikzfigM{./qc-axioms/PphiId}\end{align}
    \end{subfigure}
    
        \hspace{-1.5em}
    \begin{subfigure}{0.46\textwidth}
      \begin{align}\label{CZ-ugp}\tag{CZ}\tikzfigM{./qc-axioms/H2CNOTH2}\leftrightarrows\tikzfigM{./qc-axioms/CZ}\end{align}
    \end{subfigure}\hspace{2.5em}
     \begin{subfigure}{0.36\textwidth}
      \begin{align}\label{ctrl2pi-ugp}\tag{I}\tikzfigM{./qc-axioms/mctrlP2pi-}\rightarrow\tikzfigM{./qc-axioms/multi-id-braket}\end{align}
    \end{subfigure}
    
    \hspace{-1.5em}
    \begin{subfigure}{0.45\textwidth}\begin{align}  \tikzfigM{./qcoriginal-axioms/Pphi1Pphi2}\leftrightarrows\tikzfigM{./qcoriginal-axioms/Pphi1phi2}\label{Paddition-prime}\tag{P$_+$}\end{align}\end{subfigure}
    
        \hspace{-1.5em} \begin{subfigure}{0.78\textwidth}
   \begin{align}\label{euler-ugp}\tag{Euler} \tikzfigM{./qcprime-axioms/eulerprimebis-left}\rightarrow\tikzfigM{./qcprime-axioms/eulerprimebis-right}\end{align}
    \end{subfigure}\hspace{1.2em}

    \vspace{.5em}
    
  \end{center}\end{minipage}}
  \caption{\label{fig:qcsf-axioms}Complete rewriting rules for quantum circuits with initialisation. In \cref{euler-ugp}, the angles $\beta_i$ of the RHS are computed as follows: For any angles $\alpha_1,\alpha_2\in\R$, let $z\defeq -\sin\left(\frac{\alpha_1+\alpha_2}{2}\right)+i\cos\left(\frac{\alpha_1-\alpha_2}{2}\right)$,  $z'\defeq \cos\left(\frac{\alpha_1+\alpha_2}{2}\right)-i\sin\left(\frac{\alpha_1-\alpha_2}{2}\right)$ and 
  \\-- If $z'=0$ then  $\beta_1\defeq2\arg(z)$, $\beta_2\defeq0$ and $\beta_3\defeq0$;\\-- If $z=0$ then  $\beta_1\defeq2\arg(z')$, $\beta_2\defeq\pi$ and $\beta_3\defeq0$; \\-- Otherwise  $\beta_1\defeq\arg(z)+\arg(z')$, $\beta_2\defeq2\arg\left(i+\left\lvert\frac{z}{z'}\right\rvert\right)$ and $\beta_3\defeq\arg(z)-\arg(z')$.\\Notice that, in any case, $\beta_1, \beta_3\in [0,2\pi)$ and $\beta_2\in [0,\pi]$. }

\end{figure*} 

We consider the rewrite rules of \QC~of \cref{fig:qcsf-axioms}, which we prove to be sound (up to global phases) and complete:

\begin{theorem}[Completeness]\label{proof:THMcompleteQC} 
\QC~ is sound and complete, i.e.\ for any pair of circuits $C_0:n\to m$ and $C_1: n\to m$,  $\exists \theta\in \left[0,2\pi\right)$ such that $\interp {C_0} = e^{i\theta}\interp{C_1}$ if and only if $C_0\rightarrow^* C_1$ 
 using the rules of \cref{fig:qcsf-axioms} (and the deformation rules). 
\end{theorem}

\begin{proof} Soundness is easy to prove by inspecting that each equation is sound.
	To prove completeness, first we show that \eqref{euler-ugp} $\tikzfigS{./qcprime-axioms/eulerprimebis-left}\!\!\rightarrow\tikzfigS{./qcprime-axioms/eulerprimebis-right}$ can be reversed.
\begin{align*}
\tikzfigS{./qcprime-axioms/eulerprimebis-right}&\rightarrow_{\text{\ref{euler-ugp}}} &&\tikzfigS{./Euler_rev/eulerprimebis-right-1}\\
&\rightarrow_{\text{\ref{HH-ugp},\ref{Paddition-prime}}} &&\tikzfigS{./Euler_rev/eulerprimebis-right-2}\\
&\rightarrow_{\text{\ref{euler-ugp}}} &&\tikzfigS{./Euler_rev/eulerprimebis-right-3}\\
&\rightarrow_{\text{\ref{HH-ugp},\ref{Paddition-prime}}} &&\tikzfigS{./Euler_rev/eulerprimebis-right-4}
\end{align*}
By soundness, $U = \interp{\tikzfigS{./Euler_rev/eulerprimebis-right-4}}$ is equal to $V= \interp{\tikzfigS{./qcprime-axioms/eulerprimebis-left}}$ up to a global phase $e^{i\theta}$. In particular $e^{i\theta}\bra 0 U\ket 0 =\bra 0 V \ket 0$, so $e^{i\theta}\sqrt 2 = 1+e^{i\delta_2}$. As $\delta_2\in [0,\pi]$ by definition of the Euler angles, we get $\delta_2 = \frac \pi 2 \bmod{2\pi}$ and $\theta = \frac \pi 4\bmod {2\pi}$. Hence,
\begin{align*}
\tikzfigS{./Euler_rev/eulerprimebis-right-5}
\end{align*}\begin{align*}
&\qquad\rightarrow_{\text{\ref{Paddition-prime}}}&&\tikzfigS{./Euler_rev/eulerprimebis-right-6}\\
&\qquad\rightarrow_{\text{\ref{euler-ugp}}} &&\tikzfigS{./Euler_rev/eulerprimebis-right-7}\\
&\qquad\rightarrow_{\text{\ref{P0-ugp},\ref{HH-ugp},\ref{Paddition-prime}}} &&\tikzfigS{./Euler_rev/eulerprimebis-right-8}
\end{align*}
Again by soundness, we must have $\alpha_1 = \gamma_1+\delta_1+\frac {3\pi} 2$ and $\alpha_2= \delta_3-\frac \pi 2$ which achieves the derivation of the reverse of \cref{euler-ugp}.  

Eq.~(\ref{initP-ancilla})  can also be reversed: $\tikzfigS{./qcancilla-axioms/initId}\rightarrow_{\text{\ref{P0-ugp},\ref{Paddition-prime}}}  \tikzfigS{./qcancilla-axioms/initPphi-phi} \rightarrow_{\text{\ref{initP-ancilla}}}  \tikzfigS{./qcancilla-axioms/initPphi}$. Similarly for \cref{initCNOT-ancilla}, using \cref{Paddition-prime} and \cref{CNOTPCNOT-ugp}. Moreover \cref{bigebre-ugp} can also be reversed by introducing two CNOTs at the end of the circuit, thanks to \cref{P0-ugp} and \cref{CNOTPCNOT-ugp}, and then applying \cref{bigebre-ugp} from left to right.

Now we show that \cref{ctrl2pi-ugp} can be reversed. It has been proved in \cite{CDP24} that all equations  of \cref{fig:qcsf-axioms} (including  \cref{euler-ugp} and \cref{bigebre-ugp}  in both directions) but without \cref{initP-ancilla}, \cref{initCNOT-ancilla}, and \cref{ctrl2pi-ugp} are enough to prove that the composition of a $\varphi_1$- and a $\varphi_2$-multi-controlled phase gate is equivalent to a $(\varphi_1+\varphi_2)$-multicontrolled phase gate; and that the $4\pi$-multicontrolled phase gate is equivalent to the identity. Thus, to simulate the application of \cref{ctrl2pi-ugp} from right to left, one can introduce a $4\pi$-multicontrolled phase gate, split it into 2 instances of $2\pi$-multicontrolled phase gate and then use  \cref{ctrl2pi-ugp} to remove one of them.

Thus every rule of \cref{fig:qcsf-axioms} can be reversed, so we consider in the rest of the proof that every rule can be applied in both directions.
Let $E_0$ be the equational theory made of all equations of \cref{fig:qcsf-axioms} except \cref{initP-ancilla} and \cref{initCNOT-ancilla}.
In \cite{CDP24}, it was proved that $E_0$ is complete for unitary circuits, i.e.\ circuits of the form $C:n\to n$. Moreover, in \cite{CDPV24}, it has been proved that an equational theory $E_1$ made of equations \cref{initP-ancilla}, \cref{initCNOT-ancilla} together with a family $\mathcal F$ of equations involving only unitary circuits, is complete for quantum circuits with qubit initialisation. Note that every equation of $\mathcal F$ can be derived using the equations of \cref{fig:qcsf-axioms} by the completeness of its unitary fragment $E_0$. Thus any derivation based on $\mathcal F$ can also be done using the equations of \cref{fig:qcsf-axioms}.
\end{proof}

\subsection{Translation between circuits with initialisation and ZX-diagrams with extended causal flow}

Consider the following translation from the quantum circuit generators to ZX-diagrams, with the empty diagram mapping to an empty diagram:
\begin{align*}
	\gI \quad&\mapsto\quad \begin{tikzpicture}
	\begin{pgfonlayer}{nodelayer}
		\node [style=Z dot] (0) at (0, 0) {};
		\node [style=none] (1) at (-0.5, 0) {};
		\node [style=none] (2) at (0.5, 0) {};
	\end{pgfonlayer}
	\begin{pgfonlayer}{edgelayer}
		\draw (1.center) to (2.center);
	\end{pgfonlayer}
\end{tikzpicture}
 &
	\gH \quad&\mapsto\quad \begin{tikzpicture}
	\begin{pgfonlayer}{nodelayer}
		\node [style=Z dot] (0) at (-0.5, 0) {};
		\node [style=Z dot] (1) at (0.5, 0) {};
		\node [style=none] (3) at (-1, 0) {};
		\node [style=none] (4) at (1, 0) {};
	\end{pgfonlayer}
	\begin{pgfonlayer}{edgelayer}
		\draw (3.center) to (0);
		\draw (1) to (4.center);
		\draw [style=hadamard edge] (0) to (1);
	\end{pgfonlayer}
\end{tikzpicture}
 &
	\gP \quad&\mapsto\quad \input{zx-circuits/phase-gate.tikz} \\
	\gSWAP \quad&\mapsto\quad \input{zx-circuits/swap.tikz} &
	\gCNOT \quad&\mapsto\quad \input{zx-circuits/cnot.tikz} &
	\ginit \quad&\mapsto\quad \begin{tikzpicture}
	\begin{pgfonlayer}{nodelayer}
		\node [style=Z dot] (0) at (-0.5, 0) {};
		\node [style=Z dot] (1) at (0.5, 0) {};
		\node [style=none] (4) at (1, 0) {};
	\end{pgfonlayer}
	\begin{pgfonlayer}{edgelayer}
		\draw (1) to (4.center);
		\draw [style=hadamard edge] (0) to (1);
	\end{pgfonlayer}
\end{tikzpicture}

\end{align*}
Parallel composition is the usual juxtaposition, sequential composition merges not just the dangling plain wires but also the spiders connected by these wires.

\begin{example}\label{ex:translation}
	Consider the left-hand side of \eqref{CNOTPCNOT-ugp}.
	This can be decomposed as:
	\[
		\tikzfigM{./qc-axioms/CNOTPphiCNOT} \quad=\quad \gCNOT \circ (\gP\otimes\gI) \circ \gCNOT
	\]
	The corresponding ZX-translation is therefore:
	\begin{align*}
		\input{zx-circuits/cnot.tikz} \circ (\input{zx-circuits/phase-gate.tikz} \otimes ) \circ \input{zx-circuits/cnot.tikz}
		\quad&=\quad \input{zx-circuits/example0.tikz} \\
		\quad&=\quad \input{zx-circuits/example1.tikz}
	\end{align*}
	which is in MBQC form after merging the spiders connected by plain wires.
\end{example}

This translation preserves the interpretation, as can straightforwardly be checked on a generator-by-generator basis.
It is well-known that standard methods of translating a circuit into the ZX-calculus result in ZX-diagrams with extended causal flow (once the diagrams are brought into MBQC-form); the same holds for our translation method which already ensures the ZX-diagram is in MBQC-form.

We would like to show in a local way that the ZX-translations of each of the rules of Figure~\ref{fig:qcsf-axioms} can be derived from the rules of Figure~\ref{fig:flow-preserving}: i.e.\ where a ZX-calculus rule preserves flow only under certain conditions, we would like to be able to verify that the conditions are met by considering a small subdiagram rather than having to reason about the global context.
To achieve this, we split the derivation of the translated circuit rewrite rules into two parts: Firstly, we show that, given a certain relationship between a subdiagram and its context, the subdiagram can be replaced by any other diagram which has extended causal flow without breaking the global flow.
Secondly, we show that each of the circuit rewrite rules can be derived from the flow-preserving rewrite rules without worrying about the context.
Putting the two results together then means that the derivations will also work in the context of a larger diagram.

In a quantum circuit, it is well understood when a rewrite rule can be applied, even if it might act on not necessarily connected parts of the circuit, such as the right-to-left direction of \eqref{CNOTPCNOT-ugp}.
The following definitions and proposition achieve a similar process in the ZX-calculus by isolating a subdiagram $D$ and splitting its context into a `past' and a `future' in such a way that $D$ can be replaced without breaking the causal flow, as visualised below:
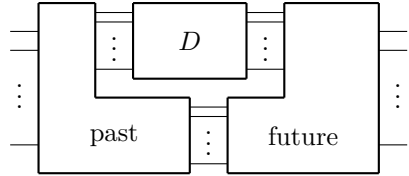
\ctikzfig{zx-circuits/D-past-future}

First, we will need the following definition of the `inputs' and `outputs' of a subset of vertices.

\begin{definition}\label{def:subdiagram-inputs-outputs}
	Consider a labelled open graph $(G,I,O,\ld)$ with extended causal flow $g$.
	Let $W\sse V$ be a non-empty subset of the vertices.
	Define the \emph{set of outputs of $W$} as
	\begin{align*}
		O_W &:= \{v\in W \mid v\in O \vee g(v)\notin W\} \\
		\intertext{and choose a \emph{set of inputs $I_W$ of $W$} satisfying $I_{W,min}\sse I_W \sse I_{W,max}$, where:}
		I_{W,min} &:= (W\cap I) \cup \{v\in W\mid \exists w\in V\setminus W \text{ s.t. } g(w)=v\} \\
		I_{W,max} &:= \{v\in W\mid \nexists w\in W \text{ s.t. } g(w)=v\}.
	\end{align*}
\end{definition}

It is straightforward to see that no element $v\in I_{W,min}$ corrects any element of $W$: inputs do not appear in correction sets at all, and in the second case of the definition of $I_{W,min}$, $v$ corrects some element of $V\setminus W$, therefore it cannot also correct an element of $W$.
Thus $I_{W,min}\sse I_{W,max}$ and there always exists some valid choice for $I_W$.
The idea behind this definition is that in a case such as Example~\ref{ex:ancilla-past} below, which corresponds to an ancilla initialisation followed by a Hadamard gate, we may sometimes want to include the ancilla initialisation in the rewriting process and at other times rewrite only the unitary part of the diagram.

\begin{lemma}\label{lem:subdiagram-flow}
	Consider a labelled open graph $(G,I,O,\ld)$ with extended causal flow $g$.
	Let $W\sse V$ be a non-empty subset of the vertices with inputs $I_W$ and outputs $O_W$ as given in Definition~\ref{def:subdiagram-inputs-outputs}.
	Define $G'$ be the subgraph induced by $W$ and define $\ld', g'$ to be the restrictions of $\ld$ and $g$ to domain $W\setminus O_W$.
	Then $g'$ is an extended causal flow on $(G',I_W,O_W,\ld')$.
\end{lemma}
\begin{proof}
	First, note that $(G',I_W,O_W,\ld')$ is a valid \LOG.
	Now, for any $v\in W\setminus O_W$, we have $g'(v)\in W$ by the definition of $O_W$ and	$g'(v) \notin I_W$ by the property $I_W\sse I_{max}$ of Definition~\ref{def:subdiagram-inputs-outputs}.
	Thus $g': W\setminus O_W \to W\setminus I_W$ as desired.
	Therefore the $g'$ inherits the property of being an extended causal flow from $g$:
	The relation $R_{g'}^*$ is a subrelation of $R_g^*$, so it is also a partial order.
	Similarly, the conditions \ref{itm:cXY} and \ref{itm:cYZ} are satisfied for $g'$ because they are satisfied for $g$ and $G'$ is an induced subgraph.
	Thus  $g'$ is an extended causal flow on $(G',I_W,O_W,\ld')$, as desired.
\end{proof}

\begin{observation}
	As edges between inputs and edges between outputs do not affect whether there is a flow, $G'$ need not be the full induced subgraph: any subset of edges in $(I_W\times I_W) \cup (O_W\times O_W)$ can be left out without changing the result of Lemma~\ref{lem:subdiagram-flow}.
\end{observation}

Note that, in standard ZX-calculus rewriting, we `cut' out a subdiagram to be replaced by something else by cutting the edges connecting it to its context.
In the context of circuit-like flow-preserving rewriting, we instead want to `cut' the vertices that connect two diagrams by unfusing them, an inverse of the spider fusion that forms part of the translation from circuits to ZX-diagrams shown in Example~\ref{ex:translation}.
Instead of partitioning the vertices of the ZX-diagram into non-intersecting sets, we will therefore have subdiagrams that share vertices at their interfaces, i.e.\ at their inputs and outputs.
On the other hand, we require that each edge belongs to at least one of the subdiagrams, in the sense that both of its endpoints are within that subdiagram.
Additionally, a certain compatibility with the induced partial order of the flow will be needed.

\begin{definition}\label{def:causal-subdiagram}
	Consider a ZX-diagram $C$ corresponding to a labelled open graph $(G,I,O,\ld)$ with extended causal flow $g$.
	We say a subdiagram $D$ on vertices $V_D\sse V$ is a \emph{causal subdiagram} if there are two additional sets $V_\past,V_\fut\sse V$ such that:
	\begin{enumerate}
		\item\label{it:future} $V_\fut$ contains $O$ and is upwards closed according to $\prec_g$.
		\item\label{it:past} $V_\past$ contains $I$ and is downwards closed according to $\prec_g$.
		\item\label{it:edges} For every edge $\{u,v\}\in E$ there exists a set $W\in\{V_D,V_\past,V_\fut\}$ such that $u,v\in W$.
		(This $W$ need not be unique.)
		\item\label{it:overlap} The three sets overlap exactly on their inputs and outputs, in the sense that:
			\begin{itemize}
				\item $V_\past \cap (V_D\cup V_\fut) = O_\past$,
				\item $V_D\cap V_\past = I_D$,
				\item $V_D\cap V_\fut = O_D$, and
				\item $V_\fut\cap (V_\past\cup V_D) = I_\fut$.
			\end{itemize}
	\end{enumerate}
\end{definition}

Condition~\ref{it:edges} implies that every non-isolated vertex is contained in $V_\past\cup V_D\cup V_\fut$.
It is possible for a labelled open graph with extended causal flow to contain an isolated vertex; this vertex must be either an output or $YZ$-measured.
Outputs are contained in $V_\fut$ by definition, yet an isolated $YZ$-measured qubit need not be contained in any of the three sets $V_\past, V_D, V_\fut$ of the above definition.
Any isolated $YZ$-measured qubit only yields a scalar factor -- it does not affect the linear operation being implemented.
We may therefore assume without loss of generality that there are no isolated $YZ$-measurement and hence $V_\past, V_D,$ and $V_\fut$ together cover $V$.

\begin{observation}
	A subcircuit of the kind that could have a circuit rewrite rule applied to it will translate to a causal subdiagram in the ZX-calculus.
	
	To rewrite a subcircuit $C$ of some larger circuit, the large circuit must be decomposable as $C_\fut \circ (C\otimes C_\para)\circ C_\past$.
	Without loss of generality, we may even assume that $C_\para$ is just a tensor product of identities.
	Then the `past' and `future' parts of the decomposition become $V_\past$ and $V_\fut$ when translating into the ZX-calculus, and they satisfy Properties~\ref{it:future}, \ref{it:past} and~\ref{it:overlap} because of the sequential composition structure.
	Moreover, Property~\ref{it:edges} follows from the way composition was defined in the translation to ZX.
\end{observation}

\begin{example}\label{ex:unfuse}
	Consider the diagram constructed in Example~\ref{ex:translation}:
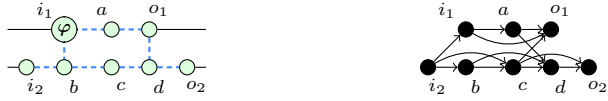
	\[
		\input{zx-circuits/compatible0.tikz} \qquad\qquad\qquad\qquad
		\input{zx-circuits/compatible0-order.tikz}
	\]
	which has causal flow $g(i_1) = a, g(i_2) = b, g(a) = o_1, g(b) = c, g(c) = d, g(d) = o_2$ with the partial order given on the right.
	Then $\{i_1,a,o_1\}$ is a causal subdiagram, as is $\{a,c,d\}$; the corresponding ways of unfusing the diagram into three parts are shown below:
	\[
		\input{zx-circuits/compatible0-unfused.tikz} \qquad\qquad
		\input{zx-circuits/compatible0-ex2.tikz}
	\]
	On the other hand, the sets $\{i_2,b,c\}$ or $\{i_2,c\}$ or $\{a,o_2\}$ marked with dashed lines in the below diagrams are not causal:
	\[
		\input{zx-circuits/compatible-counterex1.tikz} \qquad\qquad
		\input{zx-circuits/compatible-counterex2.tikz} \qquad\qquad
		\input{zx-circuits/compatible-counterex3.tikz}
	\]
	In the first case, to have both endpoints of the edge $\{i_1,b\}$ in the same set (as required by Condition~\ref{it:edges}), we need $b\in V_\past\cup V_\fut$ because $i_1\notin V_D$.
	\begin{itemize}
		\item If $b\in V_\past$, then $i_2\in V_\past\cap V_D$ but $i_2\notin O_\past$, contradicting Condition~\ref{it:overlap}.
		\item If $b\in V_\fut$, then $c\in V_\fut\cap V_D$ but $c\notin I_\fut$, again contradicting Condition~\ref{it:overlap}.
	\end{itemize}
	The argument is analogous for the second case.
	
	In the third case, consider Condition~\ref{it:edges} applied to the edge $\{o_1,d\}$.
	\begin{itemize}
		\item If $o_1,d\in V_\past$, then $a\in V_\past$ by the closure property of Condition~\ref{it:past}.
		Yet then $a\in V_D\cap V_\past$ and $a\notin O_\past$, contradicting Condition~\ref{it:overlap}.
		\item If $o_1,d\in V_\fut$, then $o_2\in V_\fut\cap V_D$ but $o_2\notin I_\fut$, again contradicting Condition~\ref{it:overlap}.
	\end{itemize}
\end{example}

\begin{example}\label{ex:ancilla-past}
	Note that $V_\past$ may be empty if $I=\emptyset$.
	For example, consider the path graph on three vertices $\{a,b,c\}$ with $I=\emptyset$ and $O=\{c\}$: \begin{tikzpicture}
	\begin{pgfonlayer}{nodelayer}
		\node [style=Z dot] (0) at (-1, 0) {};
		\node [style=Z dot] (1) at (1, 0) {};
		\node [style=none] (2) at (1.75, 0) {};
		\node [style=Z dot] (3) at (0, 0) {};
	\end{pgfonlayer}
	\begin{pgfonlayer}{edgelayer}
		\draw [style=hadamard edge] (0) to (1);
		\draw (1) to (2.center);
	\end{pgfonlayer}
\end{tikzpicture}

	If $V_D = \{a,b\}$, then $I_{min} = \emptyset$ and $I_{max} = \{a\}$.
	There are thus two options for $I_D$ (and both lead to causal subdiagrams) -- either $I_D=\emptyset$ or $I_D = \{a\}$.
	The first of those options yields an empty $V_\past$:
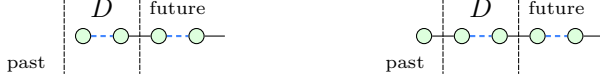
	\[
		\input{zx-circuits/ancilla-start-unfuse1.tikz} \qquad\qquad\qquad
		\input{zx-circuits/ancilla-start-unfuse2.tikz}
	\]
\end{example}

\begin{example}
	For a given set of vertices, $D$ must contain all edges incident on any internal vertex, i.e.\ on any vertex $v\in V_D\setminus (I_D\cup O_D)$ and all edges in $I_D\times O_D$.
	Edges between inputs or edges between outputs can either be part of $D$ or they can be unfused into the `past' or `future'.
	For example, in the diagram of Examples~\ref{ex:translation} and~\ref{ex:unfuse}, if $V_D = \{i_1,b\}$, then there are three different possible decompositions depending on whether the edge $\{i_1,b\}$ becomes part of the past, of $D$ itself, or of the future:
	\[
		\input{zx-circuits/compatible-edges1.tikz} \qquad\quad
		\input{zx-circuits/compatible-edges2.tikz} \qquad\quad
		\input{zx-circuits/compatible-edges3.tikz}
	\]
\end{example}

\begin{definition}\label{def:replacement}
	Consider a ZX-diagram $C$ corresponding to a labelled open graph $(G,I,O,\ld)$ with extended causal flow $g$.
	Suppose $D$ is a causal subdiagram of $C$.
	Let $D'$ be another diagram that has extended causal flow and satisfies $I_{D'} = I_D$, $O_{D'} = O_D$.
	Define a diagram $C'$ by unfusing $D$ from its past and future, replacing it by $D'$, and then re-fusing.
\end{definition}

\begin{example}
	Consider the following diagram (where we have left out phase labels for simplicity and because they do not affect the argument), again the partial order corresponding to the extended causal flow is shown on the right:
	\begin{center}
		\input{zx-circuits/compatible-new.tikz} \qquad\qquad\qquad\qquad
		\input{zx-circuits/compatible-new-order.tikz}
	\end{center}
	Take $V_D = \{a\}$ with $V_\past = \{i_1,i_2,a,o_2\}$ and $V_\fut = \{a,o_1,o_2\}$, then $D$ is a causal subdiagram and can be unfused a below:
	\ctikzfig{zx-circuits/compatible-new-decomposed}
	Replacing the subdiagram \begin{tikzpicture}
	\begin{pgfonlayer}{nodelayer}
		\node [style=Z dot] (0) at (0, 0) {};
		\node [style=none] (1) at (-0.75, 0) {};
		\node [style=none] (2) at (0.75, 0) {};
	\end{pgfonlayer}
	\begin{pgfonlayer}{edgelayer}
		\draw (1.center) to (2.center);
	\end{pgfonlayer}
\end{tikzpicture}
 by \input{zx-circuits/triple.tikz} then leads to:
	\[
	\input{zx-circuits/compatible-new-decomposed.tikz} \quad\mapsto\quad
	\input{zx-circuits/compatible-new-replaced.tikz} \quad = \quad
	\input{zx-circuits/compatible-new-merged.tikz}
	\]
\end{example}

\begin{proposition}
	Consider a ZX-diagram $C$ corresponding to a labelled open graph $\Gamma = (G,I,O,\ld)$ with extended causal flow $g$.
	Suppose $D$ is a causal subdiagram of $C$, and $D'$ is a diagram that has extended causal flow $f$ and satisfies $I_{D'} = I_D$, $O_{D'} = O_D$.
	Let $\Gamma'= (G',I,O,\ld')$ be the labelled open subgraph that results from $\Gamma$ when $D$ is replaced by $D'$ according to Definition~\ref{def:replacement}.
	Then $\Gamma'$ has extended causal flow.
\end{proposition}
\begin{proof}
	Write $V' = (V\setminus V_D) \cup V_{D'}$.
	Let $V_\past, V_\fut$ be two sets that witness $D$ being a causal subdiagram according to Definition~\ref{def:causal-subdiagram}.
	Define a correction function on $\Gamma'$ as
	\[
		g'(v) =
		\begin{cases}
			g(v) &\text{if } v\in (V_\past\setminus O_\past) \cup (V_\fut\setminus O) \\
			f(v) &\text{if } v\in V_{D'}\setminus O_{D'}.
		\end{cases}
	\]
	This is indeed a function $g':V'\setminus O\to V'\setminus I$ because $O$ is contained in $V_\fut$ and $g,f$ are both flows, so do not contain any inputs in their codomains.
	Moreover, \ref{itm:cXY} and \ref{itm:cYZ} hold as both $f$ and $g$ are extended causal flows, so all that remains is to check whether the induced relation is a partial order.
	For easier legibility, we will use the symbol $\prec_{g'}$ even though we have not yet proved this is a valid strict partial order.
	
	Suppose for a contradiction that there exist $v,w\in V$ such that $v\prec_{g'} w \prec_{g'} v$.
	There can be no loop in the partial order that is fully contained in $D'$ since $f$ is an extended causal flow on $D'$.
	Similarly, there can be no loop that is fully contained in $\Gamma'\setminus D'$, as this would imply a loop in $\prec_g$ on $\Gamma$, yet $g$ is an extended causal flow on $\Gamma$.
	Hence assume without loss of generality that $v\in V_{D'}\setminus O_{D'}$ and $w\in V'\setminus V_{D'}$.
	
	Now $w\prec_{g'} v$ means there exists a sequence $w = w_0 R_{g'} w_1 R_{g'} \ldots R_{g'} w_k = v$ for some positive integer $k$ and $v\prec_{g'} w$ means there exists a sequence $v = v_0 R_{g'} v_1 R_{g'} \ldots R_{g'} v_\ell = w$ for some positive integer $\ell$.
	By relabelling which vertices in the cycle are $v$ and $w$, we may further assume without loss of generality that $w_1\in V_{D'}\setminus O_{D'}$ and that $v_1\notin V_{D'}\setminus O_{D'}$.
	
	Consider first the part $v\prec_{g'} v_1 \preceq_{g'} w$, where $v\in V_{D'}\setminus O_{D'}$ and $v_1\notin V_{D'}\setminus O_{D'}$.
	There are two cases that induce this relationship: either $v_1 = g'(v) = f(v)$ or $v_1\in N_{G'}(f(v))$.
	\begin{itemize}
		\item Suppose $v_{1} = f(v)$, then we must have $v_1\in V_{D'}$.
		Yet $v_1\notin V_{D'}\setminus O_{D'}$ by assumption, therefore $v_1 \in O_{D'} = O_D$.
		Now, $O_D \sse V_\fut$ by Condition~\ref{it:overlap} of Definition~\ref{def:causal-subdiagram}, so $w\in V_\fut$ by upwards closure of $V_\fut$ (cf.\ Condition~\ref{it:future}).
		
		\item Suppose $v_{1} \in N_{G'}(f(v))\setminus\{v\}$; there are two subcases depending on whether $v_1\in O_{D'}$.
		\begin{itemize}
			\item If $v_1\in O_{D'}$, then the argument is the same as in the previous case.
			
			\item If $v_1\notin O_{D'}$ and thus $v_1\notin V_{D'}$, then $f(v)$ has a neighbour $v$ inside $D'$ and a neighbour $v_1$ outside $D'$, so $f(v)\in V_{D'}\cap (V_\past\cup V_\fut) = I_{D'}\cup O_{D'}$.
			Since $f(v)$ corrects an element of $V_{D'}$, it cannot be an input of $D'$; in other words $f(v)\in O_{D'}\setminus I_{D'}$.
			By Condition~\ref{it:overlap}, $v_1\in V_\fut$ and the argument proceeds as in the first case.
		\end{itemize}
	\end{itemize}
	Hence, $v\prec_{g'} w$ implies $w\in V_\fut$.
	
	Consider now the part $w\prec_{g'} w_1 \preceq_{g'} v$, where $w\notin V_{D'}$ and $w_1\in V_{D'}\setminus O_{D'}$.
	There are again two cases that induce this relationship, either $w_1 = g'(w) = g(w)$ or $w_1\in N_{G'}(g(w))$.
	Note that $g(w)$ must be a vertex in the original labelled open graph $\Gamma$, it cannot be one of the newly introduced vertices in $V_{D'}\setminus (I_{D'}\cup O_{D'})$.
	\begin{itemize}
		\item Suppose $w_1 = g(w)$, then $w_1$ is a vertex in the original labelled open graph $\Gamma$.
		Thus, $w\in V_\fut$ implies $w_1\in V_\fut$ by upwards closure.
		But $V_\fut\cap V_{D'} = V_\fut \cap V_D = O_D = O_{D'}$, so this means $w_1\in O_{D'}$, contradicting the assumption $w_1\in V_{D'}\setminus O_{D'}$.
		
		\item Suppose $w_1\in N_{G'}(g(w))$.
		If $w_1\in I_{D'}\cup O_{D'}$, the argument proceeds as in the previous case.		
		Hence instead we must have $w_1\in V_{D'}\setminus (I_{D'}\cup O_{D'})$, so that $w_1$ is not a vertex in $\Gamma$.
		Then $g(w)$ has a neighbour in $V_{D'}\setminus (V_\past \cup V_\fut)$, so $g(w)\in V_{D'}$.
		This means $g(w)\neq w$ because $w\notin V_{D'}$, and thus $g(w)\sim w$.
		Therefore $g(w)$ has neighbours both inside and outside of $V_{D'}$, which means $g(w)\in I_{D'}\cup O_{D'} = I_D\cup O_D$: in other words, $g(w)$ appears in the original labelled open graph and $g(w)\in V_D$.
		
		Analogous to the previous case, $w\in V_\fut$ implies $g(w)\in V_\fut$.
		By Condition~\ref{it:overlap}, we have $V_D\cap V_\fut \sse I_\fut$, yet we just showed $g(w)\in V_D\cap V_\fut$, and $g(w)\notin I_\fut$ since $w\in V_\fut$.
		This is a contradiction.
	\end{itemize}
	Both cases lead to contradictions, so the original assumption that there exists a cycle in the partial order must have been false.
	Hence there can be no cycles, meaning $\prec_{g'}$ is indeed a strict partial order and $g'$ is an extended causal flow.
\end{proof}

This lemma means we can check the flow-preservation conditions only within the small diagram we are rewriting without having to worry about the context.

\subsection{Translations of the rewrite rules}

We can now translate the equational theory of Figure~\ref{fig:qcsf-axioms} into the ZX-calculus and show that they follow from the flow-preserving rewrite rules of Figure~\ref{fig:flow-preserving}.
The equation \eqref{ctrl2pi-ugp} is more complicated than the others as it makes use of the syntactic sugar that is the multiply-controlled phase gate.
We will thus prove it individually first.

\begin{lemma}\label{lem:zx-ctrl2pi-ugp}
The ZX-translation of \eqref{ctrl2pi-ugp} follows from the flow-preserving rewrite rules.
\end{lemma}
\begin{proof}
We first show that the ZX-translation of the multiply-controlled phase gate on $n$ qubits with phase $\varphi$ is equivalent to the full spider nest on $n$ qubits with phase $\varphi$ (cf.\ Definition~\ref{def:full-spider-nest}).

The proof is by induction and the base case is the (uncontrolled) phase gate on one qubit:
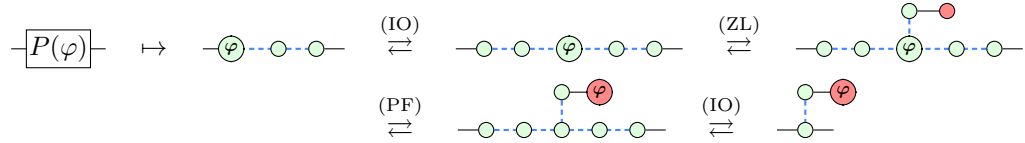
\begin{align*}
	\begin{tikzpicture}
	\begin{pgfonlayer}{nodelayer}
		\node [style=none] (0) at (-1.25, 0) {};
		\node [style=none] (1) at (1.25, 0) {};
		\node [style=gate] (2) at (0, 0) {$P(\varphi)$};
	\end{pgfonlayer}
	\begin{pgfonlayer}{edgelayer}
		\draw (0.center) to (1.center);
	\end{pgfonlayer}
\end{tikzpicture}

	\quad\mapsto\quad \input{zx-circ-deriv/gadg1-0.tikz}
	\quad\overset{\eqref{eq:IO}}&{\rightleftarrows}\quad \input{zx-circ-deriv/gadg1-1.tikz}
	\quad\overset{\eqref{eq:ZL}}{\rightleftarrows}\quad \input{zx-circ-deriv/gadg1-2.tikz} \\
	\quad\overset{\eqref{eq:PF}}&{\rightleftarrows}\quad \input{zx-circ-deriv/gadg1-3.tikz}
	\quad\overset{\eqref{eq:IO}}{\rightleftarrows}\quad \input{zx-circ-deriv/gadg1-4.tikz}
\end{align*}
The $YZ$-insertion with a single $XY$-measured neighbour in the second step is always flow-preserving by Corollary~\ref{cor:YZ-insertion}.
The pivot in the third step does not change the structure of the diagram, so it, too, preserves the extended causal flow.

Now suppose that multi-controlled phase gates on $n$ qubits have the desired form and consider a multi-controlled phase gate on $n+1$ qubits.
Using the inductive definition of the multi-controlled phase gate, we have the following ZX-diagram, where $\alpha_s = (-1)^{\abs{s}-1}\frac{\varphi}{2^{n}}$ for all $s\in\{0,1\}^n\setminus\{00\ldots 0\}$.
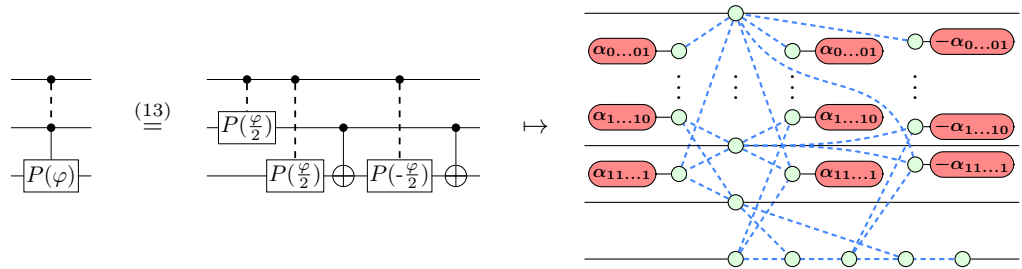
\begin{align*}
	\tikzfigM{./shortcut/mctrlPphi-}
	\quad\overset{\eqref{mctrlPinducdef}}&{=}\quad \tikzfigM{./shortcut/mctrlPphidef-}
	\quad\mapsto\quad \input{zx-circ-deriv/mctrl-sn0.tikz}
\end{align*}
Note that to simplify the presentation, the qubits here are labelled from bottom to top: \ie in all three spider nests, the last $n-1$ indices refer to the top $n-1$ qubits.
Moreover, in the leftmost spider nest, the first bit in the bit string $s$ refers to the penultimate qubit, while in the other two families, the first bit refers to the last qubit as the penultimate qubit does not participate in those spider nests.

Pivot-and-deleting the second and third qubit on the bottom line yields
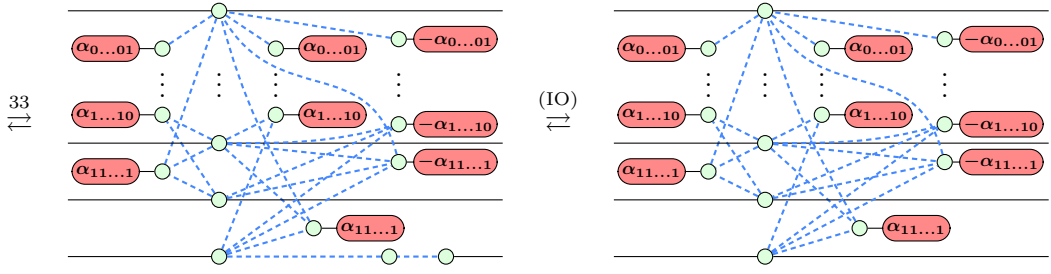
\[
\overset{\ref{lem:pivot-delete}}{\rightleftarrows}\quad  \input{zx-circ-deriv/mctrl-sn1.tikz}
\quad\overset{\eqref{eq:IO}}{\rightleftarrows}\quad  \input{zx-circ-deriv/mctrl-sn2.tikz}
\]
where now in the third spider nest, each gadget is either connected to the penultimate and last qubits (if the first bit of its label is 1) or it is connected to neither of the last two qubits (if the first bit of the label is 0).
It is straightforward to see that the pivot-and-delete operation preserves the extended causal flow.
Thus we have the following situation:
\begin{itemize}
	\item For any $t\in\{0,1\}^{n-1}\setminus\{00\ldots 0\}$, there are three phase gadgets connected to the corresponding subset of the top $n-1$ qubits (with the last bit corresponding to the topmost qubit).
	The phase gadgets from the first two spider nests have phases $\alpha_{0t}$ whereas the gadget from the third spider nest has phase $-\alpha_{0t}$.
	By Lemma~\ref{lem:phase-gadget-addition}, the three phase gadgets can be merged into a single gadget with phase $\alpha_{0t} = (-1)^{\abs{t}-1}\frac{\varphi}{2^{n}}$.
	
	\item For any $t\in\{0,1\}^{n-1}$, there is exactly one phase gadget connected to the corresponding subset of the first $n-1$ qubits plus the $n$-th qubit.
	This gadget comes from the first spider nest and has phase $\alpha_{1t} = (-1)^{\abs{t}+1-1}\frac{\varphi}{2^{n}}$.
	
	\item Similarly, for any $t\in\{0,1\}^{n-1}$, there is exactly one phase gadget connected to the corresponding subset of the first $n-1$ qubits plus the $(n+1)$-th qubit.
	This gadget comes from the second spider nest and has phase $\alpha_{1t} = (-1)^{\abs{t}+1-1}\frac{\varphi}{2^{n}}$.
	
	\item Finally, for any $t\in\{0,1\}^{n-1}$, there is exactly one phase gadget connected to the corresponding subset of the first $n-1$ qubits plus both the $n$-th qubit and $(n+1)$-th qubit.
	This gadget comes from the third spider nest and has phase $-\alpha_{1t} = (-1)^{\abs{t}+2-1}\frac{\varphi}{2^{n}}$.
\end{itemize}
By inspection, these gadgets and phases exactly combine to the full spider nest of phase $\varphi$ on $n+1$ qubits.
This proves the inductive step and thus the desired relationship between multi-controlled phase gates and full spider nests.

Rule \eqref{ctrl2pi-ugp} then follows directly from \eqref{eq:SN}:
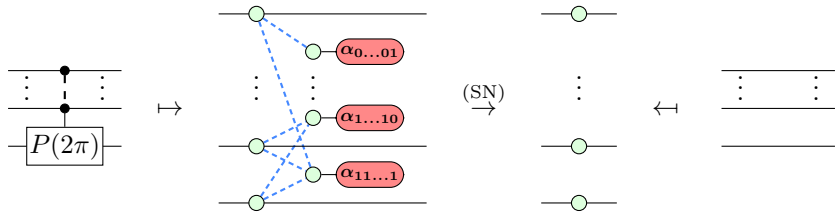
\begin{align*}
	\input{./qc-axioms/mctrlP2pi-.tikz} \quad\mapsto\quad
	\input{zx-circ-deriv/mctrl0.tikz}
	\quad\overset{\eqref{eq:SN}}{\to}\quad \input{zx-circ-deriv/mctrl1.tikz}
	\quad\mapsfrom\quad \input{./qc-axioms/multi-id.tikz}
\end{align*}
where here now $\alpha_s = (-1)^{\abs{s}-1} \frac{2\pi}{2^{n}}$ for all $s\in\{0,1\}^{n+1}\setminus\{00\ldots 0\}$.
\end{proof}

\begin{proposition}\label{prop:circuits-in-ZX}
The flow-preserving rewrite rules imply the ZX-translations of all equations of Figure~\ref{fig:qcsf-axioms}, each taken in the trivial context (for now).
\end{proposition}
\begin{proof}
We consider the equations in order.
When applying rules that preserve extended causal flow only under certain conditions, it is straightforward to check that those conditions are satisfied within the subdiagram being rewritten.

For \eqref{initP-ancilla}, we have:
\[
\begin{tikzpicture}
	\begin{pgfonlayer}{nodelayer}
		\node [style=gate] (0) at (0, 0) {$P(\varphi)$};
		\node [style=ancilla] (1) at (-1.25, 0) {};
		\node [style=none] (2) at (1.25, 0) {};
	\end{pgfonlayer}
	\begin{pgfonlayer}{edgelayer}
		\draw (1) to (0);
		\draw (0) to (2.center);
	\end{pgfonlayer}
\end{tikzpicture}

\quad\mapsto\quad \input{zx-circ-deriv/AP0.tikz}
\quad\overset{\ref{lem:graph-like-copy}}{\to}\quad \begin{tikzpicture}
	\begin{pgfonlayer}{nodelayer}
		\node [style=Z dot] (1) at (1.25, 0) {};
		\node [style=Z dot] (2) at (2.25, 0) {};
		\node [style=none] (4) at (2.75, 0) {};
	\end{pgfonlayer}
	\begin{pgfonlayer}{edgelayer}
		\draw (2) to (4.center);
		\draw [style=hadamard edge] (1) to (2);
	\end{pgfonlayer}
\end{tikzpicture}

\quad\mapsfrom\quad \begin{tikzpicture}
	\begin{pgfonlayer}{nodelayer}
		\node [style=ancilla] (1) at (-0.5, 0) {};
		\node [style=none] (2) at (0.5, 0) {};
	\end{pgfonlayer}
	\begin{pgfonlayer}{edgelayer}
		\draw (1) to (2.center);
	\end{pgfonlayer}
\end{tikzpicture}

\]
For \eqref{HH-ugp}, we have:
\[
\begin{tikzpicture}
	\begin{pgfonlayer}{nodelayer}
		\node [style=none] (4) at (-1.25, 0) {};
		\node [style=none] (5) at (1.5, 0) {};
		\node [style=gate] (6) at (-0.5, 0) {$H$};
		\node [style=gate] (7) at (0.75, 0) {$H$};
	\end{pgfonlayer}
	\begin{pgfonlayer}{edgelayer}
		\draw (4.center) to (5.center);
	\end{pgfonlayer}
\end{tikzpicture}

\quad\mapsto\quad \input{zx-circ-deriv/H2.tikz}
\quad\overset{\eqref{eq:IO}}{\rightleftarrows}\quad 
\quad\mapsfrom\quad \begin{tikzpicture}
	\begin{pgfonlayer}{nodelayer}
		\node [style=none] (0) at (-0.55, 0) {};
		\node [style=none] (1) at (0.55, 0) {};
	\end{pgfonlayer}
	\begin{pgfonlayer}{edgelayer}
		\draw (0.center) to (1.center);
	\end{pgfonlayer}
\end{tikzpicture}

\]
The flow-preserving ZX-calculus derivation of \eqref{P0-ugp} actually looks identical to that of \eqref{HH-ugp}:
\[
\begin{tikzpicture}
	\begin{pgfonlayer}{nodelayer}
		\node [style=none] (4) at (-1.25, 0) {};
		\node [style=none] (5) at (1.25, 0) {};
		\node [style=gate] (6) at (0, 0) {$P(0)$};
	\end{pgfonlayer}
	\begin{pgfonlayer}{edgelayer}
		\draw (4.center) to (5.center);
	\end{pgfonlayer}
\end{tikzpicture}

\quad\mapsto\quad \input{zx-circ-deriv/H2.tikz}
\quad\overset{\eqref{eq:IO}}{\rightleftarrows}\quad 
\quad\mapsfrom\quad 
\]
For \eqref{initCNOT-ancilla}, using asterisks to mark the pair of vertices to which the pivot-and-delete rule of Lemma~\ref{lem:pivot-delete} is applied in the next step, we have:
\begin{align*}
	\input{./qcancilla-axioms/initCNOT.tikz}
	\quad\mapsto\quad \input{zx-circ-deriv/ACX0.tikz}
	\quad\overset{\eqref{eq:IO}}&{\to}\quad \input{zx-circ-deriv/ACX1.tikz} \\
	\quad\overset{\ref{lem:pivot-delete}}&{\to}\quad \input{zx-circ-deriv/ACX2.tikz}
	\quad\overset{\eqref{eq:IO}}{\to}\quad \input{zx-circ-deriv/ACX3.tikz}
	\quad&\mapsfrom\quad \input{./qcancilla-axioms/initIdId.tikz}
\end{align*}
For \eqref{bigebre-ugp}, again using asterisks to mark vertices that are about to be deleted after a pivot, we have:
\begin{align*}
	\input{./qc-axioms/CNOTNOTC.tikz}
	\quad\mapsto\quad \input{zx-circ-deriv/B0.tikz}
	\quad\overset{\eqref{eq:IO}}&{\to}\quad \input{zx-circ-deriv/B1.tikz} \\
	\quad\overset{\ref{lem:pivot-delete}}&{\to}\quad \input{zx-circ-deriv/B2.tikz} \\
	\quad\overset{\eqref{eq:IO}}&{\to}\quad \input{zx-circ-deriv/B3.tikz}
	&\mapsfrom\quad \input{./qc-axioms/SWAPCNOT.tikz}
\end{align*}
For \eqref{CNOTPCNOT-ugp}, we have:
\begin{align*}
	\input{./qc-axioms/CNOTPphiCNOT.tikz}
	\quad\mapsto\quad \input{zx-circ-deriv/C0.tikz} \quad
	\overset{\ref{lem:pivot-delete}}&{\rightleftarrows}\quad \input{zx-circ-deriv/C1.tikz} \\
	\overset{\eqref{eq:IO}}&{\rightleftarrows}\quad \input{zx-circ-deriv/C2.tikz} \\
	\overset{\ref{lem:pivot-delete}}&{\rightleftarrows}\quad \input{zx-circ-deriv/C3.tikz} \\
	\overset{\eqref{eq:IO}}&{\rightleftarrows}\quad \input{zx-circ-deriv/C4.tikz}
	&\mapsfrom\quad \input{./qc-axioms/PphiId.tikz}
\end{align*}	
For \eqref{CZ-ugp}, we will prove the result in the opposite direction for simplicity.
Recall that whenever a vertex can be deleted in a flow-preserving way, the reverse vertex insertion is also flow-preserving.
We now also use asterisks to mark single vertices to which a local-complementation-and-delete operation as in Lemma~\ref{lem:lc-delete} is about to be applied.
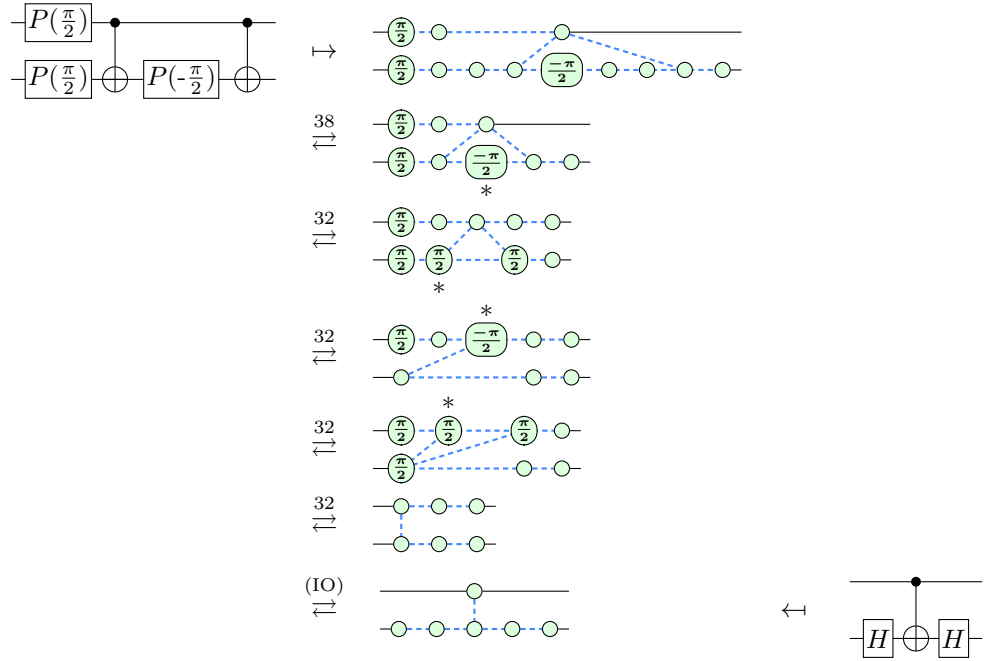
\begin{align*}
	\input{./qc-axioms/CZ.tikz} \quad
	&\mapsto\quad \input{zx-circ-deriv/CZ0.tikz} \\
	\overset{\ref{prop:vertex-splitting}}&{\rightleftarrows}\quad \input{zx-circ-deriv/CZ1.tikz} \\
	\overset{\ref{lem:lc-delete}}&{\rightleftarrows}\quad \input{zx-circ-deriv/CZ2.tikz} \\
	\overset{\ref{lem:lc-delete}}&{\rightleftarrows}\quad \input{zx-circ-deriv/CZ3.tikz} \\
	\overset{\ref{lem:lc-delete}}&{\rightleftarrows}\quad \input{zx-circ-deriv/CZ4.tikz} \\
	\overset{\ref{lem:lc-delete}}&{\rightleftarrows}\quad \input{zx-circ-deriv/CZ5.tikz} \\
	\overset{\eqref{eq:IO}}&{\rightleftarrows}\quad \input{zx-circ-deriv/CZ6.tikz}
	&\mapsfrom\quad \input{./qc-axioms/H2CNOTH2.tikz}
\end{align*}
Equation \eqref{ctrl2pi-ugp} is proved in Lemma~\ref{lem:zx-ctrl2pi-ugp}.
Equation \eqref{Paddition-prime} is a special case of `vertex splitting':
\begin{align*}
	\begin{tikzpicture}
	\begin{pgfonlayer}{nodelayer}
		\node [style=none] (0) at (-3, 0) {};
		\node [style=none] (1) at (2.5, 0) {};
		\node [style=gate] (4) at (-1.5, 0) {$P(\varphi_1)$};
		\node [style=gate] (5) at (1, 0) {$P(\varphi_2)$};
	\end{pgfonlayer}
	\begin{pgfonlayer}{edgelayer}
		\draw (0.center) to (1.center);
	\end{pgfonlayer}
\end{tikzpicture}
 \quad
	&\mapsto\quad \input{zx-circ-deriv/Pplus0.tikz} \\
	\quad\overset{\ref{prop:vertex-splitting}}&{\rightleftarrows}\quad \input{zx-circ-deriv/Pplus1.tikz}
	\quad&\mapsfrom\quad \begin{tikzpicture}
	\begin{pgfonlayer}{nodelayer}
		\node [style=none] (0) at (-2.25, 0) {};
		\node [style=none] (1) at (2.25, 0) {};
		\node [style=gate] (4) at (0, 0) {$P(\varphi_1{+}\varphi_2)$};
	\end{pgfonlayer}
	\begin{pgfonlayer}{edgelayer}
		\draw (0.center) to (1.center);
	\end{pgfonlayer}
\end{tikzpicture}

\end{align*}
Finally, for \eqref{euler-ugp}:
\begin{align*}
	\begin{tikzpicture}
	\begin{pgfonlayer}{nodelayer}
		\node [style=none] (4) at (-3.5, 0) {};
		\node [style=none] (5) at (3.5, 0) {};
		\node [style=gate] (14) at (-2, 0) {$P(\alpha_1)$};
		\node [style=gate] (15) at (0, 0) {$H$};
		\node [style=gate] (16) at (2, 0) {$P(\alpha_2)$};
	\end{pgfonlayer}
	\begin{pgfonlayer}{edgelayer}
		\draw (4.center) to (5.center);
	\end{pgfonlayer}
\end{tikzpicture}

	&\mapsto\quad \input{zx-circ-deriv/Euler0.tikz} \\
	\overset{\ref{prop:vertex-splitting}}&{\to}\quad \input{zx-circ-deriv/Euler1.tikz} \\
	\overset{\eqref{eq:EU-ZX}}&{\to}\quad \input{zx-circ-deriv/Euler2.tikz} \\
	\overset{\ref{prop:vertex-splitting}}&{\to}\quad \input{zx-circ-deriv/Euler3.tikz} \\
	&\mapsfrom\quad \begin{tikzpicture}
	\begin{pgfonlayer}{nodelayer}
		\node [style=none] (4) at (-6.75, 0) {};
		\node [style=none] (5) at (6.75, 0) {};
		\node [style=gate] (15) at (-4, 0) {$P(\beta_1)$};
		\node [style=gate] (17) at (0, 0) {$P(\beta_2)$};
		\node [style=gate] (18) at (4, 0) {$P(\beta_3)$};
		\node [style=gate] (19) at (-2, 0) {$H$};
		\node [style=gate] (20) at (2, 0) {$H$};
		\node [style=gate] (21) at (-6, 0) {$H$};
		\node [style=gate] (23) at (6, 0) {$H$};
	\end{pgfonlayer}
	\begin{pgfonlayer}{edgelayer}
		\draw (4.center) to (5.center);
	\end{pgfonlayer}
\end{tikzpicture}

\end{align*}
This completes the proof of every equation in Figure~\ref{fig:qcsf-axioms}.
\end{proof}

We can now prove our main result.

\begin{theorem}
The rules of Figure~\ref{fig:flow-preserving} are complete for flow-preserving rewriting.
\end{theorem}
\begin{proof}
Suppose $D_1$ and $D_2$ are two MBQC-form ZX-diagrams which both have Pauli flow and represent the same linear embedding.
Using the procedure of Section~\ref{s:circuit-extraction}, we can use the rules of Figure~\ref{fig:flow-preserving} to transform $D_1$ in a flow-preserving way into a diagram $D_1'$ which has extended causal flow.
Similarly, we can transform $D_2$ into a diagram $D_2'$ which also has extended causal flow.

Any ZX-diagram with extended causal flow corresponds to a circuit over the generators from the beginning of Section~\ref{s:circuits} together with multi-controlled phase gates to express any phase gadgets.
By Theorem~\ref{proof:THMcompleteQC}, the rules of Figure~\ref{fig:qcsf-axioms} (plus the definition of multi-controlled phase gates in \eqref{mctrlPinducdef}) are complete for such circuits.
Proposition~\ref{prop:circuits-in-ZX} shows that the ZX-translations of each of those rules can be derived from the rules of Figure~\ref{fig:flow-preserving}.
The contexts in which circuit rewrite rules are valid to apply translate in ZX to the condition of Definition~\ref{def:causal-subdiagram}.
Thus any sequence of circuit rewrite steps translates to a sequence of valid flow-preserving ZX-calculus rewrites.
In other words, we can use the rule of Figure~\ref{fig:flow-preserving} to transform $D_1'$ into $D_2'$. 
\end{proof}

\phantomsection
\addcontentsline{toc}{section}{References}
\bibliography{MBQC}

\end{document}